\documentclass[11pt]{elsarticle}
\usepackage{amsmath,amssymb,amsthm}
\usepackage{mathtools}
\usepackage{cases}
\usepackage{hyperref}
\usepackage{mathrsfs}%$\mathscr$
\usepackage{graphicx}
\usepackage{tikz}
\usepackage[all,pdf]{xy}
\usepackage{enumerate}
\usepackage{float}
\usepackage{caption}
\usepackage[all]{xy}
\theoremstyle{plain}
\newtheorem{theorem}{Theorem}[section]
\newtheorem{lemma}[theorem]{Lemma}
\newtheorem{proposition}[theorem]{Proposition}
\theoremstyle{definition}
\newtheorem{definition}[theorem]{Definition}

\newtheorem{remark}[theorem]{Remark}
\newtheorem{fact}[theorem]{Fact}
\newtheorem{corollary}[theorem]{Corollary}
\newtheorem{example}[theorem]{Example}

\newcommand{\tr}[1]{#1}
\newcommand{\supp}{\mathrm{supp}\,}
\newcommand{\dda}{\mathord{\mbox{\makebox[0pt][l]{\raisebox{-.4ex}{$\downarrow$}}$\downarrow$}}}% 下 waybelow
\newcommand{\dua}{\mathord{\mbox{\makebox[0pt][l]{\raisebox{.4ex}{$\uparrow$}}$\uparrow$}}}

\newcommand{\da}{\downarrow\!\!}
\newcommand{\ra}{\rightarrow\!}
\newcommand{\ua}{\uparrow\!\!}%

\newcommand{\casef}{\overset{\triangleright}{\longrightarrow}}

\makeatletter
\def\ps@pprintTitle{%
  \let\@oddhead\@empty
  \let\@evenhead\@empty
  \def\@oddfoot{\reset@font\hfil\thepage\hfil}
  \let\@evenfoot\@oddfoot
}\makeatother

\begin{document}
\begin{frontmatter}

\title{Monad Structures on Topological Spaces Comprising Mislove's Random Variables}
\tnotetext[t1]{This work is supported by the National Natural Science Foundation of China (Grant No. 12231007).}

\author[1]{Chengyu Zhou}
\ead{zhouchy126@126.com}

\author[1]{Qingguo Li\corref{cor1}}
\cortext[cor1]{Corresponding author.}
\address[1]{School of Mathematics, Hunan University, Changsha, Hunan, 410082, China}
\ead{liqingguoli@aliyun.com}

\begin{abstract}
  Mislove, Goubault and Varacca investigated how to define random variables in Domain theory to form monads over the category of bounded complete domains.
   They intended to model probabilistic programming languages with their random variables monads.
   In this paper, we focus on the random variables defined by Mislove from a topological perspective.
   We provide a topology for $\surd$-max continuous random variables on a $T_0$ space, we construct a new $T_0$ space, where $\surd$-max property is essential for the monad structures.
   We show that the spaces of normalized $\surd$-max simple random variables form a monad over the category of $T_0$ spaces and that the spaces of normalized $\surd$-max continuous random variables give a monad over the category of d-spaces.
   In addition, on a sober space, the space of normalized $\surd$-max continuous random variables is the sobrification of the space of normalized $\surd$-max simple random variables.
\end{abstract}
\begin{keyword}
 Domain \sep $T_0$ space \sep Random variable \sep Monad

 \vspace*{0.6cm}

\end{keyword}

\end{frontmatter}
\section{Introduction}
%In Domain theory, scholars mainly research valuations for modeling probability distributions.
% Unlike Borel measures, valuations only take values on open subsets, rather than on all Borel sets.
% Under certain conditions, valuations and Borel measures can be converted into one another \cite{Keimel2005}.
% Numerous results about valuations have emerged over the past forty years.
% C. Jones  \cite{Jones1990} proved that the continuous valuations on a domain $L$ either is a domain and gives the free dcpo cones over $L$.
% R. Heckmann \cite{Heckmann1996} showed that on a $T_0$ space, the space of point continuous valuations with the weak topology is the sobrification of the space simple valuations.
% J. Goubault and X. Jia \cite{Goubault2019,Goubault2024} studied the algebras of valuation power space monad.
% X. Jia, M. Mislove and etc. \cite{Jia2022,Jia2021} investigated the $K$-ification of the space of simple valuations.

%In the first decade of the 20th century, some scholars began to model probability distributions using random variables.
In the first half of the 20th century, some scholars studied random variables in sample spaces.
 Recall that a random variable $f:(X,S_X,\mu)\ra(Y,S_Y)$ on the measurable space $(Y,S_Y)$ is a measurable map from a probability space $(X,S_X,\mu)$ to $(Y,S_Y)$, where $S_X,S_Y$ are $\sigma$-algebra on $X,Y$ and $\mu:S_X\ra[0,1]$ is a probability measure on $X$.

In the first decade of the 21th century, some scholars began to research random variables in Domain theory.
 The first glance at random variables in Domain theory is the paper \cite{Mislove2005}.
 Mislove and Varacca regarded a partial map $f:\mathbb{R}_+\ra L$ from non-negative real numbers $\mathbb{R}_+$ to a domain $L$, whose $dom(f)$ is finite, as a finite random variable on $L$.
 The authors provided an abstract basis by the finite random variables endowed with an approximate relation.
 Hence, they model the discrete random variable on a domain by the rounded ideal completion of the abstract basis.

In \cite{Goubault2011}, Goubault and Varacca discussed their continuous random variables on bounded complete domains.
 A valuation on Cantor tree $\mathbb{C}$ is called thin if the mass of the valuation concentrates on the maximal elements of the support of the valuation.
 On a bounded complete domain $L$, Goubault and Varacca regarded a pair $(\mu,f)$ as a continuous random variable on $L$, where $\mu$ is a normalized thin valuation on the Cantor tree $\mathbb{C}$ and $f$ is a partially continuous map from the support of $\mu$ to $L$.
 Then, with a given order, the continuous random variables on $L$ form a bounded complete domain.
 However, the continuous random variables defined by them failed to form a monad \cite{Mislove2013,Mislove2014}.

In Barker's Ph.D dissertation \cite{Barker2016}, he investigated the partially continuous maps from some anti-chains of Cantor tree $\mathbb{C}$ to bounded complete domains, and showed that these partial maps with a given order form a monad over the category of bounded complete domains.
 But the results only concern about non-determinism.

Mislove \cite{Mislove2017} modified Cantor tree.
 He added some maximal finite elements to Cantor tree $\mathbb{C}$, and say that a continuous random variable on a bounded complete domain $L$ is a pair $(\mu,f)$, where $\mu$ is a valuation on Mislove's Cantor tree $\mathbb{M}$ and $f$ a partially continuous map from the support of $\mu$ to $L$.
 One main goal of \cite{Mislove2017} is to provide a monad of continuous random variables over the category of bounded complete domains.

Gianantonio and Edalat \cite{Pietro2024} introduced their PER-domains.
 A PER-domain is a countably based bounded complete domain equipped with a partial equivalence relation that the relation is closed under the sups of directed subsets.
 The PER-domains form a cartesian closed category with given morphisms.
 The authors constructed a monad of their random variables over the category of PER-domains, in which a Scott continuous map from their fixed probability spaces, including Cantor space with the standard uniform product measure to a PER-domain is recognized as a continuous random variable on the PER-domain.

The series of these works establish models of stochastic processes and probabilistic automatons (see \cite{Mislove2013} for a motivating example, see \cite{Mislove2020} for a deep connection with stochastic process theory).

In Domain theory, there is interest in developing probabilistic models on non-Hausdorff spaces \cite{Heckmann1996,Goubault2019,Goubault2024,Jia2022,Jia2021}.
% R. Heckmann \cite{Heckmann1996} showed that on a $T_0$ space, the space of point continuous valuations with the weak topology is the sobrification of the space simple valuations.
% J. Goubault and X. Jia \cite{Goubault2019,Goubault2024} studied the algebras of valuation power space monad.
% X. Jia, M. Mislove and etc. \cite{Jia2022,Jia2021} investigated the $K$-ification of the space of simple valuations.
 Especially, $T_0$ spaces, d-spaces and sober spaces are welcomed, they are considered as the extensions of posets, dcpos and domains.
 The goal of this paper is to find a non-Hausdorff topological structure for random variables in Domain theory, thereby to form monad structures.
 In line with Mislove's definition, a continuous random variable on a $T_0$ space $X$ is a pair $(\mu,f)$ consisting of a continuous valuation $\mu$ on Mislove's Cantor tree $\mathbb{M}$ and a continuous map from the support of $\mu$ to the space $X$.
 To avoid unintended consequences, we additionally require that continuous random variables are $\surd$-max.
 We find a topology for $\surd$-max continuous random variables on $X$ such that the resulting space is $T_0$.
 The topology is generated by a subbase that each member of the subbase has the form of $\langle\ua z,r,V\rangle$, where the first two parameters indicate the weak open subset $\langle\ua z>r\rangle$ of $\mathcal{V}_{\leq 1}\Sigma\mathbb{M}$, and the third parameter indicates that the partial maps of random variables in this subset have to map $z$ into an open subset $V\in\mathcal{O}X$.
 One also defines simple random variables, which are continuous random variables whose valuations are simple valuations with finite support.
 We show that the spaces of normalized $\surd$-max simple random variable give a monad over $T_0$ spaces in Section \ref{SecMonad}.
 We verify that \tr{every normalized $\surd$-max continuous random variable is approximated by normalized $\surd$-max simple random variables, and that} the space of normalized $\surd$-max continuous random variables on a $T_0$ space is the sobrification of the space of normalized $\surd$-max simple random variables in Section \ref{SecSob}, and that the spaces of normalized $\surd$-max continuous random variables form a monad over d-spaces in Section \ref{SecMonad2}.
 \tr{Finally we reveal that our monads are not commutative in Section \ref{SecStrength}.}

%\tr{We mainly compare our results with \cite{Pietro2024}.
%%First of all, in essences, we treat our random variables as a
%Our results are more abundant, the construction in \cite{Pietro2024} requires the sample space to be restricted to only four cases, and our sample spaces can be Misloves' Cantor tree with any continuous valuations removed from the Cantor tree; in addition, we consider bounded complete domain, and various spaces as the probability space, \cite{Pietro2024} need
%In addition, we develop the simple random variables and show that every continuous random variable is approximated by simple random variables, which is fit the Domain-theoretical model.
%The essential difference is attitudes for random variables, Pietro treat them like their push forward -- valuation on the probability space given by , but we treat them like a stochastic process manipulated by a battery-operated tossing machine.}

\tr{Within the existing literature, few studies have investigated monads constructed based on random variable.
 Accordingly, we mainly compare our work with \cite{Pietro2024}.
 Our research deviates from \cite{Pietro2024} in both the core understanding of random variables and targets.
 The work of \cite{Pietro2024} dedicates to provide a solution to the famous Jung-Tix problem, where they reagrd random variables as corresponding push forwards.
 In contrast, our study targets the resolution of Mislove’s question in \cite{Mislove2013}, and interprets random variables as processes manipulated by a coin-flipping machine.
 Further distinctions lie in the settings of the sample space and probability space.
 In \cite{Pietro2024}, they fixed four-point sample spaces, and their probability spaces equip with countable bases.
 Differing from \cite{Pietro2024}, our research adopts Mislove’s Cantor tree as the sample space, whose associated valuations cover all measures defined on the Cantor space, and our probability space is not limited to a countable base.
 Besides, the work in \cite{Pietro2024} only considers continuous random variables.
 By comparison, our definition and construction of random variables contain simple random variables, which approximates continuous random variables.}
% This inclusive design enables our framework to accommodate both discrete and continuous scenarios, achieving a more comprehensive and rigorous characterization of random variables compared to the limited continuous-only framework in \cite{Pietro2024}.}

\tr{
Contrasted with those classical investigations on valuations over $T_0$ and sober spaces, whose core aim coincides with that of \cite{Pietro2024}, namely modelling distributions on spaces, our goal in this paper differs.}

\section{Preliminaries}

The standard domain theory is introduced in \cite{Gierz2003,Goubault2013}.

\subsection{Domains}

Let $P$ be a poset.
 For every subset $A$, $\da A$ denotes $\{y:\exists x\in A,y\leq x\}$ and $\ua A$ indicates $\{y:\exists x\in A,y\geq x\}$.
 In general, $\da\{x\},\ua\{x\}$ are abbreviated as $\da x,\ua x$.
 A lower subset $A$ (resp., an upper subset $A$) is a subset of $P$ with $A=\da A$ (resp., $A=\ua A$).
 The element $y\in P$ is an upper bound of a subset $A$ (resp., a lower bound of $A$) if \tr{$\forall x\in A, x\leq y$ (resp., $\forall x\in A, y\leq x$)}.
 The subset of all upper bounds of $A\subseteq P$ is written as $A^u$.%; \tr{likewise, $A^l$ stands for the subset consisting of all lower bounds of $A$.}
 A directed subset $D\subseteq^\ua P$ is a non-empty subset of $P$ that every finite subset of $D$ has an upper bound in $D$.
 \tr{A subset of a poset is a filtered if it is directed in the dual order.}
 The least upper bound of a subset $A$ (also called the supremum of $A$) is denoted as \tr{$\sup(A)$}; especially, the least upper bound of a directed subset $D$ is denoted by $\bigsqcup D$.
 \tr{The greatest lower bound of a subset $A$ (the so-called infimum of $A$) is written as $\inf (A)$}
 The poset $P$ is called bounded complete if $A^u\neq\emptyset$ implies \tr{$\sup(A)$} exists for every $A\subseteq P$.
 \tr{And it is well-known that a poset is bounded complete if and only if every non-empty subset of it has an inf (since $\sup(F)=\inf (F^u)$).}
 We say that $P$ is a dcpo (directed complete poest) if every directed subset of $P$ has a least upper bound.
 For every $x,y\in P$, we write $x\ll y$ and say that $x$ is way-below $y$ if
 \begin{equation*}
   \forall D\subseteq^\ua P,\bigsqcup D\geq y\Rightarrow \exists d\in D,d\geq x.
 \end{equation*}
 The notation $\dda y$ denotes the subset $\{x\in P:x\ll y\}$, and $\dua y$ denotes $\{x\in P:y\ll x\}$.
 \tr{A dcpo $Q$ is a domain if for every $y\in Q$, $\dda y$ is a directed subset and $\bigsqcup\dda y=y$.}
 A basis $B$ of a \tr{poset} $Q$ is a subset such that for every $y\in Q$, $\dda y\cap B$ is directed and $\bigsqcup\dda y\cap B=y$.
 A dcpo is a domain iff it has a basis.
 An element $x\in P$ is said compact if $x\ll x$.
 The set of all compact elements of $P$ is denoted by $\mathcal{K}(P)$.
 An algebraic domain $Q'$ is a dcpo such that $\mathcal{K}(Q')$ is a basis.
 A subset $U\subseteq P$ is a Scott open subset if for every directed subset $U$, \tr{$\bigsqcup D\in U$} implies $U\cap D\neq\emptyset$.
 All of Scott open subsets form a topology on $P$, which is called the Scott topology.
 A poset $P$ equipped with its Scott topology is denoted by $\Sigma P$.
 A Scott closed subset of a domain $Q$ with the relative order is a domain; moreover, if $Q$ is a bounded complete domain, so is the Scott closed subset.
 For each $y$ of a domain $Q$, $\dua y$ is Scott open and $\{\dua y:y\in Q\}$ is a base $\Sigma Q$; especially, if $B$ is a basis of $Q$, then $\{\dua y:y\in B\}$ is a base of $\Sigma Q$.
 A map $f:P\ra P'$ between the posets is order-preserving (equivalently, $f$ preserves the order) if $x\leq y$ implies $f(x)\leq f(y)$.
 A Scott continuous map $f':P\ra P'$ between the posets is a map preserving the sups of directed subsets.
 Obviously, every Scott continuous map is order-preserving.
 A map $g:P\ra P'$ between the posets is Scott continuous iff $g$ is continuous with respect to the Scott topology.

\subsection{Spaces}

For a topological space $X$, $\mathcal{O}X$ denotes the open subsets with the inclusion order, and $\Gamma X$ denotes the closed subsets with the inclusion order.
 A closure of a subset $A$ of $X$ is denoted as $\overline{A}$.
 The specialization order $\leq$ on a space $X$ is given by $x\leq y$ if
\begin{equation*}
  \forall U\in\mathcal{O}X,x\in U\Rightarrow y\in U.
\end{equation*}
 For every poset $P$, the specialization order generated by $P$ with the Scott topology agrees with the order of $P$.
 An upper subspace (resp., a lower subspace) $A$ of $X$ is an upper subset (resp., a lower subset) $A$ with respect to the specialization order of $X$ endowed with the relative topology.
 By the Scott topology on a space, we mean the one generated by the specialization order.
 It is easy to see that for an $x\in X$ of a space $X$, $\da x$ is equal to $\tr{\overline{\{x\}}}$.

A d-space $X$ is a $T_0$ space such that $X$ with the specialization order forms a dcpo and every open subset $U\in\mathcal{O}X$ is Scott open.
 Obviously, every directed subset of \tr{a d-space} converges to its sup.
\begin{definition}
  For a topological space $X$, a subset $F\subseteq X$ is said irreducible if for every $U,V\in\mathcal{O}X$, $U\cap F\neq\emptyset$ and $V\cap F\neq\emptyset$ imply $U\cap V\cap F\neq\emptyset$; and we call a subset irreducible closed subset if it is both irreducible and closed.
   A $T_0$ space $Y$ is called sober if every irreducible closed subset is of the form \tr{$\overline{\{y\}}$} for some $y\in Y$.
\end{definition}
 It is straightforward to verify that the closure of an irreducible subset is \tr{an irreducible closed subset}.
 The following remark says that a $T_0$ space is sober if and only if every irreducible subset has a sup and converges to its sup.
 This result is well-known and trivial, but few papers directly stated it.
  \begin{remark}\label{SoberEquiva}
    \tr{Suppose that $Y$ is sober and $F\subseteq Y$ is irreducible.
     Then we have $\overline{F}=\overline{\{y\}}=\da y$ for some $y\in Y$.
     If $x\in F^u$, then $F\subseteq \da x=\overline{\{x\}}$.
     Hence $\overline{F}=\da y\subseteq\da x$.
     It follows that $y$ is the sup of $F$, and obviously $F$ converges to $y$.
     Conversely, Assume that $X$ is a space such that every irreducible subset of it has a sup and converges to the sup.
     So every irreducible closed subset has a sup, and it is easy to see the irreducible closed subset is the closure of the sup.}
  \end{remark}
 \tr{In addition}, every domain endowed with its Scott topology is sober \cite[Corollary II-1.13]{Gierz2003}, and every sober space is a d-space \tr{since directed subsets are irreducible}.

\begin{definition}
  Let $f:X\rightharpoonup Y$ be a partial map between the topological spaces.
   We say that $f$ is partially continuous if for every $V\in\mathcal{O}Y$, $f^{-1}(V)$ is a relative open subset of $dom(f)$.
\end{definition}
 Equivalently, $f:X\rightharpoonup Y$ is partially continuous if $f:dom(f)\ra Y$ is continuous.
 \begin{remark}
   A partially continuous map $f:X\rightharpoonup Y$ preserves the specialization order of $dom(f)$, and $dom(f)$ ordered by its specialization order is a subposet of $X$.
 \end{remark}
 We denote $\{f(x):x\in S\cap dom(f)\}$ by $f(S)$ for every subset $S\subseteq X$.
 \begin{definition}
   The poset of all partially continuous maps between the topological spaces $X,Y$ is denoted as $[X\rightharpoonup Y]$, \tr{and it is ordered by $f\leq g$ if $dom(f)\subseteq dom(g)$ and $f(x)\leq g(x)$ for every $x\in dom(f)$.}
 \end{definition}

\subsection{\tr{Continuous} valuations}

\tr{Let $[0,1]$ be the unite interval with the usual order.
 A continuous valuation on a topological space $X$ is a Scott continuous map $\mu:\mathcal{O}X\ra [0,1]$ satisfying:}
\begin{enumerate}[(i)]
  \item $\mu(\emptyset)=0$;
  \item if $U\subseteq V$ in $\mathcal{O}X$, then $\mu(U)\leq\mu(V)$;
  \item $\mu(U)+\mu(V)=\mu(U\cap V)+\mu(U\cup V)$.
\end{enumerate}
 \tr{A continuous valuation $\mu$ is normalized if $\mu(X)=1$.}
 The set of continuous valuation endowed with the point wise order (that is, $\mu\leq\upsilon$ iff $\mu(U)\leq\upsilon(U)$ for all $U\in\mathcal{O}X$) is denoted by $\mathcal{V}_{\leq 1}X$.
 The Dirac valuation $\delta_x$ of an $x\in X$ is given by
 \begin{equation*}
   \delta_x(U)=\begin{cases}
                 1, & \mbox{if $x\in U$}; \\
                 0, & \mbox{otherwise}.
               \end{cases}
 \end{equation*}
 A simple valuation $\Sigma_{x\in F}r_x\delta_x$ is a finitely \tr{subconvex combination (that is, $\Sigma_{x\in F}r_x\leq 1$)} of some Dirac valuations.
 Obviously, every Dirac valuation is a simple one, and every simple valuation is continuous \tr{\cite[Remark under Definition IV-9.9]{Gierz2003}}.
% For a simple valuation $\Sigma_{x\in F}r_x\delta_x$, we say that $r_x$ is the mass of $x$.
 A finitely valued valuation $\mu$ is a continuous valuation such that $\{\mu(U):U\in\mathcal{O}X\}$ is a finite set.

 \begin{proposition}(\cite[\tr{Proposition 2.2}]{Jia2022})\label{FiniteAreSimple}
   For a $T_0$ space $X$, finitely valued valuations are exactly the simple valuations iff $X$ is sober.
 \end{proposition}

For a space $X$, we denote by $\mathcal{V}_1X$ the subposet of normalized continuous valuations.

\begin{definition}
  \tr{For every space $X$, the weak topology on $\mathcal{V}_{\leq 1}X$ is generated by the subbase $\{\langle U>r\rangle:U\in\mathcal{O}X,r\geq 0\}$, where}
\begin{equation*}
  \tr{\langle U>r\rangle:=\{\mu\in\mathcal{V}X:\mu(U)>r\}.}
\end{equation*}
\end{definition}
% Obviously, $\mathcal{V}_{\leq 1}X$ is a weak closed subset of $\mathcal{V}X$.

\begin{theorem}\label{BigConsequenceofV}
  \tr{For every domain $L$, the following statements hold:
  \begin{enumerate}[(i)]
    \item $\mathcal{V}_{\leq 1}\Sigma L$ is a domain \cite[Corollary 5.4]{Jones1990};
    \item the set of simple valuations on $X$ forms a basis of $\mathcal{V}_{\leq 1}\Sigma L$, where $\Sigma_{x\in F}r_x\delta_x\ll \mu$ iff for every $K\subseteq F$, $\mu(\bigcup_{x\in K}\dua x)>\Sigma_{x\in K}r_x$;\label{WBRelationOfVL}
    \item the weak topology agrees with the Scott topology on $\mathcal{V}_{\leq 1}\Sigma L$;\label{WeakAgreeScott}
    \item if $L$ is a bounded complete domain, then every filtered subset of $\mathcal{V}_{\leq 1}\Sigma L$ has an inf. \label{VC*FiteredHasInf}
  \end{enumerate}}
\end{theorem}
\begin{proof}
  \tr{Since $\mathcal{V}_{\leq 1}\Sigma L$ is a Scott closed subset of the extended probabilistic powerdomain over $L$ introduced in \cite[Section IV-9]{Gierz2003}, (ii) follows from \cite[Theorem IV-9.16]{Gierz2003}.}

  \tr{Every Scott open neighborhood $U$ of an $x\in L$ contains a $y$ such that $x\in\dua y\subseteq U$.
   Hence, $\Sigma L$ is a c-space satisfying the assumption of \cite[Theorem 3.6]{Lyu2018}, and (iii) directly follows from the cited theorem.}

  \tr{Likewise, if $L$ is bounded complete, then $L$ satisfies requirements of \cite[Corollary 3.4]{Lyu2018}; thus $\mathcal{V}_{\leq 1}\Sigma L$ is coherent, which induced that $\mathcal{V}_{\leq 1}\Sigma L$ is Lawson compact by \cite[Theorem 4.2 (2)]{Xi2017}, where well-filteredness is provided by \cite[Proposition 8.3.5]{Goubault2013} and compactness is provided by the constant zero valuation $\underline{0}:\mathcal{O}\Sigma L\ra [0,1]$.
  Finally, filtered subsets of $\mathcal{V}_{\leq 1}\Sigma L$ have infs by \cite[Lemma III-5.2]{Gierz2003}.}
\end{proof}

\tr{
%The structure $\mathcal{V}_{\leq 1}\Sigma$ forms a monad over the category of dcpos and Scott continuous maps.
%In this paper, we denote $f[-]$ instead of $\mathcal{V}$
For every Scott continuous $f:L\ra M$ between the dcpos, $f$ gives a Scott continuous map $f[-]:\mathcal{V}_{\leq 1}\Sigma L\ra\mathcal{V}_{\leq1}\Sigma M$ by $f[\mu](U)=\mu(f^{-1}(U))$.}
 \begin{remark}\label{p[-]properties}
 \tr{$\mathcal{V}_{\leq 1}\Sigma$ forms a monad over the categories of dcpos and Scott continuous maps \cite[Theorem 4.5]{Jones1990}.
  In \cite{Jones1990}, $f[-]$ was introduced as $\mathcal{V}_{\leq 1}(f)$.
  So there are some equalities induced from monadic properties: $(g\circ f)[-]=g[f[-]]$, $id_L[-]=id_{\mathcal{V}_{\leq 1}\Sigma L}$ and $f[\delta_x]=\delta_{f(x)}$.
  Furthermore, $f[-]$ preserves subconvex combinations, since for every open subset $U\in\mathcal{O}\Sigma L$ and a subconvex combination $\Sigma_{i\in I}r_i\mu_i$,
  $$f[\Sigma_{i\in I}r_i\mu_i](U)=\Sigma_{i\in I}r_i\mu_i(f^{-1}(U))=\Sigma_{i\in I}r_if[\mu_i](U).$$}
   %In addition, if $\{f_i\}_{i\in I}$ is a directed subset of a function space, then $(\bigsqcup_{i\in I}f)[-]=\bigsqcup_{i\in I} f_i[-]$.}
 \end{remark}

\subsection{Mislove's Cantor tree}
If $x,y$ are strings, then we write $xy$ for the concatenation of $y$ after $x$.

The Cantor tree $\mathbb{C}$ consists of all countable-length \tr{strings} built from the tokens $0,1$ and \tr{endowed} with the prefix order.
 Then, $\mathbb{C}$ is a bounded complete algebraic domain, the compact elements are exactly finite strings.
 The empty string $\epsilon$ is the bottom element of $\mathbb{C}$, and the infinite-length strings are the maximal elements.

Mislove's Cantor tree and some results are introduced in \cite{Mislove2017}.
 \tr{The notation $\mathbb{C}^\surd$ stands for the set of all finite strings satisfying the following two condition:
  \begin{enumerate}[(i)]
    \item built from the tokens $0,1$;
    \item ending with an extra token $\surd$.
  \end{enumerate}}
 Equivalently, $\mathbb{C}^\surd=\{x\surd:x\in\mathcal{K}(\mathbb{C})\}$.
 Mislove's Cantor tree is the set $\mathbb{M}:=\mathbb{C}\cup\mathbb{C}^\surd$ endowed with the \emph{prefix order} \tr{(see Figure \ref{PicOfMisloveTree})}.
 We write $x\surd\in F$ for some $x\in\mathbb{C}$ such that $x\surd\in F$, where $F\subseteq\mathbb{M}$.
 \tr{The prefix order on $\mathbb{M}$ can also be defined by:
 \begin{enumerate}[(i)]
   \item $x\leq y$ if $x,y\in\mathbb{C}$ and $x\leq y$ in $\mathbb{C}$;
   \item for every $x\in\mathbb{C}$, $x\leq x\surd$;
   \item $x\surd,y\in\mathbb{M}$, $x\surd\leq y$ if $x\surd=y$.
 \end{enumerate}}
 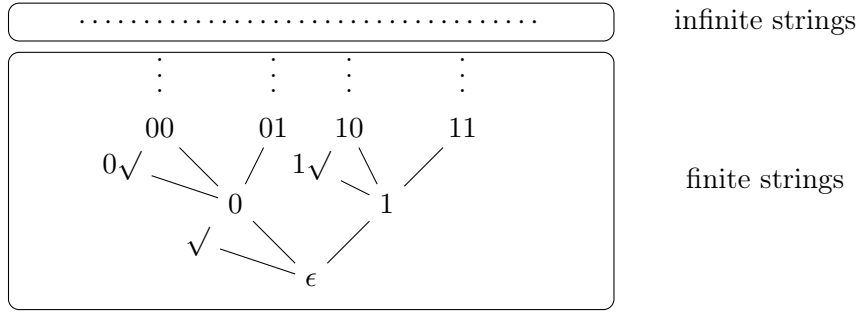
\begin{figure}
   \centering
   \begin{tikzpicture}
       \node (bot)at(0,0){$\epsilon$};
       \node (surd)at(-1.5,0.5){$\surd$};
       \node (0)at(-1,1){$0$};
       \node (1)at(1,1){$1$};
       \node (0surd)at(-2.5,1.5){$0\surd$};
       \node (00)at(-2,2){$00$};
       \node (01)at(-0.5,2){$01$};
       \node (1surd)at(0,1.5){$1\surd$};
       \node (10)at(0.5,2){$10$};
       \node (11)at(2,2){$11$};
       \draw [-](node cs: name=bot)--(surd);
       \draw [-](node cs: name=bot)--(0);
       \draw [-](node cs: name=bot)--(1);
       \draw [-](node cs: name=0)--(0surd);
       \draw [-](node cs: name=0)--(00);
       \draw [-](node cs: name=0)--(01);
       \draw [-](node cs: name=1)--(1surd);
       \draw [-](node cs: name=1)--(10);
       \draw [-](node cs: name=1)--(11);
       \node at(-2,2.5){.};
       \node at(-2,2.7){.};
       \node at(-2,2.9){.};
       \node at(-0.5,2.5){.};
       \node at(-0.5,2.7){.};
       \node at(-0.5,2.9){.};
       \node at(0.5,2.5){.};
       \node at(0.5,2.7){.};
       \node at(0.5,2.9){.};
       \node at(2,2.5){.};
       \node at(2,2.7){.};
       \node at(2,2.9){.};
%       \draw (-3,3.2) rectangle (3,3.7);
       \tikzstyle{esquare}=[rectangle,rounded corners,minimum width=8cm,minimum height=0.5cm,draw=black];
       \coordinate (locinf) at (0,3.4);
       \node (infs) [esquare] at(locinf){$\cdots\cdots\cdots\cdots\cdots\cdots\cdots\cdots\cdots\cdots\cdots\cdots$};
       \node at (6,3.4){infinite strings};
       \coordinate (locfin) at (0,1.3);
       \node (fins) [esquare,minimum height=3.4cm] at(locfin){};
       \node at (6,1.3){finite strings};
   \end{tikzpicture}
   \caption{Mislove's Cantor tree $\mathbb{M}$}
   \label{PicOfMisloveTree}
 \end{figure}
 \begin{remark}
   The maximal elements of $\mathbb{M}$ consist of $\mathbb{C}^\surd$ and the infinite strings of $\mathbb{C}$.
    The compact elements $\mathcal{K}(\mathbb{M})$ of $\mathbb{M}$ are finite strings, \tr{which also form a tree.
    Since every infinite string is a sup of the directed subset consisting of its finite prefixes, $\mathbb{M}$ is a countably based bounded complete algebraic domain.
    The concatenation is not an order-preserving operation in $\mathbb{M}$.
    Especially, $\mathbb{C}^\surd$ is Scott open and $\mathbb{C}$ is Scott closed in $\mathbb{M}$.}
 \end{remark}

\begin{remark}
  \tr{Every domain is meet-continuous \cite[Theorem III-2.11]{Gierz2003}, which ensures that $\ua(U\cap F)$ is Scott open for every lower subset $F$ and Scott open subset $U$ \cite[Proposition III-2.3]{Gierz2003}.
   Thus, for every Scott open subset $U\in\mathcal{O}\Sigma\mathbb{M}$ and Scott closed subset $F\in\Gamma\Sigma\mathbb{M}$, the subset $\ua(U\cap F)$ is Scott open in $\mathbb{M}$.}
\end{remark}

\subsection{Valuations on Mislove's Cantor tree}

\tr{Firstly, as $\mathbb{M}$ is countably based bounded complete, an application of \cite[Propostion 2.5]{Mislove2020} to $\Sigma\mathbb{M}$ yields as the following:}
\begin{proposition}\label{BasisOfVC*}
  $\mathcal{V}_{\leq 1}\Sigma\mathbb{M}$ has a basis $B_{\mathbb{M}}:=\{\Sigma_{x\in F}r_x\delta_x\in\mathcal{V}_{\leq 1}\Sigma\mathbb{M}:F\subseteq^{fin}\mathcal{K}(\mathbb{M})\}$.
\end{proposition}
%\begin{proposition}
%  For every Scott open subset $U\in\mathcal{O}\Sigma\mathbb{M}$ and Scott closed subset $F\in\Gamma\Sigma\mathbb{M}$, the subset $\ua(U\cap F)$ is Scott open.
%\end{proposition}
%\begin{proof}
%  Suppose that $C$ is a chain of $\mathbb{M}$ with $\bigsqcup C\in\ua(U\cap F)$.
%   In the case $\bigsqcup C\in U\cap F$, there is a $c\in C\cap U$.
%   Since $c\leq\bigsqcup C\in F$ and $F$ is a lower subset and $c\in F$.
%   Thus, $c\in U\cap F$.
%   In the case that $\bigsqcup C\notin U\cap F$, there is an $a\in U\cap F$ such that $\bigsqcup C>a$.
%   It is easy to verify that every element strictly less than another element of $\mathbb{M}$ has to be compact; hence, $a\in\mathcal{K}(\mathbb{M})$.
%   It follows that $a$ is lower than some $c'\in C$.
%   Obviously, $c'\in C\cap\ua (U\cap F)$.
%\end{proof}

\tr{
For every $T_0$ space, there is a Scott continuous map $\supp_X:\mathcal{V}_{\leq 1}X\ra\Gamma X$ given by $\supp_X(\mu)=X\setminus\bigcup\{U\in\mathcal{O}X:\mu(U)=0\}$ (see the proof of Lemma 3 of \cite{Mislove2017}, it carries over to $T_0$ spaces).
 In \cite{Fritz2021}, $\supp$ is thought of as a monad transformation.
 Since the underlying space for distributions is fixed to be $\Sigma\mathbb{M}$ throughout this paper, we abbreviate $\supp_{\Sigma\mathbb{M}}$ as $\supp$.
 The Scott closed subset $\supp(\mu)$ is the minimum Scott closed subset that $\mu$ concentrates on.
 So it is clearly to understand the following.}

%Define a map $\tr{\supp}:\mathcal{V}_{\leq 1}\Sigma\mathbb{M}\ra\Gamma\Sigma\mathbb{M}$ by $\tr{\supp}(\mu)=\mathbb{M}\setminus\bigcup\{U\in\mathcal{O}\Sigma\mathbb{M}:\mu(U)=0\}$.
% The subset $\tr{\supp}(\mu)$ indicates the smallest Scott closed subset $A$ such that $\mu(\mathbb{M}\setminus A)=0$.
% We generally call $\tr{\supp}(\mu)$ the support of $\mu$.
% The map $\tr{\supp}$ is Scott continuous \cite{Goubault2011}.

\begin{proposition}\label{SuppDetermineValue}
  For every \tr{Scott} open subset $U\in\mathcal{O}\Sigma\mathbb{M}$ and $\mu\in\mathcal{V}_{\leq 1}\Sigma\mathbb{M}$, $\mu(U)=\mu(\ua(U\cap \tr{\supp}(\mu)))$.
\end{proposition}
\begin{proof}
  \tr{Obviously $\mu(U)\geq\mu(\ua(U\cap\supp(\mu)))$; and conversely,
  \begin{align*}
   \mu(U) & \leq \mu(\ua(U\cap\supp(\mu))\cup(\mathbb{M}\setminus\supp(\mu))) \\
     & = \mu(\ua(U\cap\supp(\mu)))+\mu(\mathbb{M}\setminus\supp(\mu))-\mu(\ua(U\cap\supp(\mu))\cap(\mathbb{M}\setminus\supp(\mu))) \\
     & = \mu(\ua(U\cap\supp(\mu)))+0-0 \\
     & = \mu(\ua(U\cap\supp(\mu))).
  \end{align*}}
\end{proof}
%\begin{remark}\label{TopOfLowerSet}
%  Let $S$ be a lower subset of $\mathbb{M}$.
%   Then, $S$ is an algebraic poset with the relative order, and the Scott topology of $S$ agrees the relative Scott topology.
%   For every Scott open subset $U$ of $S$, $\ua U=\{y\in\mathbb{M}:\exists x\in S\text{ s.t. }x\leq y\}$ is a Scott open subset of $\mathbb{M}$
%\end{remark}

 \begin{corollary}\label{ValuationsOnFiniteTree}
   Let $\mu\in\mathcal{V}_{\leq 1}\Sigma\mathbb{M}$ be a continuous valuation that $\supp(\mu)$ is finite.
    Then, $\mu$ is a simple valuation.
 \end{corollary}
 \begin{proof}
   The support $\supp(\mu)$ is finite; hence, the family $\{\ua(U\cap\supp(\mu)):U\in\mathcal{O}\Sigma\mathbb{M}\}$ is finite.
    Since $\mu(U)=\mu(\ua(U\cap\supp(\mu)))$ for every $U\in\mathcal{O}\Sigma\mathbb{M}$, $\mu$ is a finitely valued valuation.
    Notice that $\Sigma\mathbb{M}$ is sober, $\mu$ is a simple valuation by Proposition \ref{FiniteAreSimple}.
 \end{proof}
 \begin{proposition}\label{BMEnjoyMeet}
   The basis $B_{\mathbb{M}}$ enjoys meet, that is, for every $\Sigma_{x\in F}r_x\delta_x$ and $\Sigma_{y\in G}s_y\delta_y$ of the basis $B_{\mathbb{M}}$, $\Sigma_{x\in F}r_x\delta_x\wedge\Sigma_{y\in G}s_y\delta_y$ exists in $\mathcal{V}_{\leq 1}\Sigma\mathbb{M}$, and it is also included in $B_{\mathbb{M}}$.
 \end{proposition}
 \begin{proof}
 \tr{Note that the meet $\Sigma_{x\in F}r_x\delta_x\wedge\Sigma_{y\in G}s_y\delta_y$ should be a simple valuation by Corollary 2.14 if it exists.
    Since it is complex to determine the mass of each point of $\supp(\Sigma_{x\in F}r_x\delta_x\wedge\Sigma_{y\in G}s_y\delta_y)$, we directly identify the meet by its values on open subsets}.

   For every $z\in\da F\cap\da G$, let $h_z=\tr{\min}(\{\Sigma_{x\in\ua z\cap F}r_x,\Sigma_{y\in\ua z\cap G}s_y\})$.
    \tr{We firstly need a claim that if $z\in\da F\cap\da G$ and $H$ is an arbitrary anti-chain of $\ua z\cap\da F\cap\da G$, then $h_z\geq \Sigma_{z'\in H}h_{z'}$.
    The claim is proved as follws:
    \begin{align*}
      h_z & = \min(\{\Sigma_{x\in\ua z\cap F}r_x,\Sigma_{y\in\ua z\cap G}s_y\}) \\
       & \geq \min(\{\Sigma_{x\in\ua H\cap F}r_x,\Sigma_{y\in\ua H\cap G}s_y\}) \\
       & = \min(\{\Sigma_{z'\in H}\Sigma_{x\in\ua z'\cap F},\Sigma_{z'\in H}\Sigma_{y\in\ua z'\cap G}s_y\}) \\
       & \geq \Sigma_{z'\in H}\min(\{\Sigma_{x\in\ua z'\cap F}r_x,\Sigma_{y\in\ua z'\cap G}s_y\})\\
       & = \Sigma_{z'\in H}h_{z'}.
    \end{align*}
    This claim will be used shortly afterwards.}

   Define $\mu:\mathcal{O}\Sigma\mathbb{M}\ra[0,1]$ by
   \begin{equation*}
     \mu(W)=\begin{cases}
              \Sigma_{z\in \min(\da F\cap\da G\cap W)}h_z, & \mbox{if $\da F\cap\da G\cap W\neq\emptyset$}; \\
              0, & \mbox{otherwise}.
            \end{cases}
   \end{equation*}
    \tr{Obviously $\mu$ is a finitely valued map by its formula.
    Now we show that $\mu$ is the meet.
    And the first step, we check that $\mu$ is a valuation.}
    \tr{It is easy to see $\mu(\emptyset)=0$.}
    Suppose $U,V$ are Scott open subsets of $\mathbb{M}$.
    We directly have
    \begin{align*}
      \mu(U)+\mu(V) & = \Sigma_{z\in \min(\da F\cap\da G\cap U)}h_z+\Sigma_{z'\in \min(\da F\cap\da G\cap V)}h_{z'} \\
       & = \Sigma_{z\in \min(\da F\cap\da G\cap U\cap V)}h_z+\Sigma_{z'\in \min(\da F\cap\da G\cap(U\cup V))}h_{z'} \\
       & = \mu(U\cap V)+\mu(U\cup V).
    \end{align*}
    If $U\subseteq V$ and $\da F\cap\da G\cap U\neq\emptyset$, then, for every $z'\in \min(\da F\cap\da G\cap U)$, there is a $z\in \min(\da F\cap G\cap V)$ lower than it; hence,
    \begin{align*}
      \mu(V) & = \Sigma_{z\in \min(\da F\cap\da G\cap V)}h_z \\
       & \geq \Sigma_{z\in \min(\da F\cap\da G\cap V)}\Sigma\{h_{z'}:z'\in\ua z\cap \min(\da F\cap\da G\cap U)\} \tag{\tr{by our claim}}\\
       & \geq \Sigma_{z'\in \min(\da F\cap\da G\cap U)}h_{z'} \\
       & = \mu(U).
    \end{align*}
    \tr{Since $\mu$ is finitely valued, it is easy to see the Scott continuity of $\mu$ from its monotonic.}
    By Corollary \ref{ValuationsOnFiniteTree}, $\mu$ is a simple valuation of $\mathcal{V}_{\leq 1}\Sigma\mathbb{M}$.
    Obviously, $\supp(\mu)\subseteq\da F\cap\da G\subseteq^{fin}\mathcal{K}(\mathbb{M})$, \tr{$\mu$ belongs to $B_{\mathbb{M}}$}.

   Assume \tr{$\upsilon\in\mathcal{V}_{\leq 1}\Sigma\mathbb{M}$ is less than $\Sigma_{x\in F}r_x\delta_x$ and $\Sigma_{y\in G}s_y\delta_y$ and $W\in\mathcal{O}\Sigma\mathbb{M}$}.
    Then, $\supp(\upsilon)\subseteq\da F\cap\da G$ is finite, and
    \begin{align*}
      \upsilon(W) & = \tr{\upsilon(\ua(W\cap\supp(\upsilon)))} \tag{\tr{by Proposition \ref{SuppDetermineValue}}} \\
       & \leq \tr{\upsilon(\ua(W\cap\da F\cap\da G))}\\
       & = \tr{\upsilon(\bigcup\{\ua z:z\in \min(\da F\cap\da G\cap W)\})} \\
       & = \Sigma_{z\in \min(\da F\cap\da G\cap W)}\upsilon(\ua z) \tag{\tr{due to the fact that $\da F\cap\da G\cap W\subseteq \da F\cap\da G$ is finite}}\\
       & \leq \tr{\Sigma_{z\in \min(\da F\cap\da G\cap W)}\min(\{\Sigma_{x\in\ua z\cap F}r_x,\Sigma_{y\in\ua z\cap G}s_y\})} \\
       & = \Sigma_{z\in \min(\da F\cap\da G\cap W)}h_z \\
       & = \mu(W).
    \end{align*}
 \end{proof}
 \begin{corollary}
   $\mathcal{V}_{\leq 1}\Sigma\mathbb{M}$ is a bounded complete domain.
 \end{corollary}
 \begin{proof}
   \tr{We only need to verify that every non-empty subset has an inf.}
   For every $\mu,\upsilon\in\mathcal{V}_{\leq 1}\Sigma\mathbb{M}$, denote
   \begin{align*}
     \tr{\varpi} & = \bigsqcup\{\mu'\wedge\upsilon':\mu'\in\dda\mu\cap B_{\mathbb{M}},\upsilon'\in\dda\upsilon\cap B_{\mathbb{M}}\}.
   \end{align*}
   \tr{Assume that $\pi\leq\mu,\upsilon$.
   For every $\pi'\in\dda\pi\cap B_\mathbb{M}$, we have $\pi'\ll \pi\leq\mu,\upsilon$.
   Hence, $\pi'$ is less than a $\mu'\in\dda\mu\cap B_{\mathbb{M}}$ and a $\upsilon'\in\dda\upsilon\cap B_{\mathbb{M}}$ at the same time; furthermore, $\pi'$ is less than $\mu'\wedge\upsilon'$.
   Therefore, we obtain $\pi'\leq \varpi$.
   It follows that $\pi\leq\varpi$, and $\varpi$ is exactly the inf of $\mu$ and $\upsilon$.}

   \tr{So for every non-empty subset $\mathcal{F}$ of $\mathcal{V}_{\leq 1}\Sigma\mathbb{M}$, its inf is equal to $\inf (\{\wedge F:F\subseteq^{fin}\mathcal{F}\})$: $\{\wedge F:F\subseteq^{fin}\mathcal{F}\}$ obviously is a filtered subset, and the existence of the inf follows from Theorem \ref{BigConsequenceofV} (\ref{VC*FiteredHasInf}).
   Consequently, $\mathcal{V}_{\leq 1}\Sigma\mathbb{M}$ is bounded complete.}
 \end{proof}
 Finally, we conclude that $\mathcal{V}_{\leq 1}\Sigma\mathbb{M}$ is a bounded complete domain and the weak topology coincides with the Scott topology.

\tr{Then, we provide a simple characterization of continuous valuations on $\mathbb{M}$.}

\tr{
Let $X$ be a $T_0$ space, $\mu\in\mathcal{V}_{\leq 1}X$, $U\in\mathcal{O}X$ and $F\in\Gamma X$.
 The restriction $\mu|_U:\mathcal{O}X\ra[0,1]$ is defined by $\mu|_U(V)=\mu(U\cap V)$, and the co-restriction $\mu|^F:\mathcal{O}X\ra[0,1]$ is defined to be $\mu|^F(V)=\mu(V)-\mu(V\setminus F)$, and according to \cite[Proposition 3.3]{Heckmann1996}, we have the following:}
\begin{proposition}
   \tr{$\mu|_U,\mu|^F$ are continuous valuation on $X$.
   Especially, if $U=X\setminus F$, then $\mu=\mu|_U+\mu|^F$.}
\end{proposition}
 \begin{corollary}
   \tr{For every $\mu\in\mathcal{V}_{\leq 1}\Sigma\mathbb{M}$, $\mu|_{\mathbb{C}^\surd}=\Sigma_{x\surd\in\supp(\mu)}\mu(\{x\surd\})\delta_{x\surd}$.}
 \end{corollary}
 \begin{proof}
   \tr{Suppose $U$ is Scott open.
    It is easy to see that for every $F\subseteq^{fin}U\cap\mathbb{C}^\surd$, $F$ is Scott open in and $\mu(F)=\Sigma_{x\surd\in F}\mu(\{x\surd\})$.
    Then, $U\cap\mathbb{C}^\surd=\bigsqcup\{F\subseteq^{fin}U\}$.
    Since $U\cap\mathbb{C}^\surd$ is countable and $\mu(U)\leq 1$,
    \begin{align*}
      \mu|_{\mathbb{C}^\surd}(U) & =\mu(U\cap\mathbb{C}^\surd) \\
       & = \mu(\bigsqcup\{F\subseteq^{fin}U\}) \\
       & = \sup(\{\mu(F):F\subseteq^{fin}U\cap\mathbb{C}^\surd\})\\
       & = \sup(\{\Sigma_{x\surd\in F}\mu(\{x\surd\}):F\subseteq^{fin}U\cap\mathbb{C}^\surd\}) \\
       & = \Sigma_{x\surd\in U\cap\mathbb{C}^\surd}\mu(\{x\surd\})\\
       & = \Sigma_{x\surd\in\mathbb{M}}\mu(\{x\surd\})\delta_{x\surd}(U).
    \end{align*}
    Because for every $x\surd\in\mathbb{C}^\surd\setminus\supp(\mu)$, $\mu|_{\mathbb{C}^\surd}(\{x\surd\})=\mu(\{x\surd\})=0$, we have $$\Sigma_{x\surd\in\mathbb{M}}\mu(\{x\surd\})\delta_{x\surd}=\Sigma_{x\surd\in\supp(\mu)}\mu(\{x\surd\})\delta_{x\surd}.$$}
 \end{proof}
 \begin{corollary}
   \tr{Every $\mu\in\mathcal{V}_{\leq 1}\Sigma\mathbb{M}$ is equal to the sum of a continuous valuation $\mu|^{\mathbb{C}}$ and discrete distribution $\mu|_{\mathbb{C}^\surd}$.}
 \end{corollary}

\tr{
The inclusion $inc_{\mathbb{C}}:\mathbb{C}\subseteq\mathbb{M}$ obviously is a Scott continuous embedding.
 So for every $\mu\in\mathcal{V}_{\leq 1}\Sigma\mathbb{C}$, $inc_{\mathbb{C}}[\mu]\in\mathcal{V}_{\leq 1}\Sigma\mathbb{M}$.
 Observe that for every $x\surd\in\mathbb{M}$, $inc_{\mathbb{C}}[\mu](\ua x\surd)=inc_{\mathbb{C}}[\mu](\{x\surd\})=\mu(\emptyset)=0$.
 Thus, $\mu$ can be regard as a valuation on $\mathbb{M}$ with no mass on these $x\surd$.}
\begin{theorem}(\cite[Theorem 5.3]{Keimel2005})\label{ValuationExtension}
  \tr{Every continuous valuation on a locally compact sober space uniquely extends to a Borel measure.}
\end{theorem}
\tr{It is well-known that on a countably based space $X$, every Borel subprobability measure is Scott continuous from $\mathcal{O}X$ to $[0,1]$.
 Hence, if $X$ is locally compact sober, its continuous valuations and Borel subprobability measures are coincident.}

\tr{The Cantor space $\max(\mathbb{C})$ is the subspace of $\Sigma\mathbb{C}$ that consists of maximal elements, hence it is a subspace of $\Sigma\mathbb{M}$.
 In addition, since the Cantor space is compact Hausdorff, and every compact Hausdorff space is locally compact sober \cite[Proposition 8.2.12]{Goubault2013}, the Cantor space $\max(\mathbb{C})$ is countably based and locally compact sober.
 By Theorem \ref{ValuationExtension}, Borel subprobablity measures on the Cantor space are exactly $\mathcal{V}_{\leq 1}\max(\mathbb{C})$.
 For every $\mu\in\mathcal{V}_{\leq 1}\max(\mathbb{C})$, the push forward $(inc_{\mathbb{C}}\circ inc_{\max})[\mu]$ restricted to the Cantor space $\max(\mathbb{C})$ agrees with $\mu$, where $inc_{\max}:\max(\mathbb{C})\subseteq\mathbb{C}$ is inclusion.
 And by Lemma \ref{SupPn=Id}, $(inc_{\mathbb{C}}\circ inc_{\max})[\mu]$ is approximated by $(p_{\mathbb{M}_n}\circ inc_{\mathbb{C}}\circ inc_{\max})[\mu]$.}

\section{Spaces of random variables}\label{SecSpace}

\begin{definition}
  A lower subset $F\subseteq\mathbb{M}$ is called $\surd$-max if for every $x\surd\in F$, $x$ is maximal in $F\cap\mathbb{C}$.
\end{definition}
 %If $F$ is a $\surd$-max lower subset, then $F\cap\mathbb{C}$ is a lower subset.

\begin{example}
\begin{enumerate}[(i)]
 \item \tr{Every lower subset of $\mathbb{C}$ is $\surd$-max.}
  \item \tr{For every $\mu\in\mathcal{V}_{\leq 1}\Sigma\mathbb{C}$, $\supp(inc_\mathbb{C}[\mu])\subseteq\mathbb{C}$ is $\surd$-max.}
  \item \tr{If $C$ is an anti-chain of $\mathbb{C}$, then for every $S\subseteq C$, $\{x\surd:x\in S\}\cup C$ is $\surd$-max.}
  \item \tr{Not every lower subset generated by an anti-chain of $\mathbb{M}$ is $\surd$-max.
  As an example, take $\da 1\surd\cup\da \surd=\{\epsilon,\surd,1,1\surd\}$: since for $\surd=\epsilon\surd$, we have $\epsilon\leq 1$, where $1\in(\da 1\surd\cup\da \surd)\cap\mathbb{C}=\{\epsilon, 1\}$.}
\end{enumerate}
\end{example}

\begin{definition}
  For a $T_0$ space $X$, $(\mu,f)\in\mathcal{V}_{\leq 1}\Sigma\mathbb{M}\times[\Sigma\mathbb{M}\rightharpoonup X]$ is called a continuous random variable on $X$ if $f$ is partially continuous with $dom(f)=\supp(\mu)$.
   The pair $(\Sigma_{y\in G}s_y\delta_y,g)$ is called a simple random variable on $X$ if $(\Sigma_{y\in G}s_y\delta_y,g)$ is a continuous random variable on $X$ and $\Sigma_{y\in G}s_y\delta_y\in B_{\mathbb{M}}$.
%   Let $\mathcal{V}_{\leq 1}^\surd\Sigma\mathbb{M}\subseteq\mathcal{V}_{\leq 1}\Sigma\mathbb{M}$ be the subset that each $\upsilon\in \mathcal{V}_{\leq 1}^\surd\Sigma\mathbb{M}$ has a $\surd$-max support (that is, $\supp(\upsilon)$ is $\surd$-max), $\mathcal{V}_{\leq 1}^\surd\Sigma\mathbb{M}\subseteq\mathcal{V}_{\leq 1}\Sigma\mathbb{M}$ be defined analogous.
   The continuous random variable $(\mu,f)$ is called a $\surd$-max continuous random variable on $X$ if $\supp(\mu)$ is $\surd$-max.
\end{definition}

\begin{remark}
  For every $\Sigma_{y\in G}s_y\delta_y\in B_{\mathbb{M}}$, $\supp(\Sigma_{y\in G}s_y\delta_y)=\da G$ is a finite subset, and a partial map $g:\mathbb{M}\rightharpoonup X$ with $dom(g)=\da G$ is partially continuous iff $g$ preserves the order.
\end{remark}

\begin{definition}
  \tr{Let $X$ be a $T_0$ space.
   For every $z\in\mathcal{K}(\mathbb{M}),r\in[0,1)$ and $V\in\mathcal{O}X$, define $\langle \ua z,r,V\rangle$ as the set of all $\surd$-max continuous random variables $(\mu,f)$ on $X$ that $\mu(\ua z)>r$ and $f(z)\in V$.
   Equivalently,
   \begin{equation*}
     \langle \ua z,r,V\rangle=\{(\mu,f):(\mu,f)\text{ is a }\surd\text{-max continuous random variable on }X,\mu(\ua z)> r,f(z)\in V\}.
   \end{equation*}
   Hence, define $\mathcal{CV}^\surd X$ as the set of all $\surd$-max continuous random variables on $X$ endowed with the topology generated by the subbase $\{\langle \ua z,r, V\rangle:z\in\mathcal{K}(\mathbb{M}),r\in[0,1),V\in\mathcal{O} X\}$.
% \begin{equation*}
%   \langle \ua z,r,V\rangle:=\{(\mu,f):\mu(\ua z)> r,f(z)\in V\}.
% \end{equation*}
   Also let $\mathcal{SV}^\surd X$ be the subspace of $\mathcal{CV}^\surd X$ which consists of all $\surd$-max simple random variables on $X$.}
\end{definition}

\begin{proposition}
  On $\mathcal{V}_{\leq 1}\Sigma\mathbb{M}$, the topology generated by the subbase $\{\langle\ua z>r\rangle:z\in\mathcal{K}(\mathbb{M}),r\in[0,1)\}$ coincides with the Scott topology.
\end{proposition}
\begin{proof}
  Every $\langle\ua z>r\rangle$ is weak open, then it is Scott open.
   Suppose $\Sigma_{x\in F}r_x\delta_x\in B_{\mathbb{M}}$.
   Consider the Scott open subset $U:=\bigcap_{z\in F}\langle \ua z>\Sigma_{x\in\ua z\cap F}r_x\rangle$.
   For every $\mu\in\dua\Sigma_{x\in F}r_x\delta_x$ and $z\in F$, by Theorem \ref{WBRelationOfVL}, we have
   \begin{equation*}
     \mu(\ua z)=\mu(\ua(\ua z\cap F))>\Sigma_{x\in\ua z\cap F}r_x.
   \end{equation*}
   It implies that $\mu\in U$, and hence, $\dua\Sigma_{x\in F}r_x\delta_x\subseteq U$.
   Conversely, if $\upsilon\in U$ and $G\subseteq F$, then
   \begin{equation*}
     \upsilon(\ua G)=\Sigma_{z\in \min(G)}\upsilon(\ua z)>\Sigma_{z\in \min(G)}\Sigma_{x\in\ua z\cap F}r_x\geq\Sigma_{x\in\ua G\cap F}r_x\geq\Sigma_{x\in G}r_x.
   \end{equation*}
   Thus, $\Sigma_{x\in F}r_x\delta_x\ll\upsilon$.
   We conclude $\dua\Sigma_{x\in F}r_x\delta_x=U$.
\end{proof}
\begin{corollary}\label{STopCoincideUWeakTopOnVC*}
  $\{\langle\ua z>r\rangle:z\in\mathcal{K}(\mathbb{M}),r\in[0,1)\}$ is a subbase of the weak topology on $\mathcal{V}_{\leq 1}\Sigma\mathbb{M}$
\end{corollary}

\begin{remark}\label{SpecOrderBySubbase}
  For every topological space $X$ with a subbase $B$, $x\leq y$ with respect to the specialization order iff for every $U\in B$, $x\in U$ implies $y\in U$.
\end{remark}
\begin{proposition}
  Let $X$ be a $T_0$ space.
   Then, $\mathcal{CV}^\surd X$ is a $T_0$ space, and $(\mu,f)\leq (\upsilon,g)$ iff $\mu\leq\upsilon$ in $\mathcal{V}_{\leq 1}\Sigma\mathbb{M}$ and \tr{$f\leq g$ in $[\Sigma\mathbb{M}\rightharpoonup X]$.}
\end{proposition}
\begin{proof}
  ($T_0$ separation axiom)
   Suppose that there are $\surd$-max continuous random variables $(\mu,f)$ and $(\upsilon,g)$ on $X$ with $(\mu,f)\neq(\upsilon,g)$.
   In the case that $\mu\neq\upsilon$, without loss of generality, there is a weak open subset $\bigcap_{i\in I}\langle\ua z_i>r_i\rangle$ that contains $\mu$ and excludes $\upsilon$.
   Then, $\bigcap_{i\in I}\langle\ua z_i,r_i,X\rangle$ is an open subset of $\mathcal{CV}^\surd X$ that contains $(\mu,f)$ and excludes $(\upsilon,g)$.
   In the case that $\mu=\upsilon$ and $f\neq g$, there is an $x\in \supp(\mu)$ such that $f(x)\neq g(x)$.
   As for the $T_0$ separation axiom of $X$, we assume that $V_x$ is a open subset of $X$ such that $f(x)\in V_x$ and $g(x)\notin V_x$.
   Hence, the open subset $\langle \ua x,0,V_x\rangle$ contains $(\mu,f)$ and excludes $(\upsilon,g)$.

  (The specialization order)
   Suppose that $(\mu,f)\in\langle \ua z,r,V\rangle$ always implies $(\upsilon,g)\in\langle \ua z,r,V\rangle$.
   Then, $\mu\in\langle\ua z>r\rangle$ implies $\upsilon\in\langle\ua z>r\rangle$.
   By Corollary \ref{STopCoincideUWeakTopOnVC*} and Remark \ref{SpecOrderBySubbase}, $\mu$ is lower than $\upsilon$ in $\mathcal{V}_{\leq 1}\Sigma\mathbb{M}$.
   Fix a $z\in \supp(\mu)$ and pick an arbitrary open subset $V_z\in\mathcal{O}X$ with $f(z)\in V_z$, we have $(\mu,f)\in\langle\ua z,0,V_z\rangle$, moreover, $(\upsilon,g)\in\langle \ua z,0,V_z\rangle$.
   Therefore, we obtain $g(z)\in V_z$.
   We conclude that $f(z)\leq g(z)$ with respect to the specialization order of $X$.
   The converse direction is obvious.
\end{proof}

For every $T_0$ space $X$, we denote by $[-]_1$ the projection from $\mathcal{CV}^\surd X$ to $\mathcal{V}_{\leq 1}\Sigma\mathbb{M}$, and denote by $[-]_2$ the projection from $\mathcal{CV}^\surd X$ to $[\Sigma\mathbb{M}\rightharpoonup X]$ (that is, $[(\mu,f)]_1=\mu,[(\mu,f)]_2=f$).
\begin{proposition}
  The projection $[-]_1:\mathcal{CV}^\surd X\ra\mathcal{V}_{\leq 1}\Sigma\mathbb{M}$ is continuous.
\end{proposition}
\begin{proof}
   By Corollary \ref{STopCoincideUWeakTopOnVC*}, we only need to verify that $[-]_1^{-1}(\langle \ua z>r\rangle)$ is open for every $z\in\mathcal{K}(\mathbb{M})$ and $r\in[0,1)$.
   Obviously, $[-]_1^{-1}(\langle \ua z>r\rangle)=\langle\ua z,r,X\rangle$.
\end{proof}

%\section{Some preparations}
%
%\begin{proposition}
%  Let $x\in\mathcal{K}(\mathbb{C})$ and $\mu\in\mathcal{V}_{\leq 1}\mathbb{M}$.
%   Then $\supp(\mathbf{x}[\mu])$ is equal to the disjoint union $\da x\cup\mathbf{x}(\supp(\mu)\setminus\{\epsilon\})$.
%\end{proposition}
%\begin{proof}
%  We start from $\Sigma_{y\in G}s_y\delta_y\in B_{\mathbb{M}}$.
%   Directly, $\supp(\mathbf{x}[\Sigma_{y\in G}s_y\delta_y])=\da\{xy:y\in G\}$.
%   By the tree structure,
%   \begin{equation*}
%     \da\{xy:y\in G\}=\da x\cup\{xy:y\in\da G\}=\da x\cup\{xy:y\in\da G\setminus\{bot\}\}=\da x\cup(\mathbf{x}[\supp(\Sigma_{y\in G}s_y\delta_y)]\setminus\{bot\}).
%   \end{equation*}
%   Hence
%  \begin{align*}
%    \supp(\mathbf{x}[\mu]) & = \supp(\mathbf{x}[\bigsqcup(\da\mu\cap B_{\mathbb{M}})]) \\
%     & = \bigsqcup\supp(\mathbf{x}[\da\mu\cap B_{\mathbb{M}}]) \\
%     & = \bigsqcup\{\supp(\mathbf{x}[\mu']):\mu'\in\da\mu\cap B_{\mathbb{M}}\} \\
%     & = \bigsqcup\mathbf{x}[\{\supp(\Sigma_{y\in G}s_y\delta_y)\in\da\mu\cap B_{\mathbb{M}}\}]
%  \end{align*}
%\end{proof}
%
%
%
%\tr{We denote $[\mathcal{CV}^\surd X]_1$ is the subset of all continuous valuations with $\surd$-max supports.}

\begin{definition}
  \tr{Let $X$ be a $T_0$ space.
   A continuous random variable $(\mu,f)$ on $X$ is said to be normalized whenever $\mu$ is normalized.
   Denote $\mathcal{CV}^\surd_1X$ as the subspace of $\mathcal{CV}^\surd X$ consisted of all normalized $\surd$-max continuous random variables.
   The subspace $\mathcal{SV}^\surd_1X\subseteq\mathcal{SV}^\surd X$ is defined analogously.
   In addition, set $\langle\ua z,r,V\rangle_1=\langle\ua z,r,V\rangle\cap\mathcal{CV}^\surd_1X$ and $\langle\ua z,r,V\rangle'_1=\langle\ua z,r,V\rangle\cap\mathcal{SV}^\surd_1X$.}
\end{definition}

\section{A monad structure of some simple random variables}\label{SecMonad}

%
% \tr{
% For every $x\in\mathcal{K}(\mathbb{M})$ and $\mu\in\mathcal{V}_{\leq 1}\Sigma\mathbb{M}$, we define
%  \begin{equation*}
%   \delta_x\ast\mu=\begin{cases}
%                        \delta_x, & \mbox{if}~x\in\mathbb{C};\\
%                        \mathbf{x}'[\mu], & \mbox{if}~x=x'\surd.
%                   \end{cases}
% \end{equation*}
% \begin{remark}
%   Especially, for every $\Sigma_{y\in G}s_y\delta_y\in B_{\mathbb{M}}$,
%   \begin{equation*}
%   \delta_x\ast\Sigma_{y\in G}s_y\delta_y=\begin{cases}
%                      \delta_x, & \mbox{if}~x\in\mathbb{C};\\
%                      \Sigma_{y\in G}s_y\delta_{x'y}, & \mbox{if}~x=x'\surd.
%                 \end{cases}
% \end{equation*}
% \end{remark}
% We remind readers that $\ast$ is not an operation, and we list several simple derivations to familiarize readers with $\ast$.
% \begin{enumerate}[(i)]
%   \item $\delta_\surd\ast\Sigma_{y\in G}s_y\delta_y=\delta_{\epsilon\surd}\ast\Sigma_{y\in G}s_y\delta_y=\Sigma_{y\in G}s_y\delta_{\epsilon y}=\Sigma_{y\in G}s_y\delta_y$;
%   \item $\delta_{010\surd}\ast\Sigma_{y\in G}s_y\delta_y=\Sigma_{y\in G}s_y\delta_{010y}=\Sigma_{y\in G,z=010y}s_y\delta_z$;
%   \item $\delta_\epsilon\ast\Sigma_{y\in G}s_y\delta_y=\delta_\epsilon$ and $\delta_{1011}\ast\Sigma_{y\in G}s_y\delta_y=\delta_{1011}$;
%   \item $\delta_{11}\ast\delta_\surd=\delta_{11},\delta_{11\surd}\ast\delta_\surd=\delta_{11\surd}$.
% \end{enumerate}}
\subsection{\tr{Intuitions of constructions.}}
\tr{Before introducing our monad structures, we need to reveal the intuition of our constructions.}

The sequential domain $\mathbb{M}$ is adopted to model sequences of coin flips during probabilistic computations.
 The intuition is that choosing elements from semantics domain following the outcome of coin-flips.

Now imagine a firm that manufactures a line of battery-equipped random coin-flipping machines.

A process manipulated by one such machine consists of two components: the machine flips coins randomly to produce either 0 or 1, records each outcome onto a paper tape, and updates its internal state based on the resulting sequence after every flip -- this operation naturally consumes battery power.

When the machine completes its flipping routine normally, it prints a tick mark $\surd$ at the very end of the tape’s recorded sequence.
 If the machine terminates abnormally, for instance due to a drained battery, the tape will not end with the tick $\surd$.

Accordingly, such a process is formally represented as a pair $(\mu,f)$.
 The left component $\mu$ is a distribution over $M$ of all sequences generated by random coin flips, while the right component $f$ is a Scott-continuous map that manipulates and drives state evolution.

Suppose we have two machine-manipulated processes $P=(\mu,f)$ and $Q=(\upsilon,g)$.
 Their sequential composition is defined as follows:

First, machine of $P$ runs and outputs a tape together with a final internal state.
 Machine of $Q$ then takes over this tape as its input.
 Machine of $Q$ reads the entire sequence on the tape from start to finish without modifying its own state during reading.
 If machine of $Q$ reaches the tick $\surd$ at the tape’s end, it erases this tick and commences its own execution: it flips coins, logs the new sequence of outcomes, and drives state transitions accordingly.
 If machine of $Q$ does not encounter a trailing tick $\surd$ by the end of the tape, it performs no operations whatsoever.

\subsection{\tr{Constructing a monad}}

\begin{definition}
  \tr{For every $x\in\mathcal{K}(\mathbb{C})$, the Scott continuous embedding $\mathbf{x}:\mathbb{M}\ra\mathbb{M}$ is defined by $\mathbf{x}(y)=xy$.}
\end{definition}
 \tr{Hence, there is a Scott continuous map $\mathbf{x}[-]:\mathcal{V}_{\leq 1}\Sigma\mathbb{M}\ra \mathcal{V}_{\leq 1}\Sigma\mathbb{M}$.}
\begin{remark}
  \tr{For a Dirac valuation $\delta_y$ and $x\in\mathcal{K}(\mathbb{C})$,
  \begin{equation*}
    \mathbf{x}[\delta_y](U)=\delta_y(\{y':xy'\in U\})=\begin{cases}
                                                        1, & \mbox{if } xy\in U \\
                                                        0, & \mbox{otherwise}.
                                                      \end{cases}=\delta_{xy}.
  \end{equation*}
   Hence, for every $\Sigma_{y\in G}s_y\delta_y\in\mathcal{V}_{\leq 1}\Sigma\mathbb{M}$, $\mathbf{x}[\Sigma_{y\in G}s_y\delta_y]=\Sigma_{y\in G}s_y\delta_{xy}$ by Remark \ref{p[-]properties}.}
\end{remark}

\tr{Thus, roughly speaking, the map $\mathbf{x}[-]$ simply transports each $\mu\in\mathcal{V}_{\leq 1}\mathbb{M}$ onto the subtree $\mathbf{x}(\mathbb{M})=\{xy:y\in\mathbb{M}\}=\ua x$.
 Notice that $x\leq y$ can not implies $\mathbf{x}[-]\leq\mathbf{y}[-]$ (a simple example is taking $\delta_\surd$, $\mathbf{x}[\delta_\surd]=\delta_{x\surd}\nleq\delta_{y\surd}=\mathbf{y}[\delta_\surd]$).}

In this paper, we use Kleisli triples, Manes' equivalent descriptions of monads.
\begin{definition}\cite{Manes1976}
  We say that $(\mathcal{T},e,\tr{\star})$ is a Kleisli triple over a category $\mathbf{C}$ if $\mathcal{T}:Obj(\mathbf{C})\ra Obj(\mathbf{C})$, $e_A:A\ra\mathcal{T}A$, and $\tr{\star}:Mor_\mathbf{C}(A,\mathcal{T}B)\ra Mor_\mathbf{C}(\mathcal{T}A,\mathcal{T}B)$ satisfy
  \begin{enumerate}[(i)]
    \item for every $A\in Obj(\mathbf{C})$ and $f\in Mor_\mathbf{C}(A,\mathcal{T}B)$, $f^{\tr{\star}}\circ e_A=f$;
    \item for every $A\in Obj(\mathbf{C})$, $e_A^{\tr{\star}}=id_{\mathcal{T}A}$;
    \item for every $f\in Mor_\mathbf{C}(A,\mathcal{T}B),g\in Mor_\mathbf{C}(B,\mathcal{T}C)$, $g^{\tr{\star}}\circ f^{\tr{\star}}=(g^{\tr{\star}}\circ f)^{\tr{\star}}$.
  \end{enumerate}
  \tr{We usually call $\star$ the Kleisli lift of $(\mathcal{T},e,\tr{\star})$.}
\end{definition}

\tr{We will introduce two monad structures, and their Kleisli lifts are denoted by $\dagger$ and $\ddagger$.
Firstly, we construct $\dagger$ for $\mathcal{SV}^\surd_1$.}

\tr{Let $X,Y$ be $T_0$ spaces and $h:X\ra\mathcal{SV}^\surd_1Y$ a continuous map.
 We use $(\Sigma_{y\in G_z}s_y\delta_y,g_z)$ to represent the image $h(z)$ for every $z\in X$ (that is, \tr{$\Sigma_{y\in G_z}s_y\delta_y=[h(z)]_1$, $g_z=[h(z)]_2$}).
 Define $h^\dagger:\mathcal{SV}^\surd_1X\ra\mathcal{SV}^\surd_1Y$ as follows: for every $(\Sigma_{x\in F}r_x\delta_x,f)\in\mathcal{SV}^\surd_1X$,
 \begin{equation*}
   h^\dagger(\Sigma_{x\in F}r_x\delta_x,f)=([h^\dagger(\Sigma_{x\in F}r_x\delta_x,f)]_1,[h^\dagger(\Sigma_{x\in F}r_x\delta_x,f)]_2),
 \end{equation*}
 where
 \begin{align*}
   [h^\dagger(\Sigma_{x\in F}r_x\delta_x,f)]_1:= &\Sigma_{x\in F\cap\mathbb{C}}r_x\delta_x+\Sigma_{x'\surd\in F}r_{x'\surd}\mathbf{x'}[[h(f(x'\surd))]_1]  \\
     = &\Sigma_{x\in F\cap\mathbb{C}}r_x\delta_x+\Sigma_{x'\surd\in F}r_{x'\surd}\mathbf{x'}[\Sigma_{y\in G_{f(x'\surd)}}s_y\delta_y]  \\
     = & \Sigma_{x\in F\cap\mathbb{C}}r_x\delta_x+\Sigma_{x'\surd\in F}r_{x'\surd}\Sigma_{y\in G_{f(x'\surd)}}s_y\delta_{x'y};
 \end{align*}
 and
 \begin{equation*}
   [h^\dagger(\Sigma_{x\in F}r_x\delta_x,f)]_2:\supp([h^\dagger(\Sigma_{x\in F}r_x\delta_x,f)]_1)\ra Y
 \end{equation*}
 is given by
 \begin{equation*}
   [h^\dagger(\Sigma_{x\in F}r_x\delta_x,f)]_2(z)=
   \begin{cases}
      g_{f(z)}(\epsilon), & \mbox{if $z\in\da F\cap\mathbb{C}$};\\
      g_{f(x'\surd)}(y), & \mbox{if}~x'\surd\in F,~z=x'y~\&~y\in\da G_{f(x'\surd)}\setminus\{\epsilon\}.
   \end{cases}
 \end{equation*}
 By our presentation, $[h^\dagger(\Sigma_{x\in F}r_x\delta_x,f)]_2(z)$ is equal to $[h(f(z))]_2(\epsilon)$ if
 \begin{equation}
   z\in \supp(\Sigma_{x\in F}r_x\delta_x)\cap\mathbb{C},\tag{\text{Case-I}}\label{eqDaggerPaetIICase-1}
 \end{equation}
 and is equal to $[h(f(x'\surd))]_2(y)$ if
 \begin{equation}
   x'\surd\in F,~z=x'y~\&~y\in \supp([h(f(x'\surd))]_1)\setminus\{\epsilon\}.\tag{\text{Case-II}}\label{eqDaggerPaetIICase-2}
 \end{equation}}

% \begin{center}
%   $\begin{cases}
%      [h(f(x'\surd))]_2(y), & \mbox{if $z=x'y$ for some $x'\surd\in F,y\in \supp([h(f(x'\surd))]_1)\setminus\{\epsilon\}$};  \\
%      [h(f(z))]_2(\epsilon), & \mbox{if $z\in \supp(\Sigma_{x\in F}r_x\delta_x)\cap\mathbb{C}$}. \label{eqDaggerPaetIICase-2}
%   \end{cases}$
% \end{center}}
% \begin{equation}
%   \begin{cases}
%      [h(f(x'\surd))]_2(y), & \mbox{if $z=x'y$ for some $x'\surd\in F,y\in \supp([h(f(x'\surd))]_1)\setminus\{\epsilon\}$}; \tag*{eqDaggerPaetIICase-1} \\
%      [h(f(z))]_2(\epsilon), & \mbox{if $z\in \supp(\Sigma_{x\in F}r_x\delta_x)\cap\mathbb{C}$} \tag*{eqDaggerPartIICase-2}.
%   \end{cases}
% \end{equation}

 \begin{proposition}
   \tr{$h^\dagger$ is well-defined.}
 \end{proposition}
 \begin{proof}
   Obviously, \tr{since $(\Sigma_{x\in F}r_x\delta_x,f)\in\mathcal{SV}^\surd_1X$,}
   \begin{align*}
     [h^\dagger(\Sigma_{x\in F}r_x\delta_x,f)]_1(\mathbb{M}) & = \Sigma_{x\in F\cap\mathbb{C}}r_x+\Sigma_{x'\surd\in F}r_{x'\surd}\Sigma_{y\in G_{f(x'\surd)}}s_y \\
      & = \Sigma_{x\in F\cap\mathbb{C}}r_x+\Sigma_{x'\surd\in F}r_{x'\surd}\\
      & = \Sigma_{x\in F}r_x\delta_x(\mathbb{M})\\
      & = 1.
   \end{align*}
   The Scott closed subset $\supp([h^\dagger(\Sigma_{x\in F}r_x\delta_x,f)]_1)$ is the subtree given below: $\da G_{f(x\surd)}$ attaches to the tree $\da F\cap\mathbb{C}$ \tr{at} the node $x$ for all $x\surd\in F$, that is,
   \begin{align}
    \supp([h^\dagger(\Sigma_{x\in F}r_x\delta_x,f)]_1) & = \da(F\cap\mathbb{C})\cup\da\{x'y:x'\surd\in F,y\in G_{f(x'\surd)}\} \nonumber\\
      & = \da(F\cap\mathbb{C})\cup\da\{x'y:x'\surd\in F,y\in\da G_{f(x'\surd)}\} \nonumber\\
      & = \da(F\cap\mathbb{C})\cup\da\{x':x'\surd\in F\}\cup\{x'y:x'\surd\in F,y\in\da G_{f(x'\surd)}\setminus\{\epsilon\}\} \nonumber\\
      & = (\da F\cap\mathbb{C})\cup\{x'y:x'\surd\in F,y\in\da G_{f(x'\surd)}\setminus\{\epsilon\}\}\label{eqDaggerPartISupportLike}.
   \end{align}
   \tr{Denote $\{x'y:x'\surd\in F,y\in\da G_{f(x'\surd)}\setminus\{\epsilon\}\}$ by $A$.
   Observe that for every $x'y$ of $A$, $x'$ is maximal in $\da F\cap\mathbb{C}$ by the $\surd$-max property and $y$ can not be $\epsilon$.
   So every $x'y$ of $A$ is strictly greater than a maximal element of $\da F\cap\mathbb{C}$.
   Thus $\supp([h^\dagger(\Sigma_{x\in F}r_x\delta_x,f)]_1)$ is a disjoint union.
   In $\supp([h^\dagger(\Sigma_{x\in F}r_x\delta_x,f)]_1)$, only $A$ contains strings ending with $\surd$, and these strings are maximal in $A$, which henceforth are maximal in $\supp([h^\dagger(\Sigma_{x\in F}r_x\delta_x,f)]_1)$.
   It follows that $\supp([h^\dagger(\Sigma_{x\in F}r_x\delta_x,f)]_1)$ is $\surd$-max.}
% For every $z$ of the support $supp([h^\dagger(\Sigma_{x\in F}r_x\delta_x,f)]_1)$, we let
% \begin{equation*}
%   [h^\dagger(\Sigma_{x\in F}r_x\delta_x,f)]_2(z)=
%   \begin{cases}
%      g_{f(x'\surd)}(y), & \mbox{if $z=x'y$ for some $x'\surd\in F,y\in\da G_{f(x'\surd)}\setminus\{\epsilon\}$};  \\
%      g_{f(z)}(\epsilon), & \mbox{if $z\in\da F\cap\mathbb{C}$},
%   \end{cases}
% \end{equation*}
%  we mention that $[h^\dagger(\Sigma_{x\in F}r_x\delta_x,f)]_2$ can be written as
%  \begin{equation*}
%   [h^\dagger(\Sigma_{x\in F}r_x\delta_x,f)]_2(z)=
%   \begin{cases}
%      [h(f(x'\surd))]_2(y), & \mbox{if $z=x'y$ for some $x'\surd\in F,y\in supp([h(f(x'\surd))]_1)\setminus\{\epsilon\}$};  \\
%      [h(f(z))]_2(\epsilon), & \mbox{if $z\in supp(\Sigma_{x\in F}r_x\delta_x)\cap\mathbb{C}$}.
%   \end{cases}
% \end{equation*}

  \tr{As \eqref{eqDaggerPartISupportLike} is a disjoint union, \eqref{eqDaggerPaetIICase-1} and \eqref{eqDaggerPaetIICase-2} are complementary to each other.
   Hence $[h^\dagger(\Sigma_{x\in F}r_x\delta_x,f)]_2$ is a well-defined map.}
   Assume that $z,z'$ belong to $\supp([h^\dagger(\Sigma_{x\in F}r_x\delta_x,f)]_1)$ with $z\leq z'$, \tr{we now need to show $[h^\dagger(\Sigma_{x\in F}r_x\delta_x,f)]_2(z)\leq [h^\dagger(\Sigma_{x\in F}r_x\delta_x,f)]_2(z')$.}
   If $z'\in\da F\cap\mathbb{C}$, then
   \begin{equation*}
     [h^\dagger(\Sigma_{x\in F}r_x\delta_x,f)]_2(z)=g_{f(z)}(\epsilon)\leq g_{f(z')}(\epsilon)=[h^\dagger(\Sigma_{x\in F}r_x\delta_x,f)]_2(z').
   \end{equation*}
   If $z\in\da F\cap\mathbb{C},z'=x'y$ for some $x'\surd\in F$, then $z\leq x'\leq z'$ since $x'$ is maximal in $\da F\cap\mathbb{C}$.
   Furthermore,
   \begin{equation*}
     [h^\dagger(\Sigma_{x\in F}r_x\delta_x,f)]_2(z)=g_{f(z)}(\epsilon)\leq g_{f(x')}(\epsilon)\leq  g_{f(x'\surd)}(\epsilon)\leq g_{f(x'\surd)}(y)=[h^\dagger(\Sigma_{x\in F}r_x\delta_x,f)]_2(z').
   \end{equation*}
   If $z=x'y$ for some $x'\surd\in F$ and $y\in G_{f(x'\surd)}$, then there is a $y'\in G_{f(x'\surd)}$ greater than $y$ such that $z'=x'y'$.
   Therefore,
   \begin{equation*}
     [h^\dagger(\Sigma_{x\in F}r_x\delta_x,f)]_2(z)=g_{f(x'\surd)}(y)\leq g_{f(x'\surd)}(y')=[h^\dagger(\Sigma_{x\in F}r_x\delta_x,f)]_2(z').
   \end{equation*}
   By the above analysis, $[h^\dagger(\Sigma_{x\in F}r_x\delta_x,f)]_2$ preserves the order.
 \end{proof}

% Then, we establish a well-defined map $h^\dagger:\mathcal{SV}^\surd_1X\ra\mathcal{SV}^\surd_1Y$ by
%\begin{equation*}
%  h^\dagger(\Sigma_{x\in F}r_x\delta_x,f)=([h^\dagger(\Sigma_{x\in F}r_x\delta_x,f)]_1,[h^\dagger(\Sigma_{x\in F}r_x\delta_x,f)]_2).
%\end{equation*}

\begin{proposition}
  $h^\dagger$ is continuous.
\end{proposition}
\begin{proof}
  Suppose that $\langle \ua z,r,V\rangle'_1$ is an open subset of $\mathcal{SV}^\surd_1Y$ and $(\Sigma_{x\in F}r_x\delta_x,f)\in h^{\dagger-1}(\langle \ua z,r,V\rangle'_1)$.
   Then, we have
   \begin{align}
     & [h^\dagger(\Sigma_{x\in F}r_x\delta_x,f)]_1(\ua z)=(\Sigma_{x\in F\cap\mathbb{C}}r_x\delta_x+\Sigma_{x'\surd\in F}r_{x'\surd}\Sigma_{y\in G_{f(x'\surd)}}s_y\delta_{x'y})(\ua z)>r;\nonumber\\
     & [h^\dagger(\Sigma_{x\in F}r_x\delta_x,f)]_2(z)\in V .\nonumber
   \end{align}
   We analyze it in two cases: (\romannumeral1) $z\in\da F\cap\mathbb{C}$; (\romannumeral2) $z\notin\da F\cap\mathbb{C}$.

  \textbf{Case}(\romannumeral1). For every $x'\surd\in F,y\in G_{f(x'\surd)}$, since $\da F$ is $\surd$-max, $x'y\geq z$ if and only if $x'\surd\geq z$.
   Thus, we obtain
   \begin{align*}
     [h^\dagger(\Sigma_{x\in F}r_x\delta_x,f)]_1(\ua z) & = \Sigma_{x\in F\cap\mathbb{C}}r_x\delta_x(\ua z)+\Sigma_{x'\surd\in F}r_{x'\surd}\Sigma_{y\in G_{f(x'\surd)}}s_y\delta_{x'y}(\ua z)\\
      & = \Sigma_{x\in\ua z\cap F\cap\mathbb{C}} r_x+\Sigma_{x'\surd\in \ua z\cap F}r_{x'\surd}\Sigma_{y\in G_{f(x'\surd)}}s_y\\
      & = \Sigma_{x\in\ua z\cap F}r_x\\
      & > r.
   \end{align*}
    In addition, directly have
   \begin{equation*}
     [h^\dagger(\Sigma_{x\in F}r_x\delta_x,f)]_2(z)=[h(f(z))]_2(\epsilon)=g_{f(z)}(\epsilon)\in V.
   \end{equation*}
   Hence, $f(z)\in h^{-1}(\langle\ua\epsilon,0,V\rangle'_1)$.
   Then, it is simple to check that
   \begin{equation*}
     (\Sigma_{x\in F}r_x\delta_x,f)\in\langle\ua z,r,h^{-1}(\langle\ua\epsilon,0,V\rangle'_1)\rangle'_1.
   \end{equation*}
   For each $(\Sigma_{\omega\in K}a_\omega\delta_\omega,k)\in\langle\ua z,r,h^{-1}(\langle\ua\epsilon,0,V\rangle'_1)\rangle'_1$, we have $z\in\da K\cap\mathbb{C}$ and $[h(k(z))]_2(\epsilon)\in V$.
   So,
   \begin{align*}
     [h^\dagger(\Sigma_{\omega\in K}a_\omega\delta_\omega,k)]_1(\ua z) & = \Sigma_{\omega\in K\cap\mathbb{C}}a_\omega\delta_\omega(\ua z)+\Sigma_{\omega'\surd\in K}a_{\omega'\surd}\Sigma_{y\in G_{k(\omega'\surd)}}s_y\delta_{\omega'y}(\ua z)\\
      & = \Sigma_{\omega\in\ua z\cap K\cap\mathbb{C}} r_x+\Sigma_{\omega'\surd\in \ua z\cap K}a_{\omega'\surd}\Sigma_{y\in G_{f(\omega'\surd)}}s_y\\
      & = \Sigma_{\omega\in\ua z\cap K}a_\omega \\
      & > r;
   \end{align*}
   \tr{on the other hand,} $[h^\dagger(\Sigma_{\omega\in K}a_\omega\delta_\omega,k)]_2(z)=[h(k(z))]_2(\epsilon)=g_{k(z)}(\epsilon)\in V$.
   We conclude that $h^\dagger(\Sigma_{\omega\in K}a_\omega\delta_\omega,k)\in\langle\ua z,r,V\rangle'_1$.

   \textbf{Case}(\romannumeral2). There are some $x'\surd\in F$ and $y'\in G_{f(x'\surd)}\setminus\{\epsilon\}$ with $z=x'y'$.
   Thus,
   \begin{align*}
     [h^\dagger(\Sigma_{x\in F}r_x\delta_x,f)]_1(\ua z) & = r_{x'\surd}\Sigma_{y\in G_{f(x'\surd)}}s_y\delta_{x'y}(\ua z) \\
      & = r_{x'\surd}\Sigma_{y\in G_{f(x'\surd)}}s_y\delta_y(\ua y') \\
      & = r_{x'\surd}\Sigma_{y\in\ua y'\cap G_{f(x'\surd)}}s_y \\
      & > r
   \end{align*}
   and $[h^\dagger(\Sigma_{x\in F}r_x\delta_x,f)]_2(z)=[h(f(x'\surd))]_2(y')=g_{f(x'\surd)}(y')\in V$ hold.
   Pick non-negative real numbers $\varepsilon_{x'\surd},\gamma_{x'\surd}$ such that
   \begin{equation*}
     (r_{x'\surd}-\varepsilon_{x'\surd})(\Sigma_{y\in\ua y'\cap G_{f(x'\surd)}}s_y-\gamma_{x'\surd})>r.
   \end{equation*}
   It is easy to verify $(\Sigma_{x\in F}r_x\delta_x,f)\in U$, where
   \begin{equation*}
     U:=\langle\ua x'\surd,r_{x'\surd}-\varepsilon_{x'\surd},h^{-1}(\langle\ua y',(\Sigma_{y\in\ua y'\cap G_{f(x'\surd)}}s_y)-\gamma_{x'\surd},V\rangle'_1)\rangle'_1.
   \end{equation*}
   Suppose $(\Sigma_{\omega\in K}a_\omega\delta_\omega,k)\in U$.
   Then $\Sigma_{\omega\in K}a_\omega\delta_\omega(\ua x'\surd)=a_{x'\surd}>r_{x'\surd}-\varepsilon_{x'\surd}$, and
   \begin{equation*}
     [h(k(x'\surd))]_1(\ua y')=\Sigma_{G_{k(x'\surd)}}s_y\delta_{y}(\ua y')>(\Sigma_{y\in\ua y'\cap G_{f(x'\surd)}}s_y)-\gamma_{x'\surd}.
   \end{equation*}
   Thus
   \begin{align*}
     [h^\dagger(\Sigma_{\omega\in K}a_\omega\delta_\omega,k)]_1(\ua z) & = a_{x'\surd}\Sigma_{G_{k(x'\surd)}}s_y\delta_{x'y}(\ua z)\\
      & = a_{x'\surd}\Sigma_{G_{k(x'\surd)}}s_y\delta_{y}(\ua y')\\
      & > (r_{x'\surd}-\varepsilon_{x'\surd})((\Sigma_{y\in\ua y'\cap G_{f(x'\surd)}}s_y)-\gamma_{x\surd})\\
      & > r.
   \end{align*}
   Besides,
   \begin{equation*}
     h(f(x'\surd))=(\Sigma_{y\in G_{k(x'\surd)}}s_y\delta_y,g_{k(x'\surd)})\in\langle\ua y',(\Sigma_{y\in\ua y'\cap G_{k(x'\surd)}}s_y)-\gamma_{x'\surd},V\rangle'_1,
   \end{equation*}
   which implies $[h^\dagger(\Sigma_{\omega\in K}a_\omega\delta_\omega,k)]_2(z)=g_{k(x'\surd)}(y')\in V$.
\end{proof}

For every $T_0$ space $X$, define $\eta_X:X\ra\mathcal{SV}^\surd_1X$ \tr{by $\eta_X(x)=(\delta_\surd,\chi_x)$}, where
\begin{equation*}
  \chi_x=\{\surd\mapsto x,\epsilon\mapsto x\}.
\end{equation*}
 Here, we use some tricks to present a partial map $f:A\rightharpoonup B$, we regard $f$ as a subset of $A\times B$, and $x\mapsto f(x)$ indicates $(x,f(x))\in f$.
 Thus, $\chi:\{\epsilon,\surd\}\ra X$ can also be given by $\chi_x(\epsilon)=x,\chi(\surd)=x$.
 \tr{Obviously, $\supp(\delta_\surd)=\da\surd=\{\surd,\epsilon\}$ is $\surd$-max.
 Thus the random variable $(\delta_\surd,\chi_x)$ is well-defined continuous random variable.}
\begin{proposition}
  $\eta_X$ is continuous.
\end{proposition}
\begin{proof}
  For every $\langle \ua z,r,V\rangle'_1$, $\eta_X^{-1}(\langle\ua z,r,V\rangle'_1)=
  \begin{cases}
      V, & \mbox{if $z=\surd$ or $z=\epsilon$;} \\
      \emptyset, & \mbox{otherwise}.
  \end{cases}$
\end{proof}

Let $\mathbf{TOP0}$ be the category of $T_0$ spaces and continuous maps.

\begin{theorem}
  $(\mathcal{SV}^\surd_1,\eta,\dagger)$ is a Kleisli triple over $\mathbf{TOP0}$.
\end{theorem}
\begin{proof}
  Let $X,Y,Z$ be $T_0$ spaces, $h:X\ra\mathcal{SV}^\surd_1Y,l:Y\ra\mathcal{SV}^\surd_1Z$ be continuous maps and $(\Sigma_{y\in G_x}s_y\delta_y,g_x)$ represent $h(x)$ for each $x\in X$.
%   In this proof, we respectively use
%  \begin{equation*}
%    (\Sigma_{\kappa\in G_x}s_y\delta_y,g_x),(\Sigma_{\omega\in K_\kappa}a_\omega\delta_\omega,k_\kappa)
%  \end{equation*}
%   to denote $h(x),l(y)$ for every $x\in X,\kappa\in Y$.

  (\romannumeral1) For every $x\in X$,
  \begin{equation*}
    [h^\dagger\circ\eta_X(x)]_1=[h^\dagger(\delta_\surd,\chi_x)]_1=\mathbf{\epsilon}[\Sigma_{y\in G_x}s_y\delta_{y}]=\Sigma_{y\in G_x}s_y\delta_y.
  \end{equation*}
  For each $z\in\da G_x$,
  \begin{align*}
    [h^\dagger\circ\eta_X(x)]_2(z) & = [h^\dagger(\delta_\surd,\chi_x)]_2(z) \\
     & \tr{= \begin{cases}
           g_{\chi_x(z)}(\epsilon), & \mbox{if}~z\in\da\surd\cap\mathbb{C}; \\
           g_{\chi_x(\omega\surd)}(y), & \mbox{if}~\omega\surd\in\{\surd\},~z=\omega y~\&~y\in\da G_{\chi_x(\omega\surd)}\setminus\{\epsilon\}
         \end{cases}} \\
     & = \begin{cases}
           g_{\chi_x(\epsilon)}(\epsilon), & \mbox{if } z=\epsilon; \\
           g_{\chi_x(\surd)}(z), & \mbox{if } z\in\da G_x\setminus\{\epsilon\}
         \end{cases} \\
     & = g_x(z).
  \end{align*}
%  \begin{equation*}
%    [h^\dagger\circ\eta_X(x)]_2(z)=[h^\dagger(\delta_\surd,\chi_x)]_2(z)=
%    \begin{cases}
%      g_{\chi_x(\epsilon)}(\epsilon), & \mbox{if } z=\epsilon; \\
%      g_{\chi_x(\surd)}(z), & \mbox{if } z\in\da G_x\setminus\{\epsilon\}
%    \end{cases}
%    =g_x(z).
%  \end{equation*}

  (\romannumeral2) For every $(\Sigma_{x\in F}r_x\delta_x,f)\in\mathcal{SV}^\surd_1X$,
  \tr{
  \begin{align*}
    [\eta_X^\dagger(\Sigma_{x\in F}r_x\delta_x,f)]_1 & = \Sigma_{x\in F\cap\mathbb{C}}r_x\delta_x+\Sigma_{x'\surd\in F}r_{x'\surd}\mathbf{x'}[[\eta_X(f(x))]_1] \\
     & = \Sigma_{x\in F\cap\mathbb{C}}r_x\delta_x+\Sigma_{x'\surd\in F}r_{x'\surd}\mathbf{x'}[[(\delta_\surd,\chi_{f(x)})]_1] \\
     & = \Sigma_{x\in F\cap\mathbb{C}}r_x\delta_x+\Sigma_{x'\surd\in F}r_{x'\surd}\mathbf{x'}[\delta_\surd] \\
     & = \Sigma_{x\in F\cap\mathbb{C}}r_x\delta_x+\Sigma_{x'\surd\in F}r_{x'\surd}\delta_{x'\surd} \\
     & = \Sigma_{x\in F}r_x\delta_x.
  \end{align*}}
  For each $z\in\da F$,
  \begin{align*}
    [\eta_X^\dagger(\Sigma_{x\in F}r_x\delta_x,f)]_2(z) & = \begin{cases}
      [\eta_X(f(z))]_2(\epsilon), & \mbox{if}~z\in\da F\cap\mathbb{C}; \\
      [\eta_X(f(x'\surd))]_2(\surd), & \mbox{if}~x'\surd\in F,~z=x'\surd
    \end{cases} \\
     & = \begin{cases}
      \chi_{f(z)}(\epsilon), & \mbox{if } z\in\da F\cap\mathbb{C}; \\
      \chi_{f(x'\surd)}(\surd), & \mbox{if}~x'\surd\in F,~z=x'\surd
    \end{cases}\\
     & = f(z).
  \end{align*}

  (\romannumeral3) For every $(\Sigma_{x\in F}r_x\delta_x,f)\in\mathcal{SV}^\surd_1X$,
  \tr{
  \begin{align*}
      [l^\dagger\circ h^\dagger(\Sigma_{x\in F}r_x\delta_x,f)]_1
     %%%%%%%%%%%%%%%%%%%%%%%%%%%%%%%%%
     & = [l^\dagger([h^\dagger(\Sigma_{x\in F}r_x\delta_x,f)]_1,[h^\dagger(\Sigma_{x\in F}r_x\delta_x,f)]_2)]_1 \\
    %%%%%%%%%%%%%%%%%%%%%%%%%%%%%%%%%
     & = [l^\dagger(\Sigma_{x\in F\cap\mathbb{C}}r_x\delta_x+\Sigma_{x'\surd\in F}r_{x'\surd}\Sigma_{y\in G_{f(x'\surd)}}s_y\delta_{x'y},[h^\dagger(\Sigma_{x\in F}r_x\delta_x,f)]_2)]_1 \tag{by Def of $[h^\dagger(...)]_1$}\\
    %%%%%%%%%%%%%%%%%%%%%%%%%%%%%%%%%
      & = [l^\dagger(\Sigma_{x\in F\cap\mathbb{C}}r_x\delta_x+\Sigma_{x'\surd\in F}r_{x'\surd}\Sigma_{y\in G_{f(x'\surd)}\cap\mathbb{C}}s_y\delta_{x'y}\\
     & + \Sigma_{x'\surd\in F}r_{x'\surd}\Sigma_{y'\surd\in G_{f(x'\surd)}}s_{y'\surd}\delta_{x'y'\surd},[h^\dagger(\Sigma_{x\in F}r_x\delta_x,f)]_2)]_1 \\
    %%%%%%%%%%%%%%%%%%%%%%%%%%%%%%%%%
     & = \Sigma_{x\in F\cap\mathbb{C}}r_x\delta_x+\Sigma_{x'\surd\in F}r_{x'\surd}\Sigma_{y\in G_{f(x'\surd)}\cap\mathbb{C}}s_y\delta_{x'y}\\
     & + \Sigma_{x'\surd\in F}r_{x'\surd}\Sigma_{y'\surd\in G_{f(x'\surd)}}s_{y'\surd}\mathbf{x'y'}[[l([h^\dagger(\Sigma_{x\in F}r_x\delta_x,f)]_2(x'y'\surd))]_1],\tag{by Def of $[l^\dagger(...)]_1$}
  \end{align*}}
  and
  \tr{\begin{align*}
     [(l^\dagger\circ h)^\dagger(\Sigma_{x\in F}r_x\delta_x,f)]_1
     %%%%%%%%%%%%%%%%%%%%%%%%%%%%%%%%%
      & = \Sigma_{x\in F\cap\mathbb{C}}r_x\delta_x+\Sigma_{x'\surd\in F}r_{x'\surd}\mathbf{x'}[[(l^\dagger\circ h)(f(x))]_1] \\
    %%%%%%%%%%%%%%%%%%%%%%%%%%%%%%%%%
      & = \Sigma_{x\in F\cap\mathbb{C}}r_x\delta_x+\Sigma_{x'\surd\in F}r_{x'\surd}\mathbf{x'}[[(l^\dagger(\Sigma_{y\in G_{f(x'\surd)}}s_y\delta_y,g_{f(x'\surd)})]_1] \tag{expand $h(f(x'\surd))$}\\
    %%%%%%%%%%%%%%%%%%%%%%%%%%%%%%%%%
      & = \Sigma_{x\in F\cap\mathbb{C}}r_x\delta_x+\Sigma_{x'\surd\in F}r_{x'\surd}\mathbf{x'}[\Sigma_{y\in G_{f(x'\surd)}\cap\mathbb{C}}s_y\delta_y \\
      & +\Sigma_{y'\surd\in G_{f(x'\surd)}}s_{y'\surd}\mathbf{y'}[[l(g_{f(x'\surd)}(y))]_1]]\tag{by Def of $[l^\dagger(...)]_1$}\\
    %%%%%%%%%%%%%%%%%%%%%%%%%%%%%%%%%
      & = \Sigma_{x\in F\cap\mathbb{C}}r_x\delta_x+\Sigma_{x'\surd\in F}r_{x'\surd}\Sigma_{y\in G_{f(x'\surd)}\cap\mathbb{C}}s_y\delta_{x'y}\\
     & + \Sigma_{x'\surd\in F}r_{x'\surd}\Sigma_{y'\surd\in G_{f(x'\surd)}}s_{y'\surd}\mathbf{x'y'}[[l([h^\dagger(\Sigma_{x\in F}r_x\delta_x,f)]_2(x'y'\surd))]_1].
  \end{align*}}
  It is easy to see that $\supp([l^\dagger\circ h^\dagger(\Sigma_{x\in F}r_x\delta_x,f)]_1)$ is equal to the disjoint union of the following subsets: $\da F\cap\mathbb{C},\{x'y:x'\surd\in F,y\in(\da G_{f(x'\surd)}\cap\mathbb{C})\setminus\{\epsilon\}\}$ and
  \begin{equation*}
    \{x'y'\omega:x'\surd\in F,y'\surd\in G_{f(x'\surd)},\omega\in \supp([l(g_{f(x'\surd)}(y'\surd))]_1)\setminus\{\epsilon\}\}.
  \end{equation*}

  \tr{Now we need a break to introduce a new notation.
  For any statement $S$, label $\lambda$, and mathematical object $D$ equipped with a case wise definition, we write
 $S\overset{\triangleright}{\longrightarrow}\lambda\in D$
  as an abbreviation for
  the fact that the statement $S$ falls under the case labelled $\lambda$ in the definition of $D$.}

  If $z\in \da F\cap\mathbb{C}$, then
  \tr{\begin{align*}
   [l^\dagger\circ h^\dagger(\Sigma_{x\in F}r_x\delta_x,f)]_2(z)
     & = [l([h^\dagger(\Sigma_{x\in F}r_x\delta_x,f)]_2(z))]_2(\epsilon) \tag{due to} \\
     & \tag{$z\in\da F\cap \mathbb{C}\subseteq\supp([h^\dagger(\Sigma_{x\in F}r_x\delta_x,f)]_1)\cap\mathbb{C}$}\\
     & \tag{$\casef\eqref{eqDaggerPaetIICase-1}\in[l^\dagger\circ h^\dagger(\Sigma_{x\in F}r_x\delta_x,f)]_2$} \\
     & = [l([h(f(z))]_2(\epsilon))]_2(\epsilon) \tag{due to $z\in\da F\cap\mathbb{C}$} \\
     & \tag{$\casef\eqref{eqDaggerPaetIICase-1}\in[h^\dagger(\Sigma_{x\in F}r_x\delta_x,f)]_2$} \\
     & = [l(g_z(\epsilon))]_2(\epsilon),
  \end{align*}}
  and
  \tr{\begin{align*}
   [(l^\dagger\circ h)^\dagger(\Sigma_{x\in F}r_x\delta_x,f)]_2(z)
     & = [(l^\dagger\circ h)(f(z))]_2(\epsilon) \tag{due to $z\in\da F\cap\mathbb{C}$}\\
     & \tag{$\casef\eqref{eqDaggerPaetIICase-1}\in[(l^\dagger\circ h)^\dagger(\Sigma_{x\in F}r_x\delta_x,f)]_2$} \\
     & = [(l^\dagger(\Sigma_{y\in G_{f(z)}}s_y\delta_y,g_z)]_2(\epsilon) \tag{expand $h(f(z))$} \\
     & = [l(g_z(\epsilon))]_2(\epsilon) \tag{due to $\epsilon\in\da G_{f(z)}\cap\mathbb{C}$}.\\
     & \tag{$\casef\eqref{eqDaggerPaetIICase-1}\in[(l^\dagger(\Sigma_{y\in G_{f(z)}}s_y\delta_y,g_z)]_2$}
  \end{align*}}
  If $z=x'y$ for some $x'\surd\in F,y\in(\da G_{f(x'\surd)}\cap\mathbb{C})\setminus\{\epsilon\}$, then
  \tr{\begin{align*}
   [l^\dagger\circ h^\dagger(\Sigma_{x\in F}r_x\delta_x,f)]_2(z)
     & = [l([h^\dagger(\Sigma_{x\in F}r_x\delta_x,f)]_2(z))]_2(\epsilon) \tag{due to} \\
     & \tag{$z\in\supp([h^\dagger(\Sigma_{x\in F}r_x\delta_x,f)]_1)\cap\mathbb{C}$}\\
     & \tag{$\casef\eqref{eqDaggerPaetIICase-1}\in[l^\dagger\circ h^\dagger(\Sigma_{x\in F}r_x\delta_x,f)]_2$}\\
     & = [l([h(f(x'\surd))]_2(y))]_2(\epsilon) \tag{due to $z=x'y\notin \da F\cap\mathbb{C}$},\\
     & \tag{$\casef\eqref{eqDaggerPaetIICase-2}\in[h^\dagger(\Sigma_{x\in F}r_x\delta_x,f)]_2(z))]_2$}
  \end{align*}}
  and
  \tr{\begin{align*}
     [(l^\dagger\circ h)^\dagger(\Sigma_{x\in F}r_x\delta_x,f)]_2(z)
     & = [(l^\dagger\circ h)(f(x'\surd))]_2(y) \tag{due to $z=x'y\notin\da F\cap\mathbb{C}$}\\
     & \tag{$\casef\eqref{eqDaggerPaetIICase-2}\in[(l^\dagger\circ h)^\dagger(\Sigma_{x\in F}r_x\delta_x,f)]_2$}\\
     & = [l([h(f(x'\surd))]_2(y))]_2(\epsilon) \tag{due to}\\
     & \tag{$y\in\da G_{f(x'\surd)}\cap\mathbb{C}=\supp([h(f(x'\surd))]_1)\cap\mathbb{C}$}.\\
     & \tag{$\casef\eqref{eqDaggerPaetIICase-1}\in[(l^\dagger\circ h)(f(x'\surd))]_2$}
  \end{align*}}
  If $z=x'y'\omega$ for some $x'\surd\in F,y'\surd\in G_{f(x'\surd)},\omega\in \supp([l(g_{f(x'\surd)}(y'\surd))]_1)\setminus\{\epsilon\}$, then %K_{g_{f(x'\surd)}(y'\surd)}
  \tr{\begin{align*}
  [l^\dagger\circ h^\dagger(\Sigma_{x\in F}r_x\delta_x,f)]_2(z)
     & = [l([h^\dagger(\Sigma_{x\in F}r_x\delta_x,f)]_2(x'y'\surd))]_2(\omega) \tag{due to}\\
     & \tag{$z=x'y'\omega\notin\supp([h^\dagger(\Sigma_{x\in F}r_x\delta_x,f)]_1)\cap\mathbb{C}$}  \\
     & \tag{$\casef\eqref{eqDaggerPaetIICase-2}\in[l^\dagger\circ h^\dagger(\Sigma_{x\in F}r_x\delta_x,f)]_2$}\\
     & = [l([h(f(x'\surd))]_2(y'\surd))]_2(\omega)\tag{due to $x'y'\surd\notin\da F\cap\mathbb{C}$},\\
     & \tag{$\casef\eqref{eqDaggerPaetIICase-2}\in[h^\dagger(\Sigma_{x\in F}r_x\delta_x,f)]_2$}
    % & = [l^\dagger\circ h(f(x'\surd))]_2(y'\omega)\\
%     & = [(l^\dagger\circ h)^\dagger(\Sigma_{x\in F}r_x\delta_x,f)]_2(z).
  \end{align*}}
  and
  \tr{\begin{align*}
   [(l^\dagger\circ h)^\dagger(\Sigma_{x\in F}r_x\delta_x,f)]_2(z)
     & = [l^\dagger\circ h(f(x'\surd))]_2(y'\omega) \tag{due to $z=x'y'\omega\notin\da F\cap\mathbb{C}$}\\
     & \tag{$\casef\eqref{eqDaggerPaetIICase-2}\in[(l^\dagger\circ h)^\dagger(\Sigma_{x\in F}r_x\delta_x,f)]_2$}\\
     & = [l([h(f(x'\surd))]_2(y'\surd))]_2(\omega). \tag{due to}\\
     & \tag{$y'\omega\notin\da G_{f(x'\surd)}\cap\mathbb{C}=\supp([h(f(x'\surd))]_1)\cap\mathbb{C}$}\\
     & \tag{$\casef\eqref{eqDaggerPaetIICase-2}\in[l^\dagger\circ h(f(x'\surd))]_2$}
  \end{align*}}
\end{proof}

\begin{definition}
  A triple $(\mathcal{T},e,m)$ is a monad over a category $\mathbf{C}$ if $\mathcal{T}:\mathbf{C}\ra\mathbf{C}$ is a functor, $e:\mathcal{ID}_{\mathbf{C}}\ra\mathcal{T},m:\mathcal{T}^2\ra\mathcal{T}$ are natural transformations such that, for every $A\in Obj(\mathbf{C})$,
  \begin{enumerate}[(i)]
    \item $m_{\mathcal{T}A}\circ\mathcal{T}(e_A)=id_{\mathcal{T}A}=m_{\mathcal{T}A}\circ e_{\mathcal{T}A}$;
    \item $m_A\circ m_{\mathcal{T}A}=m_A\circ \mathcal{T}(m_A)$.
  \end{enumerate}
\end{definition}

The equivalence of Kleisli triples and monads is given as follows:
\begin{enumerate}[(i)]
  \item every Kleisli triple $(\mathcal{P},v,\star)$ corresponds to a monad $(\mathcal{P},v,w)$ by defining $\mathcal{P}(f:A\ra B)=(v_B\circ f)^\star$ and $w_A=id_{\mathcal{P}A}^\star$,
  \item every monad $(\mathcal{T},e,m)$ corresponds to a Kleisli triple $(\mathcal{T},e,\star)$ by defining $g^\star=m_B\circ\mathcal{T}(g)$ for each $g:A\ra\mathcal{T}B$.
\end{enumerate}

\begin{lemma}\label{SVApplyToMaps}
  For every map $h:X\ra Y$ between the $T_0$ spaces, the functor $\mathcal{SV}^\surd_1$ applied on $h$ is given by $\mathcal{SV}^\surd_1(h)(\Sigma_{x\in F}r_x\delta_x,f)=(\Sigma_{x\in F}r_x\delta_x,h\circ f)$.
\end{lemma}
\begin{proof}
\tr{
  \begin{align*}
    \mathcal{SV}^\surd_1(h)(\Sigma_{x\in F}r_x\delta_x,f) & = (\eta_Y\circ h)^\dagger(\Sigma_{x\in F}r_x\delta_x,f) \\
    & = (\Sigma_{x\in F\cap\mathbb{C}}r_x\delta_x+\Sigma_{x'\surd\in F}r_{x'\surd}\mathbf{x'}[[\eta_Y(h(x))]_1],[(\eta_Y\circ h)^\dagger(\Sigma_{x\in F}r_x\delta_x,f)]_2)\\
    & = (\Sigma_{x\in F\cap\mathbb{C}}r_x\delta_x+\Sigma_{x'\surd\in F}r_{x'\surd}\mathbf{x'}[[(\delta_\surd,\chi_{h(x)})]_1],[(\eta_Y\circ h)^\dagger(\Sigma_{x\in F}r_x\delta_x,f)]_2)\\
     & = (\Sigma_{x\in F\cap\mathbb{C}}r_x\delta_x+\Sigma_{x'\surd\in F}r_{x'\surd}\mathbf{x'}[\delta_\surd],[(\eta_Y\circ h)^\dagger(\Sigma_{x\in F}r_x\delta_x,f)]_2)\\
     & = (\Sigma_{x\in F}r_x\delta_x,[(\eta_Y\circ h)^\dagger(\Sigma_{x\in F}r_x\delta_x,f)]_2),
  \end{align*}}
  where
  \begin{align*}
    [(\eta_Y\circ h)^\dagger(\Sigma_{x\in F}r_x\delta_x,f)]_2(z) & =
    \begin{cases}
      [\eta_Y(h(f(z)))]_2(\epsilon), & \mbox{if } z\in\da F\cap\mathbb{C}; \\
      [\eta_Y(h(f(x'\surd)))]_2(\epsilon), & \mbox{if $z=x'\surd$ for some $x'\surd\in F$}
    \end{cases}\\
      & =
    \begin{cases}
      \chi_{h(f(z))}(\epsilon), & \mbox{if } z\in\da F\cap\mathbb{C}; \\
      \chi_{h(f(x'\surd))}(\surd), & \mbox{if $z=x'\surd$ for some $x'\surd\in F$}
    \end{cases}\\
      & = h(f(z)).
  \end{align*}
\end{proof}

\section[The necessity of the v-max property]{The necessity of the $\surd$-max property}

The $\surd$-max property is the key to ensuring that, for every $h:X\ra\mathcal{SV}^\surd_1Y$, $h^\dagger$ preserves the specialization order, where it is noted that every continuous map has to preserve the specialization order.

For every bounded complete domain $L$, let $CRV(L)$ be the set of all continuous random variables on $\Sigma L$ ordered by $(\mu,f)\leq (\upsilon,g)$ if $\mu\leq\upsilon$ and for every $x\in dom(f)$, $f(x)\leq g(x)$, and $SRV(L)$ be the subposet of simple random variables on $\Sigma L$.
% For every $x,y\in\mathcal{K}(\mathbb{M})$, define $x\cdot y=x$ if $x\in\mathbb{C}$, and $x\cdot y=x'y$ if $x=x'\surd$ for some $x'\in\mathbb{C}$.

For every bounded complete domain $L,M$ and order-preserving map $h:L\ra SRV(M)$, represent $h(x)$ by $(\Sigma_{y\in G_{x}}s_y\delta_y,g_x)$.
 In \cite{Mislove2017}, Mislove defined $h^\diamond:SRV(L)\ra SRV(M)$ as

 \begin{equation*}
   h^\diamond(\Sigma_{x\in F}r_x\delta_x,f)=([h^\diamond(\Sigma_{x\in F}r_x\delta_x,f)]_1,[h^\diamond(\Sigma_{x\in F}r_x\delta_x,f)]_2),
 \end{equation*}
 where
 \begin{align*}
    & [h^\diamond(\Sigma_{x\in F}r_x\delta_x,f)]_1:=\Sigma_{x\in F\cap\mathbb{C}}r_x\delta_x+\Sigma_{x'\surd\in F}r_{x'\surd}\Sigma_{y\in G_{f(x'\surd)}}s_y\delta_{x'y};\\
    & [h^\diamond(\Sigma_{x\in F}r_x\delta_x,f)]_2:=\bigwedge\{g_{f(x)}(y):z=x\in\da F\cap\mathbb{C},y=\epsilon\text{; or }x=x'\surd\in F,y\in\da G_{f(x)},\text{ s.t. }z=x'y\}.
 \end{align*}

There is a counter example that a $h^\diamond:SRV(L)\ra SRV(M)$ does not preserve the order.
 Let $L,M$ be the bounded complete domains described in Figure \ref{fig1}.
\begin{figure}
  \centering
  \begin{tikzpicture}
    \node (lb)at(-2,0){$l_\epsilon$};
    \node (l0)at(-1,1){$l_0$};
    \node (ls)at(-3,1){$l_\surd$};
    \draw [-](node cs: name=lb)--(l0);
    \draw [-](node cs: name=lb)--(ls);
    \node at (-2,-1){$L$};
    \node (mb)at(2,0){$m_\epsilon$};
    \node (m0)at(3,1){$m_0$};
    \node (ms)at(1,1){$m_\surd$};
    \draw [-](node cs: name=mb)--(m0);
    \draw [-](node cs: name=mb)--(ms);
    \node at (2,-1){$M$};
  \end{tikzpicture}
  \caption{}\label{fig1}
\end{figure}
Consider $(\mu,f),(\upsilon,g)\in SRV(L)$ and $h:L\ra SRV(M)$ which are given by
\begin{enumerate}[(i)]
  \item $(\mu,f)=(0.5\delta_\epsilon+0.5\delta_0,\{\epsilon\mapsto l_\epsilon, 0\mapsto l_0\})$;
  \item $(\upsilon,g)=(0.5\delta_\surd+0.5\delta_0,\{\epsilon\mapsto l_\epsilon, 0\mapsto l_0, \surd\mapsto l_\surd\})$;
  \item $h=\begin{cases}
          l_\epsilon\mapsto(\delta_\epsilon,\{\epsilon\mapsto m_\epsilon\}); \\
          l_0\mapsto(\delta_0,\{\epsilon\mapsto m_0, 0\mapsto m_0\}); \\
          l_\surd\mapsto(\delta_0,\{\epsilon\mapsto m_\epsilon, 0\mapsto m_0\}).
        \end{cases}$
\end{enumerate}
 We can straightforwardly obtain $(\mu,f)\leq (\upsilon,g)$.
However, $h^\diamond(\mu,f)=(0.5\delta_\epsilon+0.5\delta_0,\{\epsilon\mapsto m_\epsilon,0\mapsto m_0\})$ is not lower than $h^\diamond(\upsilon,g)=(\delta_0,\{\epsilon\mapsto m_\epsilon,0\mapsto m_\epsilon\})$.

\section{\tr{Over spaces, completions and over bounded complete domains}}\label{SecSob}

\tr{Modelling computable effects via power spaces was first proposed by Smyth in \cite{Smyth1983}.}
 After that, many papers studied the monads of power spaces over spaces \cite{Heckmann2013,Zhao2015,Xu2020,Goubault2010}.
\tr{Sober spaces have a strong connection with logic (see \cite[Theorem 8.1.26]{Goubault2013} for Stone duality).
 Besides, Scholars generally recognize open subsets as computable properties, and sobriety indicates that every completely prime filter many properties uniquely determines an element (these open subsets which meet a same irreducible subset form a completely prime filter on the lattice of open subsets).
 Another important kind consists of d-spaces, which are viewed as extensions of dcpos.
 We will introduce the monad structure of $\mathcal{CV}^\surd_1$ constructed over d-spaces in Section \ref{SecMonad2}.
  Before that, we need some preparations.}

\subsection{\tr{Over sober spaces and d-spaces}}

\tr{Combining Proposition IV-8.5 and Exercise IV-8.22 of \cite{Gierz2003}, we obtain the following:}
\begin{fact}\label{FactOfGammaC*}
  Some facts of $\Gamma\Sigma\mathbb{M}$.
  \begin{enumerate}[(i)]
    \item $\Gamma\Sigma\mathbb{M}$ is an algebraic domain and every subfamily $\mathcal{F}$ of $\Gamma\Sigma\mathbb{M}$ has a sup given by $\overline{\bigcup\mathcal{F}}$;
    \item $\mathcal{K}(\Gamma\Sigma\mathbb{M})=\{F:F\subseteq^{fin}\mathcal{K}(\mathbb{M}),F=\da F\}$;
    \item for an $F\in\mathcal{K}(\Gamma\Sigma\mathbb{M})$ and a $G\in\Gamma\Sigma\mathbb{M}$, $F\ll G$ if and only if $F\subseteq G$;
    \item for every $F_1,F_2$ of $\Gamma\Sigma\mathbb{M}$, $F_1\cup F_2=\bigsqcup\{F_1'\cup F_2':F_1',F_2'\in\mathcal{K}(\Gamma\Sigma\mathbb{M}),F_1'\subseteq F_1, F_2'\subseteq F_2\}$.
  \end{enumerate}
\end{fact}

%Since $\mathcal{CV}^\surd X\subseteq\mathcal{V}_{\leq 1}\Sigma\mathbb{M}\times[\Sigma\mathbb{M}\rightharpoonup X]$, it is natural to conjecture that the sup of every $\mathcal{F}\subseteq^{\ua}\mathcal{CV}^\surd X$, whenever it exists, is equal to $(\sup([\mathcal{F}]_1,\sup([\mathcal{F}]_2)))$.

\begin{lemma}\label{SupOfDirectedSurdMax}
  If $\mathcal{F}$ is a directed family of $\surd$-max Scott closed subsets of $\mathbb{M}$, then $\bigsqcup\mathcal{F}$ is $\surd$-max.
\end{lemma}
\begin{proof}
  Suppose that $\bigsqcup\mathcal{F}=\overline{\bigcup\mathcal{F}}$ is not $\surd$-max.
   Then, there are $x\surd\in\overline{\bigcup\mathcal{F}},y\in\overline{\bigcup\mathcal{F}}\cap\mathbb{C}$ with $x\surd\neq y$ and $x<y$.
   Since $x$ is a prefix of $y$, without loss of generality, assume that $x0$ is a prefix of $y$.
   As $y\in\ua x0\cap\overline{\bigcup\mathcal{F}}$ and $x0\in\mathcal{K}(\mathbb{M})$, we obtain $\ua x0\cap\bigcup\mathcal{F}\neq\emptyset$.
   Thus, there is an $F_1\in\mathcal{F}$ including $x0$.
   Likewise, $x\surd\in\ua x\surd\cap\overline{\bigcup\mathcal{F}}$ holds, there is an $F_2\in\mathcal{F}$ containing $x\surd$.
   By the directedness of $\mathcal{F}$, we have an $F\in\mathcal{F}$ with $F_1,F_2\subseteq F$.
   It follows $x0,x\surd\in F$; and hence, $F$ is not $\surd$-max, a contradiction.
\end{proof}

\begin{lemma}(\cite[\tr{Lemma 2.8}]{Xu2021})\label{ImageOfIrr}
  For every continuous map $f:X\ra Y$ between the $T_0$ spaces, if $I$ is an irreducible subset of $X$, then $f(I)$ is an irreducible subset of $Y$.
\end{lemma}
\begin{corollary}\label{ProjectionOfIrrIsIrr}
  Let $X$ be a $T_0$ space and $\mathcal{F}$ an irreducible subset of $\mathcal{CV}^\surd X$.
   Then, $[\mathcal{F}]_1$ is an irreducible subset of $\mathcal{V}_{\leq 1}\Sigma\mathbb{M}$.
\end{corollary}

For every Scott closed subset $F\subseteq\mathbb{M}$, there is a useful map $p_F$ introduced in \cite[page 2]{Goubault2011} (\tr{the authors of \cite{Goubault2011} defined $p_F$ on $\mathbb{C}$; we adopt the same formula to extend this map to $\mathbb{M}$, and all properties established in \cite{Goubault2011} remain valid}).
 The map $p_F:\mathbb{M}\ra\mathbb{M}$ is given by $p_F(x)=\max(\da x\cap F)$, and it is Scott continuous.

\begin{remark}\label{UseP_F[]}
  For every $\mu\in\mathcal{V}_{\leq 1}\Sigma\mathbb{M}$, $U\in\mathcal{O}\Sigma\mathbb{M}$ and $F\in\Gamma\Sigma\mathbb{M}$,
  \tr{\begin{align*}
    p_F^{-1}(U) & = \{x:p_F(x)\in U\} \\
     & = \{x:\max(\da x\cap F)\in U\} \\
     & = \{x:\da x\cap U\cap F\neq\emptyset\} \\
     & = \ua(U\cap F),
  \end{align*}
  and }
   \begin{align*}
     p_F[\mu](U) & = \mu(p_F^{-1}(U)) \\
      & = \mu(\ua(U\cap F)) \\
      & = \mu(\ua(\ua (U\cap F)\cap \supp(\mu))) \\
      & = \mu(\ua(U\cap F\cap \supp(F))).
   \end{align*}
   Hence, \tr{by Proposition \ref{SuppDetermineValue}}, $\supp(p_F[\mu])=\supp(\mu)\cap F$, and $p_{F}[\mu]=\mu$ \tr{whenever $\supp(\mu)\subseteq F$}.
   If there is a $\upsilon\in\mathcal{V}_{\leq 1}\Sigma\mathbb{M}$ greater than $\mu$, then $\mu\leq p_{\supp(\mu)}[\upsilon]$ by the Scott continuity of $p_{\supp(\mu)}[-]$.
\end{remark}

%\begin{remark}\label{SuppP_F[]=SuppCapF}
%  Suppose $\mu\in\mathcal{V}_{\leq 1}\mathbb{M}$ and $F$ is a Scott closed subset of $\mathbb{M}$.
%   Then, the following equation holds:
%   \begin{equation*}
%     supp(p_F[\mu])=supp(\mu)\cap F.
%   \end{equation*}
%\end{remark}
%\begin{proof}
%  Directly, we have
%  \begin{align}
%    supp(p_F[\mu]) & = (\bigcup\{U\in\mathcal{O}\mathbb{M}:p_F[\mu](U)=0\})^c\nonumber \\
%     & = (\bigcup\{U\in\mathcal{O}\mathbb{M}:\mu(p_F^{-1}(U))=0\})^c\nonumber \\
%     & = (\bigcup\{U\in\mathcal{O}\mathbb{M}:\mu(\ua(U\cap F))=0\})^c \nonumber\\
%     & = (\bigcup\{U\in\mathcal{O}\mathbb{M}:\ua(U\cap F)\cap supp(\mu)=\emptyset\})^c\nonumber\\
%     & \supseteq (\bigcup\{U\in\mathcal{O}\mathbb{M}:U\cap F\cap supp(\mu)=\emptyset\})^c\nonumber\\
%     & = ((supp(\mu)\cap F)^c)^c\nonumber\\
%     & = supp(\mu)\cap F.\nonumber
%  \end{align}
%
% Conversely, the following calculation
% \begin{align}
%   p_F[\mu]((supp(\mu)\cap F)^c) & = \mu(p_F^{-1}((supp(\mu)\cap F)^c))\nonumber \\
%    & = \mu(\ua((supp(\mu)\cap F)^c\cap F))\nonumber \\
%    & = \mu(\ua((supp(\mu)^c\cap F)\cup(F^c\cap F)))\nonumber \\
%    & = \mu(\ua(supp(\mu)^c\cap F))\nonumber \\
%    & \leq \mu(supp(\mu)^c)\nonumber \\
%    & = 0 \nonumber
% \end{align}
% implies $supp(p_F[\mu])\subseteq supp(\mu)\cap F$.
%\end{proof}

\begin{lemma}\label{SuppSup=SupSupp}
  For an irreducible subset $\mathcal{F}$ of $\mathcal{V}_{\leq 1}\Sigma\mathbb{M}$, the following equation holds:
  \begin{equation*}
    \sup(\{\supp(\mu):\mu\in\mathcal{F}\})=\supp(\sup(\mathcal{F})).
  \end{equation*}
\end{lemma}
\begin{proof}
  %\tr{The sup of $\mathcal{F}$ exists because $\mathcal{V}_{\leq 1}\Sigma\mathbb{M}$ is sober.}
   Let $\mu,\upsilon\in\mathcal{F}$.
   Since $\mathcal{V}_{\leq 1}\Sigma\mathbb{M}$ is a bounded complete domain, it is sober.
   \tr{As every irreducible subset of a sober space has its sup by Remark \ref{SoberEquiva}, $\sup(\mathcal{F})$ exists.}
   Then, $\mu,\upsilon\leq \sup(\mathcal{F})$; hence, $\mu\vee\upsilon$ exists.
   We show
   \begin{equation*}
     \supp(\mu\vee\upsilon)=\supp(\mu)\cup \supp(\upsilon).
   \end{equation*}
   For every $\Sigma_{x\in F}r_x\delta_x\in\dda\mu\cap B_{\mathbb{M}},\Sigma_{y\in G}s_y\delta_y\in\dda\upsilon\cap B_{\mathbb{M}}$, likewise $\Sigma_{x\in F}r_x\delta_x\vee\Sigma_{y\in G}s_y\delta_y$ exists.
   Because the Scott continuous map $\supp:\mathcal{V}_{\leq 1}\Sigma\mathbb{M}\ra\Gamma\Sigma\mathbb{M}$ preserves the order,
   \begin{equation*}
     \da F\cup\da G=\supp(\Sigma_{x\in F}r_x\delta_x)\cup \supp(\Sigma_{y\in G}s_y\delta_y)\subseteq \supp(\Sigma_{x\in F}r_x\delta_x\vee\Sigma_{y\in G}s_y\delta_y).
   \end{equation*}
   Consider the valuation $p_{\da F\cup\da G}[\sup(\mathcal{F})]$, it is greater than $\Sigma_{x\in F}r_x\delta_x,\Sigma_{y\in G}s_y\delta_y$ by Remark \ref{UseP_F[]}.
   Again, by Remark \ref{UseP_F[]}, we obtain
   \begin{equation*}
     \da F\cup\da G\subseteq \supp(\Sigma_{x\in F}r_x\delta_x\vee\Sigma_{y\in G}s_y\delta_y)\subseteq \supp(p_{\da F\cup\da G}[\sup(\mathcal{F})])=\da F\cup\da G.
   \end{equation*}
   Thus, $\supp(\Sigma_{x\in F}r_x\delta_x\vee\Sigma_{y\in G}s_y\delta_y)=\da F\cup\da G$.
   It follows
   \begin{align*}
     \supp(\mu\vee\upsilon) & = \supp(\bigsqcup(\dda\mu\cap B_{\mathbb{M}})\vee\bigsqcup(\dda\upsilon\cap B_{\mathbb{M}})) \\
      & = \supp(\bigsqcup\{\Sigma_{x\in F}r_x\delta_x\vee\Sigma_{y\in G}s_y\delta_y:\Sigma_{x\in F}r_x\delta_x\in\dda\mu\cap B_{\mathbb{M}},\Sigma_{y\in G}s_y\delta_y\in\dda\upsilon\cap B_{\mathbb{M}}\}) \\
      & = \bigsqcup\{\supp(\Sigma_{x\in F}r_x\delta_x\vee\Sigma_{y\in G}s_y\delta_y):\Sigma_{x\in F}r_x\delta_x\in\dda\mu\cap B_{\mathbb{M}},\Sigma_{y\in G}s_y\delta_y\in\dda\upsilon\cap B_{\mathbb{M}}\} \\
      & = \bigsqcup\{\supp(\Sigma_{x\in F}r_x\delta_x)\cup \supp(\Sigma_{y\in G}s_y\delta_y):\Sigma_{x\in F}r_x\delta_x\in\dda\mu\cap B_{\mathbb{M}},\Sigma_{y\in G}s_y\delta_y\in\dda\upsilon\cap B_{\mathbb{M}}\} \\
      & = \bigsqcup\{\supp(\Sigma_{x\in F}r_x\delta_x):\Sigma_{x\in F}r_x\delta_x\in\dda\mu\cap B_{\mathbb{M}}\}\cup\bigsqcup\{\supp(\Sigma_{y\in G}s_y\delta_y):\Sigma_{y\in G}s_y\delta_y\in\dda\upsilon\cap B_{\mathbb{M}}\} \\
      & = \supp(\bigsqcup(\dda\mu\cap B_{\mathbb{M}}))\cup \supp(\bigsqcup(\dda\upsilon\cap B_{\mathbb{M}})) \\
      & = \supp(\mu)\cup \supp(\upsilon) .
   \end{align*}
   Finally, we have
   \begin{align}
     \supp(\sup(\mathcal{F})) & = \supp(\bigsqcup\{\mu_1\vee\cdots\vee\mu_n:\mu_1,\ldots,\mu_n\in\mathcal{F},n\in\mathbb{N}\}) \nonumber\\
      & = \bigsqcup\{\supp(\mu_1\vee\cdots\vee\mu_n):\mu_1,\ldots,\mu_n\in\mathcal{F},n\in\mathbb{N}\}\nonumber \\
      & = \bigsqcup\{\supp(\mu_1)\cup\cdots\cup \supp(\mu_n):\mu_1,\ldots,\mu_n\in\mathcal{F},n\in\mathbb{N}\}\nonumber \\
      & = \overline{\bigcup\{\supp(\mu):\mu\in\mathcal{F}\}}\nonumber\\
      & = \sup(\{\supp(\mu):\mu\in\mathcal{F}\})\nonumber.
   \end{align}
\end{proof}

\begin{lemma}\label{SurdMaxAreLower}
   All of $\surd$-max Scott closed subsets form a lower subfamily of $\Gamma\Sigma\mathbb{M}$, that is, if $F\in\Gamma\Sigma\mathbb{M}$ is $\surd$-max and $G\subseteq F$ is a Scott closed subset, then $G$ is $\surd$-max.
\end{lemma}

For each Scott closed subset $A$ of $\mathbb{M}$, we write $\Uparrow\!\! A$ for $\{F\in\Gamma\Sigma\mathbb{M}:A\subseteq F\}$.

\begin{lemma}\label{IrrSupIsSmax}
  Let $\mathcal{F}\subseteq\Sigma\Gamma\Sigma\mathbb{M}$ be an irreducible subfamily consisting of some $\surd$-max Scott closed subsets of $\mathbb{M}$
   Then, $\sup(\mathcal{F})$ is $\surd$-max.
\end{lemma}
\begin{proof}
  Suppose $F_1,F_2\in\mathcal{F}$.
   For every $F_1',F_2'$ of $\mathcal{K}(\Gamma\Sigma\mathbb{M})$ with $F_i'\subseteq F_i$ for $i=1,2$, by Lemma \ref{SurdMaxAreLower}, $F_1',F_2'$ are $\surd$-max.
   Since $\Uparrow\!\! F_1',\Uparrow\!\! F_2'$ are Scott open subset of $\Gamma\Sigma\mathbb{M}$ and $F_i\in\Uparrow\!\! F_i'\cap\mathcal{F}$ for $i=1,2$, there is a $F_3\in\Uparrow\!\! F_1'\cap\Uparrow\!\! F_2'\cap\mathcal{F}$.
   It follows $F_1'\cup F_2'\subseteq F_3$.
   Because $F_3\in\mathcal{F}$ is $\surd$-max, $F_1'\cup F_2'$ is $\surd$-max.
   Thus,
   \begin{equation*}
     F_1\cup F_2=\bigsqcup\{F_1'\cup F_2':F_1',F_2'\in\mathcal{K}(\Gamma\Sigma\mathbb{M}),F_1'\subseteq F_1, F_2'\subseteq F_2\}
   \end{equation*}
   is $\surd$-max by Fact \ref{FactOfGammaC*} and Lemma \ref{SupOfDirectedSurdMax}.
   We conclude that the union of finitely many Scott closed subsets of $\mathcal{F}$ is $\surd$-max.
   Moreover,
   \begin{equation*}
     \sup(\mathcal{F})=\bigsqcup\{F_1\cup\cdots\cup F_n:F_1,\ldots,F_n\in\mathcal{F}\}
   \end{equation*}
   is $\surd$-max.
\end{proof}

\begin{proposition}\label{CVXIsSober}
  If $X$ is a sober space, then so is $\mathcal{CV}^\surd X$.
\end{proposition}
  \tr{We prove the sobriety by means of Remark \ref{SoberEquiva}.}
\begin{proof}
  Let $\mathcal{F}$ be an irreducible subset of $\mathcal{CV}^\surd X$, and let $I_y$ denote the subset
  \begin{equation*}
    \{f(x):(\mu,f)\in\mathcal{F},x\in\da y\cap \supp(\mu)\cap\mathcal{K}(\mathbb{M})\}
  \end{equation*}
  for every $y\in \supp(\sup([\mathcal{F}]_1))$.

  Suppose that the open subsets $U_1,U_2\in\mathcal{O}X$ meet $I_y$.
   Then, for $i=1,2$, there are $(\mu_i,f_i)\in\mathcal{F}$ and $x_i\in\da y\cap \supp(\mu_i)\cap\mathcal{K}(\mathbb{M})$ such that $f_i(x_i)\in U_i\cap I_y$.
   Pick arbitrary non-negative real numbers $\varepsilon_1,\varepsilon_2$ respectively and strictly less than $\mu_1(\ua x_1),\mu_2(\ua x_2)$.
   It implies $(\mu_i,f_i)\in\langle\ua x_i,\mu_i(\ua x_i)-\varepsilon_i,U_i\rangle$ for $i=1,2$.
   Since $\mathcal{F}$ is irreducible, there is a
   \begin{equation*}
     (\upsilon,g)\in\mathcal{F}\cap\langle\ua x_1,\mu_1(\ua x_1)-\varepsilon_1,U_1\rangle\cap\langle\ua x_2,\mu_2(\ua x_2)-\varepsilon_2,U_2\rangle.
   \end{equation*}
   Therefore, $g(x_1)\in U_1,g(x_2)\in U_2$.
   Because $\da y$ is a chain, without loss of generality, we assume $x_1\leq x_2$.
   By the partial continuity of $g$, $g(x_1)$ is less than $g(x_2)$; then, $g(x_2)\in U_1$ holds.
   It is easy to verify $g(x_2)\in I_y$.
   Hence, $g(x_2)\in I_y\cap U_1\cap U_2$, $I_y$ is an irreducible subset of $X$.
   Thus, $\sup(I_y)$ exists and $\overline{I_y}$ is equal to $\da \sup(I_y)$.

  Define a partial map $k:\mathbb{M}\rightharpoonup X$ with $dom(k)=\supp(\sup([\mathcal{F}]_1))$ by $k(y)=\sup(I_y)$.
   For every $V\in\mathcal{O}X$, the inverse
   \begin{align}
     \tr{k^{-1}(V)} & = \{y:k(y)\in V\}\nonumber \\
      & = \{y:\sup(I_y)\in V\}\nonumber \\
      & = \{y:I_y\cap V\neq\emptyset\}\nonumber \\
      & = \{y:\exists(\mu,f)\in\mathcal{F},x\in\da y\cap \supp(\mu)\cap\mathcal{K}(\mathbb{M})\text{ s.t. }f(x)\in V\}\nonumber \\
      & = \bigcup\{\ua x\cap \supp(\sup([\mathcal{F}]_1)):\exists(\mu,f)\in\mathcal{F},x\in \supp(\mu)\cap\mathcal{K}(\mathbb{M})\text{ s.t. }f(x)\in V\} \nonumber\\
      & = \supp(\sup([\mathcal{F}]_1))\cap\bigcup\{\ua x:\exists(\mu,f)\in\mathcal{F},x\in \supp(\mu)\cap\mathcal{K}(\mathbb{M})\text{ s.t. }f(x)\in V\}\nonumber
   \end{align}
   is an relative Scott open subset of $\supp(\sup([\mathcal{F}]_1))$.
   Thus, $k$ is partially continuous.
   By Proposition \ref{ProjectionOfIrrIsIrr}, Lemmas \ref{SuppSup=SupSupp} and \ref{IrrSupIsSmax}, $\supp(\sup([\mathcal{F}]_1))$ is $\surd$-max.
   We conclude that $(\sup([\mathcal{F}]_1),k)$ is a $\surd$-max continuous random variable on $X$.

  Assume that $(\pi,w)\in\mathcal{CV}^\surd X$ is greater than each $(\mu,f)$ of $\mathcal{F}$.
   Obviously, $\pi\geq \sup([\mathcal{F}]_1)$.
   Thus, $\supp(\pi)$ contains $\supp(\sup([\mathcal{F}]_1))$.
   For each $y\in \supp(\sup([\mathcal{F}]_1))$, the chain $\da y\cap\mathcal{K}(\mathbb{M})$ converges to $y$.
   It follows that $w(y)=\bigsqcup w(\da y\cap\mathcal{K}(\mathbb{M}))$.
   For every $f(x)\in I_y$, $x$ is in $\da y\cap \supp(\mu)\cap\mathcal{K}(\mathbb{M})$ for some $(\mu,f)\in\mathcal{F}$.
   Because $(\mu,f)\leq(\pi,w)$, we obtain $f(x)\leq w(x)$.
   It implies $f(x)\leq w(x)\leq\bigsqcup w(\da y\cap\mathcal{K}(\mathbb{M}))=w(y)$.
   Hence, $k(y)=\sup(I_y)\leq w(y)$.

  It remains to show $\da(\sup([\mathcal{F}]_1),k)=\overline{\mathcal{F}}$.
   Since $\mathcal{F}$ is irreducible, we only need to verify that every member of a subbase including $(\sup([\mathcal{F}]_1),k)$ meets $\mathcal{F}$.
   Suppose that an open subset $\langle \ua z,r,W\rangle$ contains $(\sup([\mathcal{F}]_1),k)$.
   Then, the weak open subset $\langle \ua z>r\rangle$ contains  $\sup([\mathcal{F}]_1)$.
   Hence, there is a $(\upsilon_1,g_1)\in\mathcal{F}$ such that $\upsilon_1\in\langle \ua z>r\rangle$.
   Additionally, $W$ includes $k(z)=\sup(I_z)$, there is a $(\upsilon_2,g_2)\in\mathcal{F}$ with $g_2(z)\in W$.
   We can easily obtain $(\upsilon_1,g_1)\in\langle\ua z,r,X\rangle$ and $(\upsilon_2,g_2)\in\langle\ua z,\upsilon_2(\ua z)-\gamma,W\rangle$, where $\gamma$ is a non-negative real number strictly less than $\upsilon_2(\ua z)$.
   Thus, $\mathcal{F}\cap\langle\ua z,r,X\rangle\cap\langle\ua z,\upsilon_2(\ua z)-\gamma,W\rangle$ is not empty.
   Notice that $\langle\ua z,r,X\rangle\cap\langle\ua z,\upsilon_2(\ua z)-\gamma,W\rangle$ is a subset of $\langle\ua z,r,W\rangle$, we obtain $\mathcal{F}\cap\langle\ua z,r,W\rangle\neq\emptyset$.
\end{proof}
\begin{lemma}(\cite[Exercise O-5.16]{Gierz2003})
  An upper subspace of a sober space is sober.
\end{lemma}
\begin{corollary}
  \tr{  For every $T_0$ space $X$, $\mathcal{CV}^\surd_1X$ is an upper subspace of $\mathcal{CV}^\surd X$.
   Thus, if $X$ is sober, $\mathcal{CV}^\surd_1X$ is a sober space.}
\end{corollary}

\begin{example}\label{ExampleOfDirectedSubsetOfCV}
  \tr{Recall that $inc_{\mathbb{C}}:\mathbb{C}\subseteq\mathbb{M}$ is the inclusion.
  Let $x_n=0\overbrace{11\cdots1}^{n}0\surd$ and $\{\mu_n\}_{n\in\mathbb{N}}\subseteq^{\ua}\mathcal{V}_{\leq 1}\Sigma\mathbb{C}$ be an arbitrary ascending chain.
   Consider $\mu'_n:=\frac{1}{2^{n+2}}\delta_{x_n}+\frac{1}{2}\mathbf{1}[inc_{\mathbb{C}}[\mu_n]]$.
   Since $\delta_{x_n}$ and $\mathbf{1}[inc_{\mathbb{C}}[\mu_n]]$ concentrate on different branch, it is easy to see that
   $$\supp(\mu'_n)=\da\{x_j\}^n_{i=1}\cup \mathbf{1}(\supp(inc_{\mathbb{C}}[\mu_n]))$$
   is $\surd$-max.
   Thus $\{\mu'_n\}_{n\in\mathbb{N}}$ is a directed family that each element has a $\surd$-max support.
   Also let $f_n:\supp(\mu'_n)\ra\overline{\mathbb{N}}\times\mathbb{M}$ be the layered identity given by $f_n(x)=(n,x)$, where $\overline{\mathbb{N}}$ is natural number with the usually order and an extra top $\top$.
   Hence, $\{(\mu'_n,f_n)\}_{n\in\mathbb{N}}$ is directed in $\mathcal{CV}^\surd_1\Sigma(\overline{\mathbb{N}}\times\mathbb{M})$, and its sup is the pair consisting of $\Sigma_{n\in\mathbb{N}}\frac{1}{2^{n+2}}\delta_{x_n}+\frac{1}{2}\mathbf{1}[inc_{\mathbb{C}}[\bigsqcup_{n\in\mathbb{N}}\mu_n]]$ and the top identity $$\{x\mapsto(\top,x):x\in dom(\Sigma_{n\in\mathbb{N}}\frac{1}{2^{n+2}}\delta_{x_n}+\frac{1}{2}\mathbf{1}[inc_{\mathbb{C}}[\bigsqcup_{n\in\mathbb{N}}\mu_n]])\}.$$}
\end{example}

\begin{lemma}\label{d-continuous}
  \tr{Let $X$ be a d-space, $f:\mathbb{M}\rightharpoonup X$ a partial map whose $dom(f)$ is a Scott closed subset of $\mathbb{M}$.
   The partial map $f$ is partially continuous if it is Scott continuous from $dom(f)$ to $X$.}
\end{lemma}
\begin{proof}
  \tr{Notice that every directed subset of $dom(f)$ converges to its sup, if $f$ is partially continuous, then $f$ preserves the sups of directed subsets.
    Conversely, if $f$ preserves the sups of directed subsets of $dom(f)$, then for every open subset $V$ of $X$,
    \begin{align*}
      f^{-1}(V) & = \{x\in dom(f):f(x)\in V\} \\
       & = \{x\in dom(f):f(\bigsqcup\da x\cap\mathcal{K}(\mathbb{M}))\in V\}\\
       & = \{x\in dom(f):\bigsqcup f(\da x\cap\mathcal{K}(\mathbb{M}))\in V\}\\
       & = \bigcup\{\ua y\cap dom(f):y\in \mathcal{K}(\mathbb{M}),f(y)\in V\},
    \end{align*}
  which is Scott open in $dom(f)$.}
\end{proof}

The following result is similar to Proposition \ref{CVXIsSober}.
\begin{remark}\label{UpperSubspaceOfD-space}
  Every upper subspace of a d-space is a d-space.
\end{remark}
\begin{proposition}\label{CVXCV1XAreDspace}
  If $X$ is a d-space, then so are $\mathcal{CV}^\surd X$ as well as $\mathcal{CV}^\surd_1 X$.
\end{proposition}
\begin{proof}(Sketch)
  Let $\mathcal{F}\subseteq^{\ua}\mathcal{CV}^\surd X$ be a directed subset.
   We can easily verify that $I_y=\{f(x):(\mu,f)\in\mathcal{F},x\in\da y\cap \supp(\mu)\cap\mathcal{K}(\mathbb{M})\}$ is a directed subset of $X$ for every $y\in \supp(\sup([\mathcal{F}]_1))$, and that $(\sup([\mathcal{F}]_1),k)$ as we defined in Proposition \ref{CVXIsSober} is the sup of $\mathcal{F}$ by Lemma \ref{d-continuous}.
   We prove that every $\langle\ua z,r,V\rangle$ is Scott open.
   Suppose that $(\sup([\mathcal{F}]_1),k)\in\langle\ua z,r,V\rangle$.
   Since $\langle\ua z>r\rangle,V$ are Scott-open, there is a $(\mu_1,f_1)\in\mathcal{F}$ such that $\mu_1\in\langle\ua z>r\rangle$ and a $(\mu_2,f_2)\in\mathcal{F}$ such that $f_2(z)\in I_z\cap V$.
   Hence, by the directedness of $\mathcal{F}$, there exists a $(\upsilon,g)\in\mathcal{F}$ greater than $(\mu_1,f_1),(\mu_2.f_2)$.
   Obviously, $\upsilon(\ua z)\geq\mu_1(\ua z)>r$ and $g(z)\geq f_2(z)\in V$ hold.
   Thus, $(\upsilon,g)$ is included in $\langle\ua z,r,V\rangle$.
   We conclude that $\mathcal{CV}^\surd X$ is a d-space.
%   Notice that $\mathcal{CV}^\surd_1X$ is an upper subspace of $\mathcal{CV}^\surd X$, which is closed under the sups of directed subsets, $\mathcal{CV}^\surd_1X$ is a subdcpo.
   Notice that $\mathcal{CV}^\surd_1 X$ is an upper subspace of $\mathcal{CV}^\surd X$, $\mathcal{CV}^\surd_1X$ is a d-space \tr{by Remark \ref{UpperSubspaceOfD-space}.}
\end{proof}

\subsection{\tr{Sobrification and D-completion}}

\tr{Next, we demonstrate connections between $\mathcal{CV}^\surd X$ and $\mathcal{SV}^\surd X$.}

We denote $\lvert s\rvert$ the length of a string $s$.
For every natural number $n\in\mathbb{N}$, the notation $\mathbb{M}_n$ denotes the subposet of $\mathbb{M}$ that consists of all strings with the length less than $n$, that is, $\mathbb{M}_n:=\{x\in\mathbb{M}:\lvert x\rvert\leq n\}$.
 Obviously, $\mathbb{M}_n$ is a Scott closed subset.
 Hence, for every $T_0$ space $X$, we let $P^X_n:\mathcal{CV}^\surd_1X\ra\mathcal{SV}^\surd_1X$ be the map given by
 \begin{equation*}
   P^X_n(\mu,f)=(p_{\mathbb{M}_n}[\mu],f\mid_{\mathbb{M}_n}),
 \end{equation*}
 where $f\mid_{\mathbb{M}_n}$ is the partial map $f$ restricted to $\mathbb{M}_n$.
 Since $dom(f)=\supp(\mu)$, we obtain
 \begin{equation*}
   dom(f\mid_{\mathbb{M}_n})=\supp(\mu)\cap\mathbb{M}_n.
 \end{equation*}
 \tr{By Remark \ref{UseP_F[]}, $\supp(p_{\mathbb{M}_n}[\mu])=\supp(\mu)\cap\mathbb{M}_n$, and it is $\surd$-max by Lemma \ref{SurdMaxAreLower}.
 Thus $P^X_n(\mu,f)$ is well-defined.}

\begin{remark}
  \tr{If $s\in\mathbb{M}$ and $|s|>n$, then $p_{\mathbb{M}_n}(s)$ is the length-$n$ prefix of $s$; if $|s|\leq n$, then $p_{\mathbb{M}_n}(s)=s$.
   A simple example is $p_{\mathbb{M}_2}(1011\surd)=10$.}
\end{remark}
% Notice that $p_{\mathbb{M}_n}[\mu]=\mathbb{M}_n\cap supp(\mu)=p_{\mathbb{M}_n}(supp(\mu))$.

\begin{proposition}
  For every $T_0$ space $X$ and every $n\in\mathbb{N}$, the map $P^X_n$ is continuous.
\end{proposition}
\begin{proof}
  It is evident that
  \begin{equation*}
    (P^X_n)^{-1}(\langle\ua z,r,V\rangle'_1)=
     \begin{cases}
        \langle\ua z,r,V\rangle_1, & \mbox{if $\lvert z\rvert\geq n$};  \\
        \emptyset, & \mbox{otherwise}.
     \end{cases}
  \end{equation*}
\end{proof}

\tr{The following result is directly induced by properties (C), (E) of $p_F$ in \cite[page 3]{Goubault2011}.}
\begin{lemma}\label{SupPn=Id}
  $\{p_{\mathbb{M}_n}[-]\}_{n\in\mathbb{N}}$ is a chain of $[\mathcal{V}_{\leq 1}\Sigma\mathbb{M}\ra\mathcal{V}_{\leq 1}\Sigma\mathbb{M}]$, and for every $\mu\in\mathcal{V}_{\leq1}\Sigma\mathbb{M}$, $\bigsqcup_{n\in\mathbb{N}}p_{\mathbb{M}_n}[\mu]=\mu$.
\end{lemma}

From the preceding lemma, the following proposition is straightforward.
\begin{proposition}
  Let $X$ be a d-space, $\{P^X_n\}_{n\in\mathbb{N}}$ is a chain of $[\mathcal{CV}^\surd_1X\ra\mathcal{SV}^\surd_1X]$.
\end{proposition}

For every $T_0$ space $X$, \tr{we let $Inc_X$ denote the inclusion $\mathcal{SV}^\surd_1X \subseteq\mathcal{CV}^\surd_1X$, which obviously is continuous}.
%we denote $Ink_X$ the inclusion from $\mathcal{SV}^\surd_1X$ to $\mathcal{CV}^\surd_1X$, obviously, it is continuous.

\begin{proposition}\label{SupBigPn=id}
  If $X$ is a d-space, then $\bigsqcup_{n\in\mathbb{N}} Inc_X\circ P^X_n=id_{\mathcal{CV}^\surd_1X}$.
\end{proposition}
\begin{proof}
  For every $(\mu,f)\in\mathcal{CV}^\surd_1X$, $\bigsqcup_{n\in\mathbb{N}}p_{\mathbb{M}_n}[\mu]=\mu$ \tr{by Lemma \ref{SupPn=Id}}.
   If $y\in \supp(\mu)$ is an infinite string, then
   \begin{equation*}
     y=\bigsqcup\da y\cap\mathcal{K}(\mathbb{M})=\bigsqcup_{n\in\mathbb{N}} \max(\da y\cap\mathbb{M}_n).
   \end{equation*}
   Moreover,
   \begin{equation*}
     f(y)=\bigsqcup_{n\in\mathbb{N}}f(\max(\da y\cap\mathbb{M}_n))=\bigsqcup_{n\in\mathbb{N}}f\mid_{\mathbb{M}_n}(\max(\da y\cap\mathbb{M}_n)).
   \end{equation*}
   If $y$ is a finite string of $dom(f)$, then $f(y)=f\mid_{\mathbb{M}_{\lvert y\rvert}}(y)$.
\end{proof}
\tr{The above proposition demonstrates that every $(\mu,f)\in\mathcal{CV}^\surd_1X$ is approximated by a directed subset $\{P_n^X(\mu,f)\}_{n\in\mathbb{N}}$ of simple random variables.}

\begin{definition}
  Let $X$ be a space and $Y$ a sober space.
   We say that $Y$ is the sobrification of $X$ if there is \tr{a topological} embedding $\tr{e}:X\ra Y$, and for every continuous map $\tr{f:}X\ra Z$ from $X$ to a sober space $Z$, there is a unique continuous map $g:Y\ra Z$ such that $g\circ \tr{e}=f$.
\end{definition}

\begin{lemma}(\cite[\tr{Theorem 2.4 (3)}]{Heckmann1996})\label{CharaSoberification}
  A sober space $Y$ is the sobrification of a subspace $X$ iff for every $U\in\mathcal{O}Y$ and every $y\in U$, there is an $x\in X\cap U$ lower than $y$.
\end{lemma}

%\tr{The above lemma still holds when substituting $U\in\mathcal{B}$ for $U\in\mathcal{O}Y$, provided $\mathcal{B}$ is a base of $Y$.}

\begin{proposition}\label{CVXisSober}
  If $X$ is a sober space, then $\mathcal{CV}^\surd X$ is the sobrification of $\mathcal{SV}^\surd X$.
\end{proposition}
\begin{proof}
   \tr{Let $(\mu,f)$ belong to an $U\in\mathcal{O}\mathcal{CV}^\surd X$.
   Since $\mathcal{CV}^\surd X$ is a sober space by Proposition \ref{CVXIsSober}, $U$ is Scott open.
   By Proposition \ref{SupBigPn=id}, $(\mu,f)=\bigsqcup_{n\in\mathbb{N}} P_n^X(\mu,f)\in U$, hence there is a $P_n^X(\mu,f)\in U$.
   Thus, by Lemma \ref{CharaSoberification}, $\mathcal{CV}^\surd X$ is the sobrification of $\mathcal{SV}^\surd X$.}
  %\tr{Notice that all finite intersections of subbase members form a base, the family $\{\bigcap_{i\in I}\langle \ua z_i,r_i,V_i\rangle:I~\text{is finite}\}$ is a base of $\mathcal{CV}^\surd X$.}
%   Let $(\mu,f)$ \tr{belong to a finite intersection} $\bigcap_{i\in I}\langle \ua z_i,r_i,V_i\rangle$, $F=\da\{z_i\}_{i\in I}$ and $g$ be equal to $f$ restricted to $\da F\cap \tr{\supp}(\mu)$.
%   By Corollary \ref{ValuationsOnFiniteTree}, $(p_F[\mu],g)\in\bigcap_{i\in I}\langle \ua z_i,r_i,V_i\rangle$ is a $\surd$-max simple random variable on $X$.
%   \tr{Obviously, $(p_F[\mu],g)\leq(\mu,f)$, which completes this proof by Lemma \ref{CharaSoberification}.}
\end{proof}
\begin{corollary}
   Moreover, $\mathcal{CV}^\surd_1X$ is the sobrification of $\mathcal{SV}^\surd_1X$ \tr{whenever $X$ is a sober space}.
\end{corollary}
\begin{proof}
  Proof analogous to Proposition \ref{CVXisSober}.
\end{proof}

\tr{Sobrification is completing a space to be sober.
 Similarly, there is also the D-completion, which completes $T_0$ spaces to be d-spaces.}

\begin{definition}
  \tr{Let $X$ be a space and $Y$ a d-space.
   We say that $Y$ is the D-completion of $X$ if there is a topological embedding $e:X\ra Y$, and for every continuous map $f:X\ra Z$ from $X$ to a d-space $Z$, there is a unique continuous map $g:Y\ra Z$ such that $g\circ e=f$.}
\end{definition}

\begin{lemma}(\cite[Theorem 6.7]{Keimel2009})
  \tr{The D-completion of a subspace $X$ of a d-space $Y$ is the minimum subdcpo within $Y$ that contains $X$, endowed with the relative topology.}
\end{lemma}
\begin{proposition}
  \tr{If $X$ is a d-space, then $\mathcal{CV}^\surd X$ is the D-completion of $\mathcal{SV}^\surd X$.}
\end{proposition}
\begin{proof}
  \tr{By Proposition \ref{CVXCV1XAreDspace}, $\mathcal{CV}^\surd X$ is a d-space; and by Proposition \ref{SupBigPn=id}, the minimum subdcpo within $\mathcal{CV}^\surd X$ that includes $\mathcal{SV}^\surd X$ is $\mathcal{CV}^\surd X$ itself.}
\end{proof}
\begin{corollary}
  \tr{$\mathcal{CV}^\surd_1 X$ is the D-completion of $\mathcal{SV}^\surd_1X$ whenever $X$ is a d-space.}
\end{corollary}

\tr{In Domain theory, it is interesting to characterize spaces with their powerspaces, we provide one such characterization as below.}
\begin{proposition}
  For a $T_0$ space $X$, if $\mathcal{CV}^\surd_1X$ is sober, then so is $X$.
\end{proposition}
\begin{proof}
  Suppose that $I$ is an irreducible subset of $X$.
   Since $\eta_X$ is continuous, $\eta_X(I)$ is an irreducible subset of $\mathcal{SV}^\surd_1X$; it is also an irreducible subset of $\mathcal{CV}^\surd_1X$.
   Then, $\overline{\eta_X(I)}=\da(\mu,f)$ for some $(\mu,f)\in\mathcal{CV}^\surd_1X$.
   Hence, $(\mu,f)$ is the sup of $\eta_X(I)$.
   If $\mu\neq\delta_\surd$, then $(\delta_\surd,\{\surd\mapsto f(\surd),\epsilon\mapsto f(\epsilon)\})$ is an upper bound of $\eta_X(I)$ which is lower than $(\mu,f)$.
   Thus, $\mu=\delta_\surd$.
   Consider the element $f(\surd)\in X$.
   Obviously, $f(\surd)$ is an upper bound of $I$.
   For every $U\in\mathcal{O}X$ with $f(\surd)\in U$, $\langle\ua\surd,0,U\rangle_1$ is an open subset that contains $(\delta_\surd,f)$.
   So $\langle\ua\surd,0,U\rangle_1\cap\eta_X(I)\neq\emptyset$, which implies $U\cap I\neq\emptyset$.
   We conclude $\overline{I}=\da f(\surd)$.
\end{proof}

Combining the above propositions, we obtain the following results.
\begin{theorem}\label{CV1SurdOverSober}
  Let $X$ be a $T_0$ space.
   Then the following are equivalent:
  \begin{center}
    $(i)$ $X$ is sober;\quad $(ii)$ $\mathcal{CV}^\surd X$ is sober;\quad $(iii)$ $\mathcal{CV}^\surd_1X$ is sober.
  \end{center}
\end{theorem}

\subsection{\tr{Over bounded complete domains}}

\begin{proposition}\label{CVSurdoverBCD}
  \tr{Let $L$ be a bounded complete domain.
   Then $\mathcal{CV}^\surd\Sigma L$ is a bounded complete domain, and its topology agrees with the Scott topology.}
\end{proposition}
\begin{proof}
  \tr{By \cite[Theorem 5]{Mislove2017}, Mislove showed that $CRV(L)$ is a bounded complete domain with a basis $\{(\Sigma_{x\in F}r_x\delta_x,f):\Sigma_{x\in F}r_x\delta_x\in B_{\mathbb{M}}\}$.
  It is self-evident that $\mathcal{CV}^\surd\Sigma L$ equipped with its specialization order is a Scott closed subset of $CRV(L)$.
   Thus, $\mathcal{CV}^\surd\Sigma L$ is a bounded complete domain with respect to its specialization order.
   Obviously, $\mathcal{SV}^\surd\Sigma L=\mathcal{CV}^\surd\Sigma L\cap\{(\Sigma_{x\in F}r_x\delta_x,f):\Sigma_{x\in F}r_x\delta_x\in B_{\mathbb{M}}\}$, $\mathcal{SV}^\surd\Sigma L$ is a basis of $\mathcal{CV}^\surd\Sigma L$.
   When $\mathcal{CV}^\surd\Sigma L$ equips its specialization order, for every $(\Sigma_{x\in F}r_x\delta_x,f)\in\mathcal{SV}^\surd\Sigma L$, the equation
   \begin{equation*}
     \dua(\Sigma_{y\in G}s_y\delta_y,g)=\bigcap_{z\in G}\langle\ua z,\Sigma_{y\in G}s_y\delta_y(\ua z),\dua g(z)\rangle
   \end{equation*}
   holds, which implies that the topology of $\mathcal{CV}^\surd\Sigma L$ agrees with the Scott topology.}
\end{proof}

\begin{theorem}\label{CV1surdoverBCD}
  \tr{For every bounded complete domain $L$, $\mathcal{CV}^\surd_1\Sigma L$ is a bounded complete domain, and its topology agrees with the Scott topology.}
\end{theorem}
\begin{proof}
  \tr{Define $\phi:\mathcal{CV}^\surd\Sigma L\ra\mathcal{CV}^\surd_1\Sigma L$ by $\phi(\mu,f)=(\mu+(1-\mu(\mathbb{M}))\delta_\epsilon,f)$, and it is continuous by
   \begin{align*}
     \phi^{-1}(\langle\ua z,r,V\rangle_1) & = \{(\mu,f)\in\mathcal{CV}^\surd\Sigma L:(\mu+(1-\mu(\mathbb{M}))\delta_\epsilon,f)\in\langle\ua z,r,V\rangle\} \\
      & = \{(\mu,f)\in\mathcal{CV}^\surd\Sigma L:\mu(\ua z)+(1-\mu(\mathbb{M}))\delta_\epsilon(\ua z)>r,f(z)\in V\} \\
      & = \begin{cases}
            \{(\mu,f)\in\mathcal{CV}^\surd\Sigma L:\mu(\ua z)>r,f(z)\in V\}, & \mbox{if $z\neq\epsilon$}; \\
            \{(\mu,f)\in\mathcal{CV}^\surd\Sigma L:f(z)\in V\}, & \mbox{if $z=\epsilon$}
          \end{cases}\\
      & = \begin{cases}
            \langle\ua z,r,V\rangle, & \mbox{if $z\neq\epsilon$}; \\
            \langle\ua\epsilon,0,V\rangle, & \mbox{if $z=\epsilon$}.
          \end{cases}
   \end{align*}
   It is also easy to observe $\phi\circ inc=id_{\mathcal{CV}^\surd_1\Sigma L}$, where $inc$ is the inclusion $\mathcal{CV}^\surd_1\Sigma L\subseteq\mathcal{CV}^\surd\Sigma L$.
   So $\mathcal{CV}^\surd_1\Sigma L$ is a topological retract of $\mathcal{CV}^\surd\Sigma L$.
   By Proposition III-3.11 and Lemma III-3.2 of \cite{Gierz2003}, bounded complete domains equipped with their Scott topologies exactly are densely injective $T_0$ spaces, and every topological retract of a densely injective $T_0$ space is densely injective.
   Thus, combining Proposition \ref{CVSurdoverBCD} and the above analysis, $\mathcal{CV}^\surd_1\Sigma L$ is a bounded complete domain and its topology coincides with the Scott topology.}
\end{proof}

\section{A monad structure of \tr{some} continuous random variables}\label{SecMonad2}

\begin{definition}
  A $T_0$ space $X$ is a weak d-space if every open subset of $X$ is Scott open.
\end{definition}
 Obviously, every d-space is a weak d-space.

\begin{fact}
  If $D$ is a directed subset of a weak d-space $X$ such that $\bigsqcup D$ exists, then $D$ converges to $\bigsqcup D$.
   Hence, we see that every continuous map between weak d-space preserves the existed sups of directed subsets.
\end{fact}

\begin{proposition}
  Let $X$ be a d-space.
   Then $\mathcal{SV}^\surd_1X$ is a lower subspace of $\mathcal{CV}^\surd_1X$ which is closed under the sups of bounded directed subsets.
\end{proposition}
\begin{proof}
  If $(\mu,f)\leq(\upsilon,g)$ and $(\upsilon,g)\in\mathcal{SV}^\surd_1X$, then $\supp(\mu)\subseteq \supp(\upsilon)$ is finite.
   It implies that $\mu$ is a finite valuation.
   By Proposition \ref{FiniteAreSimple}, $\mu$ is a simple valuation and $\mu\in B_{\mathbb{M}}$.

  Suppose that a directed subset $D$ of $\mathcal{SV}^\surd_1X$ has an upper bound in $\mathcal{SV}^\surd_1X$.
   Then $\bigsqcup_{d\in D}[d]_1$ is a finitely valued valuation.
   Thus $\bigsqcup D\in\mathcal{SV}^\surd_1X$.
\end{proof}

\begin{proposition}
  Every lower subspace of a d-space is a weak d-space.
\end{proposition}
\begin{proof}
  Let $A$ be a lower subspace of a d-space $X$ and $D$ a directed subset of $A$.
   Assume that $x$ is the sup of $D$ in $A$, $y$ is the sup of $D$ in $X$.
   It is apparent that $y\leq x$.
   Since $A$ is a lower subspace, $y\in A$.
   Hence, $y$ is the least upper bound of $D$ in $A$.
   It follows that $x=y$.
   We conclude that, for every directed subset $D'$ of $A$, if $D'$ has a sup in $A$, then the sup in $A$ coincides with the sup in $X$.
   With the above analysis, clearly, every open subset of $A$ is Scott open.
\end{proof}
\begin{corollary}
  For every d-space $X$, $\mathcal{SV}^\surd_1X$ is a weak d-space.
\end{corollary}

We denote by $[X\ra Y]$ the set of all continuous maps between the spaces ordered by $f\leq g$ if $f(x)\leq g(x)$ for each $x\in X$.

\begin{proposition}\label{TheDirectSupOfMaps}
  Let $X,Y$ be weak d-space, $\{f_i\}_{i\in I}$ a directed subset of $[X\ra Y]$.
   If $\bigsqcup_{i\in I}f(x)$ exists for each $x\in X$ , then the sup of $\{f_i\}_{i\in I}$ exists.
\end{proposition}
\begin{proof}
  Define $\bigsqcup_{i\in I}f_i$ by $(\bigsqcup_{i\in I}f_i)(x)=\bigsqcup_{i\in I}f_i(x)$.
   We only need to show the continuity of $\bigsqcup_{i\in I}f_i$.
   For every open subset $V\in\mathcal{O}Y$,
   \begin{equation*}
     (\bigsqcup_{i\in I}f_i)^{-1}(V)=\{x\in X:\bigsqcup_{i\in I}f_i(x)\in V\}=\bigcup_{i\in I}f_i^{-1}(V).
   \end{equation*}
\end{proof}

\begin{corollary}\label{CompositionIsScottContinuous}
  Let $X,Y,Z$ be weak d-spaces.
   The composition operation
   \begin{equation*}
     \circ:[X\ra Y]\times [Y\ra Z]\ra [X\ra Z],(f,g)\mapsto g\circ f
   \end{equation*}
    preserves the existed sups of directed subsets, that is, if $\{f_i\}_{i\in I}\subseteq^\ua[X\ra Y],\{g_j\}_{j\in J}\subseteq^\ua[Y\ra Z]$ have sup, then
    \begin{equation*}
      \bigsqcup_{j\in J}g_j\circ\bigsqcup_{i\in I}f_i=\bigsqcup_{i\in I,j\in J}(g_j\circ f_i)=\bigsqcup_{i\in I}\bigsqcup_{j\in J}(g_j\circ f_i)=\bigsqcup_{j\in J}\bigsqcup_{i\in I}(g_j\circ f_i).
    \end{equation*}
\end{corollary}

\begin{proposition}\label{DaggerPreserveDirectedSup}
  Let $X,Y$ be d-space and $\{h_i\}_{i\in I}$ a bounded directed subset of $[X\ra\mathcal{SV}^\surd_1Y]$.
   Then $(\bigsqcup_{i\in I} h_i)^\dagger=\bigsqcup_{i\in I}h_i^\dagger$.
\end{proposition}
\begin{proof}
  By the equivalence of Kleisli triples and monads, we have $h^\dagger=id_Y^\dagger\circ\mathcal{SV}^\surd_1(h)$ for every $h\in[X\ra\mathcal{SV}^\surd_1Y]$.
   Obviously, for every $(\Sigma_{x\in F}r_x\delta_x,f)\in\mathcal{SV}^\surd_1X$,
   \begin{align*}
     \bigsqcup_{i\in I}\mathcal{SV}^\surd_1(h_i)(\Sigma_{x\in F}r_x\delta_x,f) & = \bigsqcup_{i\in I}(\Sigma_{x\in F}r_x\delta_x,h_i\circ f) \tag{\tr{by Lemma \ref{SVApplyToMaps}}}\\
      & = (\Sigma_{x\in F}r_x\delta_x,\bigsqcup_{i\in I} h_i\circ f)\\
      & = \mathcal{SV}^\surd_1(\bigsqcup_{i\in I} h_i)(\Sigma_{x\in F}r_x\delta_x,f).
   \end{align*}
   Hence, $id_Y^\dagger\circ\mathcal{SV}^\surd_1(\bigsqcup_{i\in I}h_i)  = id_Y^\dagger\circ\bigsqcup_{i\in I}\mathcal{SV}^\surd_1(h_i)
       = \bigsqcup_{i\in I} id_Y^\dagger\circ\mathcal{SV}^\surd_1(h_i).$
\end{proof}

%%%%%%%%%%%%%%%%%%%%%%%%%%%%%%%%%%

Let $X,Y$ be d-spaces and $h:X\ra\mathcal{CV}^\surd$ a continuous map.
 Then, $Inc_Y\circ(P^Y_n\circ h)^\dagger\circ P^X_m$ is a continuous map from $\mathcal{CV}^\surd_1X$ to $\mathcal{CV}^\surd_1Y$ for every $m,n\in\mathbb{N}$.
 \begin{lemma}
   $\{Inc_Y\circ(P^Y_n\circ h)^\dagger\circ P^X_m:m,n\in\mathbb{N}\}$ is a directed subset of $[\mathcal{CV}^\surd_1X\ra \mathcal{CV}^\surd_1Y]$.
 \end{lemma}
 \begin{proof}
   By Proposition \ref{DaggerPreserveDirectedSup}, if $n,n'$ are natural numbers with $n\leq n'$, then $(P^Y_n\circ h)^\dagger\leq (P^Y_{n'}\circ h)^\dagger$.
    The rest of this proof is trivial.
 \end{proof}
 We construct a map $h^\ddagger:\mathcal{CV}^\surd_1X\ra\mathcal{CV}^\surd_1 Y$ by $h^\ddagger=\bigsqcup_{m,n\in\mathbb{N}}Inc_Y\circ(P^Y_n\circ h)^\dagger\circ P^X_m$.
  By Proposition \ref{TheDirectSupOfMaps}, $h^\ddagger$ is well-defined and continuous.
  \tr{Readers can consult Proposition \ref{DDaggerLookedLike} for the explicit form of $h^\ddagger$.}

We need some lemmas to \tr{show} the monad structure of continuous random variables.
 \begin{lemma}\label{EscapeBeforeDagger}
   Let $X,Y$ be d-spaces, $h:X\ra\mathcal{CV}^\surd_1Y$ a continuous map and $(\mu,f)\in\mathcal{CV}^\surd_1X$ a simple random variable on $X$.
    Then,
    \begin{equation*}
    (\bigsqcup_{m,n\in\mathbb{N}}Inc_Y\circ(P^Y_n\circ h)^\dagger)\circ P^X_m)(\mu,f)=(\bigsqcup_{n\in\mathbb{N}}Inc_Y\circ(P^Y_n\circ h)^\dagger)(\mu,f).
    \end{equation*}
 \end{lemma}
 \begin{proof}
   Obviously, there is an $m_0\in\mathbb{N}$ such that for every $m_1\geq m_0$, $P^X_{m_1}(\mu,f)=(\mu,f)$.
 \end{proof}
 \begin{lemma}\label{EscapeFromDagger}
   Let $X,Y$ be d-spaces and $h:X\ra\mathcal{SV}^\surd_1Y$ a continuous map.
    Then
    \begin{equation*}
    \bigsqcup_{m,n\in\mathbb{N}}Inc_Y\circ(P^Y_n\circ Inc_Y\circ h)^\dagger\circ P^X_m=\bigsqcup_{m\in\mathbb{N}}Inc_Y\circ h^\dagger\circ P^X_m.
    \end{equation*}
 \end{lemma}
 \begin{proof}
   Suppose $(\Sigma_{x\in F}r_x\delta_x,f)\in\mathcal{SV}^\surd_1X$.
    Obviously, $h(f(\da F))\subseteq\mathcal{SV}^\surd_1Y$ is a finite subset.
    Let
    \begin{equation*}
      n_0=\max(\{\lvert y\rvert:y\in \supp([h(f(x))]_1),x\in\da F\}),
    \end{equation*}
    and $n_1$ be a natural number greater than $n_0$.
    For every $x\in\da F$, $(P^Y_{n_1}\circ Inc_Y\circ h)(f(x))=h(f(x))$.
    Thus,
    \begin{equation*}
      (P^Y_{n_1}\circ Inc_Y\circ h)^\dagger(\Sigma_{x\in F}r_x\delta_x,f)=h^\dagger(\Sigma_{x\in F}r_x\delta_x,f).
    \end{equation*}
    Hence, for every $(\mu,f)\in\mathcal{CV}^\surd_1X$ and every natural number $m_0$, there is a $n_0'$ such that for every $n_1'\geq n_0'$,
    \begin{equation*}
      (P^Y_{n_1'}\circ Inc_Y\circ h)^\dagger(P^X_{m_0}(\mu,f))=h^\dagger(P^X_{m_0}(\mu,f)).
    \end{equation*}
    It follows that $\{Inc_Y\circ h^\dagger\circ P^X_m(\mu,f):m\in\mathbb{N}\}$ is a co-final subset of $\{Inc_Y\circ(P^Y_n\circ Inc_Y\circ h)^\dagger\circ P^X_m(\mu,f):m,n\in\mathbb{N}\}$.
    Therefore, their supremums are coincident.
 \end{proof}
Let $\mathbf{DS}$ be the category of d-spaces and continuous maps.
\begin{theorem}
  $(\mathcal{CV}^\surd_1,Inc\circ\eta,\ddagger)$ is a Kleisli triple over $\mathbf{DS}$.
\end{theorem}
\begin{proof}
  Set $X,Y,Z$ be d-spaces and $h:X\ra\mathcal{CV}^\surd_1Y,l:Y\ra\mathcal{CV}^\surd_1Z$ be continuous maps.

  (\romannumeral1) Straightforwardly, we have
  \begin{align*}
    (Inc_X\circ\eta_X)^\ddagger & = \bigsqcup_{m,n\in\mathbb{N}}Inc_X\circ(P^X_n\circ Inc_X\circ\eta_X)^\dagger\circ P^X_m \\
     & = \bigsqcup_{m\in\mathbb{N}}Inc_X\circ\eta_X^\dagger\circ P^X_m \tag{by Lemma \ref{EscapeFromDagger}}\\
     & = \bigsqcup_{m\in\mathbb{N}}Inc_X\circ id_{\mathcal{SV}^\surd_X}\circ P^X_m \\
     & = \bigsqcup_{m\in\mathbb{N}}Inc_X\circ P^X_m \\
     & = id_{\mathcal{CV}^\surd_X}.
  \end{align*}

  (\romannumeral2) If $x\in X$, then
  \begin{align*}
    h^\ddagger\circ Inc_X\circ\eta_X(x) & = (\bigsqcup_{m,n\in\mathbb{N}}Inc_Y\circ(P^Y_n\circ h)^\dagger\circ P^X_m)\circ Inc_X\circ\eta_X(x) \\
     & = \bigsqcup_{m,n\in\mathbb{N}}Inc_Y\circ(P^Y_n\circ h)^\dagger\circ P^X_m(\delta_\surd ,\chi_x)\\
     & = \bigsqcup_{n\in\mathbb{N}}Inc_Y\circ(P^Y_n\circ h)^\dagger(\delta_\surd ,\chi_x) \\
     & = \bigsqcup_{n\in\mathbb{N}}Inc_Y\circ P^Y_n\circ h(x) \\
     & = h(x).
  \end{align*}

  (\romannumeral3) If $(\mu,f)\in\mathcal{CV}^\surd_1X$, then,
  \begin{align*}
    l^\ddagger\circ h^\ddagger(\mu,f) & = (\bigsqcup_{m',n'\in\mathbb{N}}Inc_Z\circ(P^Z_{n'}\circ l)^\dagger\circ P^Y_{m'})\circ(\bigsqcup_{m,n\in\mathbb{N}}Inc_Y\circ(P^Y_n\circ h)^\dagger\circ P^X_m)(\mu,f) \\
     & = (\bigsqcup_{m,n\in\mathbb{N}}\bigsqcup_{m',n'\in\mathbb{N}}Inc_Z\circ(P^Z_{n'}\circ l)^\dagger\circ P^Y_{m'}\circ Inc_Y\circ(P^Y_n\circ h)^\dagger\circ P^X_m)(\mu,f) \tag{by Corollary \ref{CompositionIsScottContinuous}}\\
     & = \bigsqcup_{m,n\in\mathbb{N}}\bigsqcup_{m',n'\in\mathbb{N}}Inc_Z\circ(P^Z_{n'}\circ l)^\dagger\circ P^Y_{m'}((P^Y_n\circ h)^\dagger(P^X_m(\mu,f))) \\
     & = \bigsqcup_{m,n\in\mathbb{N}}\bigsqcup_{n'\in\mathbb{N}}Inc_Z\circ(P^Z_{n'}\circ l)^\dagger((P^Y_n\circ h)^\dagger(P^X_m(\mu,f))) \tag{by Lemma \ref{EscapeBeforeDagger}}\\
     & = (\bigsqcup_{m,n\in\mathbb{N}}\bigsqcup_{n'\in\mathbb{N}}Inc_Z\circ(P^Z_{n'}\circ l)^\dagger\circ(P^Y_n\circ h)^\dagger\circ P^X_m)(\mu,f),
  \end{align*}
  and
%  \begin{flalign*}
%    (l^\ddagger\circ h)^\ddagger & = \bigsqcup_{m',n'\in\mathbb{N}}Inc_Z\circ(P^Z_{n'}\circ l^\ddagger\circ h)^\dagger\circ P^X_{m'} &\\
%     & = \bigsqcup_{m',n'\in\mathbb{N}}Inc_Z\circ(P^Z_{n'}\circ (\bigsqcup_{m,n\in\mathbb{N}}Inc_Z\circ(P^Z_{n}\circ l)^\dagger\circ P^Y_{m})\circ h)^\dagger\circ P^X_{m'} &\\
%     & = \bigsqcup_{m,n\in\mathbb{N}}\bigsqcup_{m',n'\in\mathbb{N}}Inc_Z\circ(P^Z_{n'}\circ Inc_Z\circ(P^Z_{n}\circ l)^\dagger\circ P^Y_{m}\circ h)^\dagger\circ P^X_{m'} &\text{by the continuity of $Inc_Z,P^Z_{n'}$} \\
%     & = \bigsqcup_{m,n\in\mathbb{N}}\bigsqcup_{m',n'\in\mathbb{N}}Inc_Z\circ(P^Z_{n'}\circ Inc_Z\circ(P^Z_{n}\circ l)^\dagger\circ P^Y_{m}\circ h)^\dagger\circ P^X_{m'} & \\
%     & = \bigsqcup_{m,n\in\mathbb{N}}\bigsqcup_{m'\in\mathbb{N}}Inc_Z\circ((P^Z_{n}\circ l)^\dagger\circ P^Y_{m}\circ h)^\dagger\circ P^X_{m'} &\text{by Lemma \ref{EscapeFromDagger}}\\
%     & = \bigsqcup_{m,n\in\mathbb{N}}\bigsqcup_{m'\in\mathbb{N}}Inc_Z\circ(P^Z_{n}\circ l)^\dagger\circ(P^Y_{m}\circ h)^\dagger\circ P^X_{m'}.&
%  \end{flalign*}

  \begin{align*}
  \allowdisplaybreaks[2]
    (l^\ddagger\circ h)^\ddagger & = \bigsqcup_{m',n'\in\mathbb{N}}Inc_Z\circ(P^Z_{n'}\circ l^\ddagger\circ h)^\dagger\circ P^X_{m'} \\
     & = \bigsqcup_{m',n'\in\mathbb{N}}Inc_Z\circ(P^Z_{n'}\circ (\bigsqcup_{m,n\in\mathbb{N}}Inc_Z\circ(P^Z_{n}\circ l)^\dagger\circ P^Y_{m})\circ h)^\dagger\circ P^X_{m'} \\
     & = \tr{\bigsqcup_{m',n'\in\mathbb{N}}Inc_Z\circ(\bigsqcup_{m,n\in\mathbb{N}}P^Z_{n'}\circ Inc_Z\circ(P^Z_{n}\circ l)^\dagger\circ P^Y_{m}\circ h)^\dagger\circ P^X_{m'}} \tag{\tr{by the continuity of $P^Z_{n'}$ and Proposition \ref{TheDirectSupOfMaps}}} \\
     & = \tr{\bigsqcup_{m',n'\in\mathbb{N}}\bigsqcup_{m,n\in\mathbb{N}}Inc_Z\circ(P^Z_{n'}\circ Inc_Z\circ(P^Z_{n}\circ l)^\dagger\circ P^Y_{m}\circ h)^\dagger\circ P^X_{m'}} \tag{\tr{by Proposition \ref{DaggerPreserveDirectedSup} and the continuity of $Inc_Z$}} \\
     & = \tr{\bigsqcup_{m,n\in\mathbb{N}}\bigsqcup_{m',n'\in\mathbb{N}}Inc_Z\circ(P^Z_{n'}\circ Inc_Z\circ(P^Z_{n}\circ l)^\dagger\circ P^Y_{m}\circ h)^\dagger\circ P^X_{m'}} \displaybreak[3]\\
     & = \bigsqcup_{m,n\in\mathbb{N}}\bigsqcup_{m'\in\mathbb{N}}Inc_Z\circ((P^Z_{n}\circ l)^\dagger\circ P^Y_{m}\circ h)^\dagger\circ P^X_{m'} \tag{by Lemma \ref{EscapeFromDagger}}\\
     & = \bigsqcup_{m,n\in\mathbb{N}}\bigsqcup_{m'\in\mathbb{N}}Inc_Z\circ(P^Z_{n}\circ l)^\dagger\circ(P^Y_{m}\circ h)^\dagger\circ P^X_{m'}.
  \end{align*}
\end{proof}

\begin{corollary}
  \tr{$\mathbf{SOB}$ is closed under $\mathcal{CV}^\surd_1$ by Theorem \ref{CV1SurdOverSober}, where $\mathbf{SOB}$ is the category of sober spaces and continuous maps.}
\end{corollary}
\tr{
Let $\mathbf{DCPO}$ be the category of dcpos and Scott continuous maps.
 Then, $\Sigma:\mathbf{DCPO}\ra\mathbf{DS}$ is a functor, where $\Sigma(h:L\ra M)=h$.
 Define $\Omega:\mathbf{DS}\ra\mathbf{DCPO}$ by setting $\Omega X$ to be $X$ endowed with specialization order and $\Omega(h:X\ra Y)=h$.
 It is obvious to see that $\Omega$ is a functor from the definition of d-spaces.}

\begin{corollary}
  \tr{$(\Omega\circ\mathcal{CV}^\surd_1\circ\Sigma,\Omega(Inc_\Sigma\circ\eta_\Sigma),\ddagger)$ is a Kleisli triple over $\mathbf{DCPO}$, and $\mathbf{BCD}$ is closed under this triple, where $\mathbf{BCD}$ is the category of bounded complete domains and Scott continuous maps.}
\end{corollary}
\begin{proof}
   \tr{Since functors $\Omega$ and $\Sigma$ do not alter the underlying elements of its object, all required equations are same as those proven before.
    By Theorem \ref{CV1surdoverBCD}, $\Omega\circ\mathcal{CV}^\surd_1\circ\Sigma$ preserves bounded complete domains.}
\end{proof}

\tr{Next we provide the explicit form of $h^\ddagger$.}

\begin{remark}\label{Supp(+)=CupSupp}
  \tr{If $\mu,\upsilon\in\mathcal{V}_{\leq 1}\Sigma\mathbb{M}$, then
  \begin{align*}
    \supp(\mu+\upsilon) & = \mathbb{M}\setminus\bigcup\{U:\mu(U)+\upsilon(U)=0\} \\
     & = \mathbb{M}\setminus\bigcup\{U:\mu(U)=0,\upsilon(U)=0\} \\
     & = \mathbb{M}\setminus(\bigcup\{U:\mu(U)=0\}\cap\bigcup\{V:\upsilon(V)=0\}) \\
     & = \supp(\mu)\cup\supp(\upsilon).
  \end{align*}}
\end{remark}

%\begin{lemma}
%  For every $\mu\in\mathcal{V}_{\leq 1}\Sigma\mathbb{M}$ with a $\surd$-max support, we have $\supp(\mu|^\mathbb{C})=\supp(\mu)\cap\mathbb{C}$.
%\end{lemma}
%\begin{proof}
%   We obtain $\supp(\mu|^\mathbb{C})\subseteq\supp(\mu)\cap\mathbb{C}$ by
%   \begin{align*}
%     \mu|^\mathbb{C}(\mathbb{M}\setminus(\supp(\mu)\cap\mathbb{C})) & = \mu|^\mathbb{C}((\mathbb{M}\setminus\supp(\mu)\cup(\mathbb{M}\setminus\mathbb{C})) \\
%      & = \mu|^\mathbb{C}(\mathbb{M}\setminus\supp(\mu))+\mu|^\mathbb{C}(\mathbb{M}\setminus\mathbb{C})-\mu|^\mathbb{C}((\mathbb{M}\setminus\supp(\mu)\cap(\mathbb{M}\setminus\mathbb{C}))\\
%      & = \mu|^\mathbb{C}(\mathbb{M}\setminus\supp(\mu))+\mu|^\mathbb{C}(\mathbb{C}^\surd)-\mu|^\mathbb{C}((\mathbb{M}\setminus\supp(\mu)\cap(\mathbb{C}^\surd))\\
%      & = \mu|^\mathbb{C}(\mathbb{M}\setminus\supp(\mu))+0-0 \\
%      & \leq  \mu(\mathbb{M}\setminus\supp(\mu)) \\
%      & = 0.
%   \end{align*}
%
%   For the converse, suppose $\supp(\mu|^\mathbb{C})\subset\supp(\mu)\cap\mathbb{C}$.
%    Then $\mu|^\mathbb{C}(\mathbb{M}\setminus(\supp(\mu)\cap\mathbb{C}))\neq 0$.
%\end{proof}

\begin{lemma}\label{MathbfXandProjections}
  \tr{For every $\mu\in\mathcal{V}_1\Sigma\mathbb{M}$ and $x\in\mathcal{K}(\mathbb{M})$, if $|x|\geq n$, then $p_{\mathbb{M}_n}\circ\mathbf{x}$ is a constant map with image $x$, hence $p_{\mathbb{M}_n}[\mathbf{x}[\mu]]=\delta_x$; if $|x|<n$, then $p_{\mathbb{M}_n}[\mathbf{x}[\mu]]=\mathbf{x}[p_{\mathbb{M}_{n-|x|}}[\mu]]$.}
\end{lemma}
\begin{proof}
  \tr{For every Scott open subset $U$,
  \begin{align*}
    \mathbf{x}^{-1}(p_{\mathbb{M}_n}^{-1}(U)) & = \mathbf{x}^{-1}(\ua(U\cap\mathbb{M}_n)) \\
     & = \{y:xy\in\ua(U\cap\mathbb{M}_n)\};
  \end{align*}
  \begin{align*}
   p_{\mathbb{M}_{n-|x|}}^{-1}(\mathbf{x}^{-1}(U)) & = p_{\mathbb{M}_{n-|x|}}^{-1}(\{y:xy\in U\})\\
     & = \ua(\{y:xy\in U\}\cap\mathbb{M}_{n-|x|})\\
     & = \ua\{y:|y|\leq n-|x|,xy\in U\}\\
     & = \ua\{y:xy\in U\cap\mathbb{M}_n\}
  \end{align*}
  For every $y$ satisfying $xy\in\ua(U\cap\mathbb{M}_n)$, since $|x|\leq n$, there is a prefix $z$ of $y$ such that $xz\in U\cap\mathbb{M}_n$; hence $y\in p_{\mathbb{M}_{n-|x|}}^{-1}(\mathbf{x}^{-1}(U))$.
  Conversely, if $y\in p_{\mathbb{M}_{n-|x|}}^{-1}(\mathbf{x}^{-1}(U))$, then $y$ is greater than some $z$ with $xz\in U\cap\mathbb{M}_n$; thus $xy\in\ua(U\cap\mathbb{M}_n)$.
  Finally,
  \begin{equation*}
    p_{\mathbb{M}_n}[\mathbf{x}[\mu]](U)=\mu(\mathbf{x}^{-1}(p_{\mathbb{M}_n}^{-1}(U)))=\mu(p_{\mathbb{M}_{n-|x|}}^{-1}(\mathbf{x}^{-1}(U)))=\mathbf{x}[p_{\mathbb{M}_{n-|x|}}[\mu]](U).
  \end{equation*}}
\end{proof}

\begin{proposition}\label{DDaggerLookedLike}
  \tr{Let $X,Y$ be d-spaces, $h:X\ra\mathcal{CV}^\surd_1Y$ a continuous map and $(\mu,f)\in\mathcal{CV}^\surd_1X$.
   Denote $h(z)$ by $(\upsilon_z,g_z)$.
   Define
  \begin{equation*}
    \mu'=\mu|^\mathbb{C}+\Sigma_{x'\surd\in\supp(\mu)}\mu(\{x'\surd\})\mathbf{x}'[\upsilon_{f(x'\surd)}],
  \end{equation*}
  and define $f':(\supp(\mu)\cap\mathbb{C})\cup\bigcup\{\mathbf{x'}(\supp(\upsilon_{f(x'\surd)})\setminus\{\epsilon\}):x'\surd\in\supp(\mu)\}\ra Y$ by
  \begin{equation*}
    f'(z)=\begin{cases}
                               g_{f(x)}(\epsilon), & \mbox{if}~z=x\in \supp(\mu)\cap\mathbb{C};\\
                               g_{f(x'\surd)}(y), & \mbox{if}~x'\surd\in\supp(\mu),~z=x'y\in\mathbf{x'}(\supp(\upsilon_{f(x'\surd)})).
                             \end{cases}
  \end{equation*}
  Then, $(\mu',f')=h^\ddagger(\mu,f)$.}
\end{proposition}
\begin{proof}
  \tr{We list some derivations obtained from properties of $p_{\mathbb{M}_n[-]}$ and Lemma \ref{MathbfXandProjections}.
  \begin{align*}
    p_{\mathbb{M}_n}[\mu'] & = p_{\mathbb{M}_n}[\mu|^\mathbb{C}+\Sigma_{x'\surd\in\supp(\mu)}\mu(\{x'\surd\})\mathbf{x}'[\upsilon_{f(x'\surd)}]] \\
     & = p_{\mathbb{M}_n}[\mu|^\mathbb{C}]+\Sigma_{x'\surd\in\supp(\mu)}\mu(\{x'\surd\})p_{\mathbb{M}_n}[\mathbf{x}'[\upsilon_{f(x'\surd)}]]\\
     & = p_{\mathbb{M}_n}[\mu|^\mathbb{C}]+\Sigma_{x'\surd\in\supp(\mu),|x'\surd|>n}\mu(\{x'\surd\})p_{\mathbb{M}_n}[\mathbf{x}'[\upsilon_{f(x'\surd)}])\\
     & + \Sigma_{x'\surd\in\supp(\mu),|x'\surd|\leq n}\mu(\{x'\surd\})p_{\mathbb{M}_n}[\mathbf{x}'[\upsilon_{f(x'\surd)}]]\\
     & = p_{\mathbb{M}_n}[\mu|^\mathbb{C}]+\Sigma_{x'\surd\in\supp(\mu),|x'\surd|>n}\mu(\{x'\surd\})\delta_{p_{\mathbb{M}_n(x'\surd)}}\\
     & + \Sigma_{x'\surd\in\supp(\mu),|x'\surd|\leq n}\mu(\{x'\surd\})p_{\mathbb{M}_n}[\mathbf{x}'[\upsilon_{f(x'\surd)}]]\\
     & = p_{\mathbb{M}_n}[\mu|^\mathbb{C}]+\Sigma_{x'\surd\in\supp(\mu),|x'\surd|>n}\mu(\{x'\surd\})\delta_{p_{\mathbb{M}_n(x'\surd)}}\\
     & + \Sigma_{x'\surd\in\supp(\mu),|x'\surd|\leq n}\mu(\{x'\surd\})\mathbf{x}'[p_{\mathbb{M}_{n-|x'|}}[\upsilon_{f(x'\surd)}]],
  \end{align*}
  \begin{align*}
    (P_m^Y\circ h)^\dagger\circ P_n^X(\mu,f) & = (P_m^Y\circ h)^\dagger(p_{\mathbb{M}_n}[\mu],f|_{\mathbb{M}_n}) \\
     & = (P_m^Y\circ h)^\dagger(p_{\mathbb{M}_n}[\mu|^\mathbb{C}]+p_{\mathbb{M}_n}[\Sigma_{x'\surd\in\supp(\mu)}\mu(\{x'\surd\})\delta_{x'\surd}],f|_{\mathbb{M}_n}) \\
     & = (P_m^Y\circ h)^\dagger(p_{\mathbb{M}_n}[\mu|^\mathbb{C}]+\Sigma_{x'\surd\in\supp(\mu),|x'\surd|>n}\mu(\{x'\surd\})\delta_{p_{\mathbb{M}_n(x'\surd)}}\\
     & + \Sigma_{x'\surd\in\supp(\mu),|x'\surd|\leq n}\mu(\{x'\surd\})\delta_{x'\surd},f|_{\mathbb{M}_n}),
  \end{align*}
  and
  \begin{align*}
    [(P_m^Y\circ h)^\dagger\circ P_n^X(\mu,f)]_1
     & = p_{\mathbb{M}_n}[\mu|^\mathbb{C}]+\Sigma_{x'\surd\in\supp(\mu),|x'\surd|>n}\mu(\{x'\surd\})\delta_{p_{\mathbb{M}_n(x'\surd)}}\\
     & + \Sigma_{x'\surd\in\supp(\mu),|x'\surd|\leq n}\mu(\{x'\surd\})\mathbf{x'}[[P^Y_m\circ h(f|_{\mathbb{M}_n}(x'\surd))]_1] \\
     & = p_{\mathbb{M}_n}[\mu|^\mathbb{C}]+\Sigma_{x'\surd\in\supp(\mu),|x'\surd|>n}\mu(\{x'\surd\})\delta_{p_{\mathbb{M}_n(x'\surd)}}\\
     & + \Sigma_{x'\surd\in\supp(\mu),|x'\surd|\leq n}\mu(\{x'\surd\})\mathbf{x'}[p_{\mathbb{M}_m}[\upsilon_{f(x'\surd)}]].
  \end{align*}}

  \tr{On the other side,
  \begin{align*}
    [(P_m^Y\circ h)^\dagger\circ P_n^X(\mu,f)]_2(z) & = \begin{cases}
                                                          [(P_m^Y\circ h)(f|_{\mathbb{M}_n}(x))]_2(\epsilon), & \text{if}~z=x\in\supp(\mu)\cap\mathbb{M}_n\cap\mathbb{C}; \\
                                                          [(P_m^Y\circ h)(f|_{\mathbb{M}_n}(x'\surd))]_2(y), & \text{if}~x'\surd\in\supp(\mu)\cap\mathbb{M}_n,\\
                                                          & ~~~z=x'y\in\mathbf{x'}(\supp(p_{\mathbb{M}_m}[\upsilon_{f(x'\surd)}]))
                                                        \end{cases} \\
                                                        & = \begin{cases}
                                                          g_{f(x)}(\epsilon), & \mbox{if}~z=x\in\supp(\mu)\cap\mathbb{M}_n\cap\mathbb{C};  \\
                                                          g_{f(x'\surd)}(y), & \mbox{if}~x'\surd\in\supp(\mu)\cap\mathbb{M}_n,~z=x'y\in\mathbf{x'}(\supp(\upsilon_{f(x'\surd)})\cap\mathbb{M}_m),
                                                        \end{cases}
  \end{align*}
  and %\text{\shortstack[l]{if $x'\surd\in\supp(\mu)\cap\mathbb{M}_n$,\\ $z=x'y\in\mathbf{x'}(\supp(p_{\mathbb{M}_m}[\upsilon_{f(x'\surd)}]))$}}
  \begin{align*}
    f'|_{\mathbb{M}_n}(z) & = \begin{cases}
                               g_{f(z)}(\epsilon), & \mbox{if}~z\in \supp(\mu)\cap\mathbb{M}_n\cap\mathbb{C};\\
                               g_{f(x'\surd)}(y), & \mbox{if}~x'\surd\in\supp(\mu)\cap\mathbb{M}_n,~z=x'y\in\mathbf{x'}(\supp(\upsilon_{f(x'\surd)}))\cap\mathbb{M}_n.
                           \end{cases}
  \end{align*}}

  \tr{Observe that for every $m,n\in\mathbb{N}$, $(P_m^Y\circ h)^\dagger\circ P_n^X(\mu,f)\leq(p_{\mathbb{M}_{m+n+1}}[\mu'],f'|_{\mathbb{M}_{m+n+1}})$, and for every $n'\in\mathbb{N}$, $(P_{n'}^Y\circ h)^\dagger\circ P_{n'}^X(\mu,f)\geq(p_{\mathbb{M}_{n'}}[\mu'],f'|_{\mathbb{M}_{n'}})$.
  Directed subsets $\{(P_m^Y\circ h)^\dagger\circ P_n^X(\mu,f)\}_{m,n\in\mathbb{N}}$ and $\{(p_{\mathbb{M}_n}[\mu'],f'|_{\mathbb{M}_n})\}_{n\in\mathbb{N}}$ are co-final with each others.
  Thus their supremums are same.}
\end{proof}

\begin{proposition}\label{CVasFunctor}
  \tr{$\mathcal{CV}^\surd_1(h)(\mu,f)=(\mu, h\circ f)$.}
\end{proposition}
\begin{proof}
  \tr{Directly,
  \begin{align*}
    \mathcal{CV}^\surd_1(h)(\mu,f) & = (Inc_Y\circ\eta_Y\circ h)^\ddagger(\mu,f) \\
     & = (\mu|^\mathbb{C}+\Sigma_{x'\surd\in\supp(\mu)}\mu(\{x'\surd\})\mathbf{x'}[[Inc_Y\circ\eta_Y\circ h(f(x'\surd))]_1],[\mathcal{CV}^\surd_1(h)(\mu,f)]_2) \\
     & = (\mu|^\mathbb{C}+\Sigma_{x'\surd\in\supp(\mu)}\mu(\{x'\surd\})\mathbf{x'}[\delta_\surd],[\mathcal{CV}^\surd_1(h)(\mu,f)]_2) \\
     & = (\mu|^\mathbb{C}+\Sigma_{x'\surd\in\supp(\mu)}\mu(\{x'\surd\})\delta_{x'\surd},[\mathcal{CV}^\surd_1(h)(\mu,f)]_2)\\
     & = (\mu,[\mathcal{CV}^\surd_1(h)(\mu,f)]_2)
  \end{align*}
  and
  \begin{align*}
    [\mathcal{CV}^\surd_1(h)(\mu,f)]_2(z) & = \begin{cases}
                                                [(\eta_Y\circ h)(f(x))]_2(\epsilon), & \mbox{if}~z=x\in\supp(\mu)\cap\mathbb{C}; \\
                                                [(\eta_Y\circ h)(f(x'\surd))]_2(y), & \mbox{if}~x'\surd\in\supp(\mu), z=x'y\in\mathbf{x'}(\supp([\eta_Y(h(f(x'\surd)))]_1))
                                              \end{cases} \\
     & = \begin{cases}
            [(\delta_\surd,\chi_{h(f(x))})]_2(\epsilon), & \mbox{if}~z=x\in\supp(\mu)\cap\mathbb{C}; \\
            [(\delta_\surd,\chi_{h(f(x'\surd))})]_2(y), & \mbox{if}~x'\surd\in\supp(\mu), z=x'y\in\mathbf{x'}(\supp(\delta_\surd))
         \end{cases} \\
     & = \begin{cases}
            \chi_{h(f(x))}(\epsilon), & \mbox{if}~z=x\in\supp(\mu)\cap\mathbb{C}; \\
            \chi_{h(f(x'\surd))}(y), & \mbox{if}~x'\surd\in\supp(\mu), z=x'y\in\{x'\surd\}
         \end{cases} \\
     & = \begin{cases}
            h(f(x)), & \mbox{if}~z=x\in\supp(\mu)\cap\mathbb{C}; \\
            h(f(x'\surd)), & \mbox{if}~z=x'\surd\in\supp(\mu)
         \end{cases} \\
     & = h\circ f.
  \end{align*}}
\end{proof}

\section{\tr{Strengths and commutativity}}\label{SecStrength}

\tr{
Moggi provided some results that allow us to directly obtain tensorial strengths of monads over certain special categories (see \cite[Proposition 3.4]{Moggi1991}).
 Since our target categories satisfies requirements of Moggi's results, we simply define commutative monads as follows:
\begin{definition}
  Let $\mathbf{C}$ be a subcategory of sets with point wise finite products (that is, underlying sets of products agree with products of underlying sets), $(\mathcal{T},e,\star)$ a Kleisli triple over $\mathbf{C}$.
   The left strength $s$ consists of the family of morphisms $s_{A,B}:A\times \mathcal{T}B\ra \mathcal{T}(A\times B)$ of $\mathbf{C}$, given by
   \begin{equation*}
     s_{A,B}(a,v)=\mathcal{T}(\{b\mapsto(a,b):b\in B\})(v).
   \end{equation*}
   Dually, the right strength $t$ consists of the family of morphisms $t_{A,B}:\mathcal{T}A\times B\ra\mathcal{T}(A\times B)$ of $\mathbf{C}$, defined by
   \begin{equation*}
     t_{A,B}(u,b)=\mathcal{T}(\{a\mapsto(a,b):a\in A\})(u).
   \end{equation*}
   We say that the Kleisli triple is commutative if $s_{A,B}^\star\circ t_{A,\mathcal{T}B}=t_{A,B}^\star\circ s_{\mathcal{T}A,B}$.
\end{definition}}

\tr{By Proposition \ref{CVasFunctor}, we can directly characterize the left strength $s\!s$ and the right strength $t\!t$ of $\mathcal{CV}^\surd_1$:}

\tr{$s\!s_{X,Y}:X\times\mathcal{CV}^\surd_1Y\ra\mathcal{CV}^\surd_1(X\times Y)$,
  \begin{align*}
    s\!s_{X,Y}(x,(\upsilon,g)) & = \mathcal{CV}^\surd_1(\{z\mapsto(x,z):z\in Y\})(\upsilon,g) \\
     & = (\upsilon,\{z\mapsto(x,z):z\in Y\}\circ g) \\
     & = (\upsilon,\{z\mapsto(x,g(z)):z\in\supp(\upsilon)\}),
  \end{align*}
  and $t\!t_{X,Y}:\mathcal{CV}^\surd_1X\times Y\ra\mathcal{CV}^\surd_1(X\times Y)$,
  \begin{align*}
    t\!t_{X,Y}((\mu,f),y) & = \mathcal{CV}^\surd_1(\{z\mapsto(z,y):z\in X\})(\mu,f) \\
     & = (\mu,\{z\mapsto(z,y):z\in X\}\circ f) \\
     & = (\mu,\{z\mapsto(f(z),y):z\in\supp(\mu)\}).
  \end{align*}
Hence,}

\tr{$s\!s_{X,Y}^\ddagger:\mathcal{CV}^\surd_1(X\times\mathcal{CV}^\surd_1Y)\ra\mathcal{CV}^\surd_1(X\times Y)$,
  \begin{align*}
    & [s\!s_{X,Y}^\ddagger(\mu,\{z\mapsto(\omega_z,(\upsilon_z,g_z)):z\in\supp(\mu)\})]_1 \\
     & = \mu|^{\mathbb{C}}+\Sigma_{x'\surd\in\supp(\mu)}\mu(\{x'\surd\})\mathbf{x'}[[s\!s_{X,Y}(\omega_{x'\surd},(\upsilon_{x'\surd},g_{x'\surd}))]_1] \\
     & = \mu|^{\mathbb{C}}+\Sigma_{x'\surd\in\supp(\mu)}\mu(\{x'\surd\})\mathbf{x'}[[(\upsilon_{x'\surd},\{z\mapsto(\omega_{x'\surd},g_{x'\surd}(z)):z\in\supp(\upsilon   _{x'\surd})\})]_1] \\
     & = \mu|^{\mathbb{C}}+\Sigma_{x'\surd\in\supp(\mu)}\mu(\{x'\surd\})\mathbf{x'}[\upsilon_{x'\surd}],
  \end{align*}
  and $t\!t_{X,Y}^\ddagger:\mathcal{CV}^\surd_1(\mathcal{CV}^\surd_1X\times Y)\ra\mathcal{CV}^\surd_1(X\times Y)$,
  \begin{align*}
    & [t\!t_{X,Y}^\ddagger(\mu,\{z\mapsto((\upsilon_z,g_z),\omega_z,):z\in\supp(\mu)\})]_1 \\
     & = \mu|^{\mathbb{C}}+\Sigma_{x'\surd\in\supp(\mu)}\mu(\{x'\surd\})\mathbf{x'}[[t\!t_{X,Y}((\upsilon_{x'\surd},g_{x'\surd}),\omega_{x'\surd})]_1] \\
     & = \mu|^{\mathbb{C}}+\Sigma_{x'\surd\in\supp(\mu)}\mu(\{x'\surd\})\mathbf{x'}[[(\upsilon_{x'\surd},\{z\mapsto(g_{x'\surd}(z),\omega_{x'\surd}):z\in\supp(\upsilon_{x'\surd})\})]_1] \\
     & = \mu|^{\mathbb{C}}+\Sigma_{x'\surd\in\supp(\mu)}\mu(\{x'\surd\})\mathbf{x'}[\upsilon_{x'\surd}].
  \end{align*}}

\tr{So, we derive $s\!s_{X,Y}^\ddagger\circ t\!t_{X,\mathcal{CV}^\surd_1Y}\neq t\!t_{X,Y}^\ddagger\circ s\!s_{\mathcal{CV}^\surd_1X,Y}$ by
\begin{align*}
  & [s\!s_{X,Y}^\ddagger\circ t\!t_{X,\mathcal{CV}^\surd_1Y}((\mu,f),(\upsilon,g))]_1  \\
  = & [s\!s_{X,Y}^\ddagger(\mu,\{z\mapsto(f(z),(\upsilon,g)):z\in\supp(\mu)\})]_1 \\
  = & \mu|^{\mathbb{C}}+\Sigma_{x'\surd\in\supp(\mu)}\mu(\{x'\surd\})\mathbf{x'}[\upsilon],
\end{align*}
and
\begin{align*}
  & [t\!t_{X,Y}^\ddagger\circ s\!s_{\mathcal{CV}^\surd_1X,Y}((\mu,f),(\upsilon,g))]_1  \\
  = & [t\!t_{X,Y}^\ddagger(\upsilon,\{z\mapsto((\mu,f),g(z)):z\in\supp(\upsilon)\})]_1 \\
  = & \upsilon|^{\mathbb{C}}+\Sigma_{y'\surd\in\supp(\upsilon)}\upsilon(\{y'\surd\})\mathbf{y'}[\mu].
\end{align*}}

\tr{The monad $(\mathcal{CV}^\surd_1,Inc\circ\eta,\ddagger)$ is not commutative, and it is similar to obtain that $(\mathcal{SV}^\surd_1,\eta,\dagger)$ is not commutative either.}

\section{\tr{Conclusion and future work}}\label{SecConclusion}
  \tr{In this paper, we develop monads over $\mathbf{TOP0},\mathbf{DS},\mathbf{SOB},\mathbf{DCPO},\mathbf{BCD}$ comprising random variables.
   The most valuable results are the monad structure $(\mathcal{S}p\circ\mathcal{CV}^\surd_1\circ\Sigma,\mathcal{S}p(Inc_\Sigma\circ\eta_\Sigma),\ddagger)$ over a Cartesian closed category $\mathbf{BCD}$, and every continuous random variable of $\mathcal{S}p\circ\mathcal{CV}^\surd_1\circ\Sigma L$ on a bounded complete domain $L$ is the sup of a directed family of simple random variables.
   We believe our results provide a solution to Mislove's question in \cite{Mislove2013}.}

  \tr{Several unexplored matters of this technical line are provided here.
   The most prominent issue is that we do not establish any semantics with our monad yet.
   Then we quest that does the distributive law of continuous random variable monads over non-determinism monads.
   And we are also interested in more topological properties on $\mathcal{CV}^\surd_1$, such as whether $\mathcal{CV}^\surd_1X$ is core-compact when $X$ is.
   We think that it is rely on the investigation of spaces of partially continuous map.
   This is important to develop semantics to monads since core-compact spaces form a cartesian closed category.
   In addition, we are curious about what a kind of spaces that its valuations are totally push forwards of continuous random variables.
   Mislove provided a primary answer to this problem, but we reckon that there are more spaces in this kind.}

\bibliographystyle{amsplain}
\bibliography{OMRV-Arxiv}
\end{document}